%% file: main.tex
\documentclass[letterpaper]{article} % DO NOT CHANGE THIS
\usepackage[preprint]{aaai2027} 
\usepackage[hyphens]{url}  % DO NOT CHANGE THIS
\usepackage{graphicx} % DO NOT CHANGE THIS
\usepackage{natbib}  % DO NOT CHANGE THIS AND DO NOT ADD ANY OPTIONS TO IT
\usepackage{caption} % DO NOT CHANGE THIS AND DO NOT ADD ANY OPTIONS TO IT

\usepackage{algorithm}
\usepackage{algorithmic}

\usepackage{newfloat}
\usepackage{listings}
\DeclareCaptionStyle{ruled}{labelfont=normalfont,labelsep=colon,strut=off} % DO NOT CHANGE THIS
\floatstyle{ruled}
\newfloat{listing}{tb}{lst}{}
\floatname{listing}{Listing}

\usepackage{amsmath,amssymb,amsfonts,amsthm,mathtools}

\usepackage{subcaption}

\usepackage{tikz}
\usetikzlibrary{arrows.meta,positioning,shapes.geometric,calc,fit}

\usepackage{booktabs}

\input{macros}

\newtheorem{definition}{Definition}

\newtheorem{theorem}{Theorem}
\newtheorem{lemma}{Lemma}
\newtheorem{proposition}{Proposition}
\newtheorem{corollary}{Corollary}

\title{Auditing Emergent LLM-Agent Collaboration through Cooperation–Obligation Coupling}

\author{
    Zuyuan Zhang\textsuperscript{\rm 1}\equalcontrib,
    Hanqing Yang\textsuperscript{\rm 2}\equalcontrib,
    Carlee Joe-Wong\textsuperscript{\rm 2},
    Tian Lan\textsuperscript{\rm 1}
}

\affiliations{
    \textsuperscript{\rm 1}Department of Electrical and Computer Engineering,
    The George Washington University\\
    \textsuperscript{\rm 2}Department of Electrical and Computer Engineering,
    Carnegie Mellon University\\
    zuyuan.zhang@gwu.edu,
    hanqing3@andrew.cmu.edu,
    cjoewong@andrew.cmu.edu,
    tlan@gwu.edu
}

\newcommand{\PDFmode}{0}

\ifcase\PDFmode
\or
  \usepackage[1-8,10000]{pagesel}
\or
  \usepackage[9-]{pagesel}
\fi
\begin{document}

\maketitle

\input{01abstract}

\input{02intro}

\input{04solution}

\input{05experiments}

\input{06conclusion}

\nocite{*}
\bibliography{ref}

\input{07appendix}

\end{document}

%% file: macros.tex
\newenvironment{tightenum}{\begin{enumerate}\setlength{\itemsep}{1pt}\setlength{\parskip}{0pt}\setlength{\parsep}{0pt}}{\end{enumerate}}

\newcommand{\WorkSet}{\mathcal{W}}
\newcommand{\snd}{\operatorname{snd}}
\newcommand{\astab}{\operatorname{astab}}
\newcommand{\cq}{\operatorname{cq}}
\newcommand{\Perf}{\operatorname{Perf}}

%% file: 01abstract.tex
\begin{abstract}
LLM-agent systems can solve complex tasks through dynamic self-organization and emergent cooperation.
Auditing this process is essential because plausible intermediate or final outputs can conceal incomplete or unsupported work and poorly allocated responsibility, ultimately compromising response quality. While existing approaches may record messages, tool calls, provenance, or task dependencies, an auditability gap exists as they do not jointly represent what work remains, who is responsible for it, and what evidence justifies each work-state transition. We address this auditability gap by proposing 
\emph{Integrated Cooperation-Obligation REpresentation} (iCORE). It creates a unified encoding $X=(G,Q,\Pi)$ integrating observable interactions as a cooperation graph $G$, evolving work and assignments as an obligation graph $Q$, and the audit map $\Pi$ linking them with verifiable properties and evidence. This iCORE representation enables the auditor to certify two complementary properties: {Work soundness}, where every active decision-relevant work assertion must have a finite justification through $G$ and $\Pi$; and {Agent-assignment stability}, which requires that no feasible alternative agent improve the declared contribution value for an evaluated obligation by more than $\epsilon$. We establish local-to-global soundness and assignment-regret guarantees and a performance bound under stated conditions. iCORE is an instrumentation layer over workflows. Numerical results show that 
the full coupled state exactly reconstructs soundness and assignment defects in two execution modes and that, relative to passive full-state observation, iCORE-Audit yields absolute trajectory-quality improvements of $11.5\%$ and $26.4\%$ in controlled and real-LLM execution, respectively, with corresponding absolute terminal-performance improvements of $15.1\%$ and $31.0\%$.
\end{abstract}

%% file: 02intro.tex
\section{Introduction}

Large language models are increasingly deployed as cooperative systems rather than isolated agents \cite{wang2024survey,guo2024large}. Conversation-oriented frameworks support dynamic role coordination and intermediate-work exchange \cite{wu2023autogen,li2023camel,yang2026dig}, while software-oriented systems organize agents around shared artifacts and staged responsibilities \cite{hong2024metagpt,qian2024chatdev}. Tool-using agents adapt later actions to external observations \cite{yao2022react,schick2023toolformer,shen2023hugginggpt}. Planning methods revise execution structure \cite{wang2023plan,yao2023tree}, and self-refinement methods revise outputs from intermediate feedback \cite{madaan2023self,shinn2023reflexion}. Generative-agent simulations also show how local interactions can produce collective behavior \cite{park2023generative}. These capabilities allow cooperation to \emph{emerge} without fixing every role or work transition in advance.

We focus on an auditability gap in emergent LLM-agent collaboration--where existing approaches may record messages, tool calls, provenance, or task dependencies, but they do not jointly represent what work remains, who is responsible for it, and what evidence justifies each work-state transition. 
A \emph{work obligation} is a basic unit of progress, whether introduced by the task or created as agents organize the work. We call a cooperation state \emph{auditable} when every active work assertion that can affect a later update, diagnostic, or terminal output can be traced to observable actions or artifacts, the operation that changed the work, and currently valid evidence, e.g., from environment observations. Auditability is essential for reasoning, certifying, and improving emergent collective behavior, as a plausible %final 
output can often conceal unsupported intermediate claims, incomplete work, poor task allocation, or failures that propagated through the workflow.

Existing approaches provide important but partial views. LLM-agent frameworks record 
workflow traces, but do not define the obligations or evidence governing work-state transitions~\cite{wu2023autogen,li2023camel,hong2024metagpt,qian2024chatdev}. Distributed-event and provenance models represent ordering, causality, and derivation, but not evolving task obligations and their settlement \cite{lamport2019time,belhajjame2013prov,groth2013prov}. Planning and workflow models represent decomposition and dependencies while often abstracting away the event-level support for committed updates \cite{erol1994htn,torreno2017cooperative,van2003workflow,van1998application,murata1989petri}. Runtime verification and LLM judges assess traces or outputs, but lack the event--work correspondence needed to locate unsupported obligations or justify capability-based reassignment \cite{havelund2001monitoring,leucker2009brief,zheng2023judging}. The limitation is representational: interaction without work semantics is ambiguous, whereas work state without observable provenance is unauditable.

We introduce the \emph{Integrated Cooperation-Obligation REpresentation} (iCORE), with a unified state representation $X=(G,Q,\Pi)$. The cooperation graph $G$ encodes observable agent execution and artifacts, the obligation graph $Q$ encodes evolving work and assignments, and the audit map $\Pi$ certifies the correspondence between them with verifiable properties and evidence for reasoning on collective emergent behavior. 
iCORE is an instrumentation layer over workflows. The application specifies the environment- or task-specific observable and validation interfaces in a declaration package $D$, while the host framework retains control of execution. 

This iCORE representation enables the auditor to certify two complementary properties: \emph{Work soundness}, where every active decision-relevant work assertion from $Q$ must have a finite justification through $G$ and $\Pi$; and \emph{Agent-assignment stability}, which requires that no feasible alternative agent improve the declared contribution value for an evaluated obligation by more than a threshold $\epsilon$. The weighted satisfaction rates of these properties define iCORE-state quality $\cq_{D,\epsilon}$. 
We use the iCORE framework to prove local-to-global soundness and assignment-regret guarantees. Specifically, legal local updates preserve global work soundness and that componentwise stability bounds aggregate reassignment gain. Under explicit decomposability, validator-adequacy, and value-calibration conditions, iCORE-state quality also lower-bounds expected task performance; otherwise, it remains a process diagnostic.

We operationalize iCORE through \emph{iCORE-Audit}, an event-triggered layer around the host policy. It audits each proposal's intended state effects and can diagnose unsupported work, beneficial reassignment, or stalled progress  before a terminal result is produced. Proposed repairs are themselves audited, to prevent the system from appearing to improve simply by deleting or excluding difficult work. 
Relative to passive full-state observation, iCORE-Audit yields absolute trajectory-quality (as quantified by the iCORE state quality) improvements of $11.5\%$ and $26.4\%$ in controlled and real-LLM execution, respectively, with corresponding absolute terminal-performance improvements of $15.1\%$ and $31.0\%$.

This paper makes three main \textbf{contributions}.
\begin{tightenum}
    \item We define \textit{auditability} for emergent LLM-agent cooperation and introduce iCORE to couple observable interaction with evolving work and responsibility.
    \item We \textit{formalize} work soundness, agent-assignment stability, their local-to-global guarantees, and the conditions relating iCORE-state quality to task performance.
    \item We develop \textit{iCORE-Audit} and evaluate representation sufficiency, matched quality improvement, task performance, and cost in controlled and real-LLM settings.
\end{tightenum}

%% file: 04solution.tex
\section{Cooperation-Obligation Representation}
\label{sec:core}

We study fully cooperative, non-strategic agents with heterogeneous capabilities that jointly produce solve a single application-assigned task. 
The application supplies a quantitative outcome metric $\Perf$, and iCORE audits the process that produced the result, aiming to detect and forestall cooperation failures like execution errors, unsupported updates, incompatible partial results, or poor work assignments.

Before execution, the application fixes a declaration package $D=(\mathcal O_D,\mathcal M_D,\Omega,V,\mathcal I_D)$, where $\mathcal O_D$ specifies observable action and artifact types and semantic fields and $\mathcal I_D$ collects the capability-value, work-weighting, local-view, equivalence, and terminal interfaces; $\mathcal M_D$, $\Omega$, and $V$ are the finite observable agent set, work-operation vocabulary, and validation specification (e.g., on work correctness). The \emph{Integrated Cooperation-Obligation REpresentation} (iCORE) represents emergent cooperation through observable interaction in $G$, evolving work and responsibility in $Q$, and their certified correspondence in $\Pi$, as we formally define below.

A \emph{work-level assertion} is any semantic fact encoded in $Q$, including for example a node, relation, work status, result signature, coverage or compatibility claim, assignment, or terminal marker. It is \emph{active} if currently asserted and not superseded, and \emph{decision-relevant} if changing it can affect a
declared update, an unresolved work obligation, a iCORE diagnostic, or a terminal predicate. A state is \emph{auditable} when every active, decision-relevant assertion traces through $\Pi$ to finitely many observable actions or artifacts in $G$, its declared work operation, and active evidence valid under $V$.
An asserted work update is \emph{justified} when it instantiates an operation declared by $D$, its recorded preconditions are current, and its certificates pass $V$. Auditability thus requires more than a message log: it checks each decision-relevant update, assignment, and terminal decision up to the observational equivalence declared by $D$. This equivalence may ignore fresh identifiers and inessential timestamps, but must preserve assignees, statuses, coverage, and terminal evidence.

\begin{definition}[iCORE state]
\label{def:core-state}
A \emph{iCORE state} is
\[
    X=(G,Q,\Pi),
\]
where $G$ is a cooperation graph, $Q$ is an obligation graph, and $\Pi$ is an audit map between them.
\end{definition}

\subsection{Cooperation Graph: Observable Agent Interactions}
\label{sec:cooperation-graph}

%The cooperation graph records observable actions, artifacts, and typed production, consumption, routing, or verification dependencies.

\begin{definition}[Cooperation graph]
\label{def:cooperation-graph}
A \emph{cooperation graph} is a finite typed directed acyclic graph
\[
    G=(V_G,E_G,s_G,t_G,\theta_G,\lambda_G,\tau_G,\mu_G),
\]
where $s_G,t_G:E_G\rightarrow V_G$
are the source and target maps, and $\theta_G$ and $\lambda_G$ assign application-declared types and observable labels to vertices and edges. The logical-time map
$\tau_G:V_G\rightarrow\mathbb{N}$
satisfies
$
    u\rightarrow_G v
    \implies
    \tau_G(u)<\tau_G(v).
$
The vertex types distinguish action nodes
$\operatorname{Act}(G)\subseteq V_G$
and artifact nodes
$\operatorname{Art}(G)\subseteq V_G.$
An \emph{artifact} is a recorded, referenceable output of execution---for example, a message, plan, code patch, tool result, partial answer, or verification report---that may be consumed or checked by later actions. The metadata map $\mu_G$ stores observable information such as the acting agent,
recipient, tool, verifier, capability evidence, or monitor decision.
\end{definition}

$G$ therefore records observable actions, artifacts, and typed production, consumption, routing, or verification dependencies. This directed acyclic graph (DAG) records interaction events over time rather than an agent interaction topology: repeated agent interactions create distinct nodes in $V_G$. Thus, $G$ identifies observable provenance but not the required work or its current assignee.

\subsection{Obligation Graph: Evolving Work}
\label{sec:obligation-graph}

A \emph{work obligation} is a unit of work that must be addressed as part of completing the assigned task. It may stem from the initial request or emerge as agents organize, assign, and perform work. $Q$ tracks relationships between work obligations.

\begin{definition}[Obligation graph]
\label{def:obligation-graph}
An \emph{obligation graph} is a finite typed directed graph
\[
Q=(V_Q,E_Q,s_Q,t_Q,\theta_Q,\lambda_Q,\rho_Q,
\operatorname{stat}_Q,\sigma_Q,\alpha_Q).
\]
Here $s_Q,t_Q:E_Q\to V_Q$ are the source and target maps;
$\theta_Q$ and $\lambda_Q$ assign declared types and observable labels;
$\rho_Q$ assigns declared semantic roles to edges;
$\operatorname{stat}_Q:V_Q\to\mathcal S_Q$ records one of the eight statuses
listed below for each work obligation $v\in V_Q$; $\sigma_Q$ stores observable result signatures used for
compatibility; and $\alpha_Q:V_Q\to\mathcal M_D\cup\{\bot\}$ stores the
current assignee of unresolved work or the last assignee of settled work,
with $\bot$ denoting no assignment.
\end{definition}

The typed relations represented in $E_Q$ encode dependency, refinement, support, aggregation, or contradiction. The work obligation statuses mean: \texttt{open}, identified but not enabled; \texttt{ready}, enabled; \texttt{blocked}, waiting on a recorded condition;
\texttt{closed}, execution complete; \texttt{submitted}, awaiting validation or aggregation; \texttt{accepted}, validated and eligible to discharge work; \texttt{rejected}, failed acceptance; and \texttt{invalid}, unusable under $V$. We define $\operatorname{Unres}(Q)$ to comprise \texttt{open},
\texttt{ready}, \texttt{blocked}, and \texttt{submitted} obligations; we refer to all
other obligations as \textit{settled work}.

A \emph{required root} is an obligation with no
predecessor under the obligation relation. To complete the task, it must be discharged unless
active valid evidence marks it as excluded, duplicated, irrelevant, or
superseded. Formally, the required root set
$\operatorname{ReqRoot}(Q)\subseteq
\{q\in V_Q:\nexists p\in V_Q\text{ with }p\prec_Q^{\mathrm{obl}}q\}$.
Let $\operatorname{EffRoot}(X)$ be the roots not so removed and
$\operatorname{Cov}_Q(q)$ the roots that $q$ may discharge. These semantics
allow work to emerge without a complete decomposition in advance;
Appendix~\ref{app:work-semantics} gives the full definitions.

\begin{definition}[Relative obligation and evaluation set]
\label{def:relative-obligation}
For a component $q$, let $\operatorname{Dis}_{-q}(X)$ be the effective roots
discharged by accepted work whose active support does not depend on $q$.
Its \emph{relative obligation} is
$
\operatorname{Ob}_X(q)=\operatorname{Cov}_Q(q)\cap
\bigl(\operatorname{EffRoot}(X)\setminus
\operatorname{Dis}_{-q}(X)\bigr).
$
The evaluation set $\WorkSet(X)$ contains every unresolved $q$ with
$\operatorname{Ob}_X(q)\neq\varnothing$ and every component carrying an
active decision-relevant assertion for a iCORE diagnostic or terminal
predicate.
\end{definition}

Thus, unresolved obligations remain evaluated while nonempty, and settled, excluded, or superseded components remain only while they affect a diagnostic or terminal decision. This preserves terminal support and prevents terminal iCORE-state scores from becoming vacuous. Unlike $G$, $Q$ represents work
semantics and assignment but does not identify the interactions or capability evidence justifying them.

\subsection{Audit Map: Connecting Interaction and Work}
\label{sec:audit-map}

The audit map connects the interactions recorded in $G$ with the work evolution recorded in $Q$ using the application-declared vocabulary $\Omega$ of operations that create, assign, reassign, transform, verify, or resolve work.

\begin{definition}[Audit map]
\label{def:audit-map}
Given $G$ and $Q$, an \emph{audit map} is
\[
    \Pi=(\pi,\omega,\kappa),
    \qquad
    \pi:V_G\rightharpoonup 2^{V_Q},
    \qquad
    \omega:\mathcal L_{\pi}\rightharpoonup\Omega,
\]
where $\mathcal L_{\pi}=\{(a,q)\in\operatorname{Act}(G)\times V_Q:q\in\pi(a)\}$, $\pi$ links observable actions and artifacts to the work components they affect, $\omega$ records the declared work operation for each linked action--component pair, and $\kappa$ stores the certificates supporting the correspondence and resulting work update.
\end{definition}

For $a\in\operatorname{Act}(G)$ and $q\in\pi(a)$, $\omega(a,q)$ records how an agent action $a$ changes or maintains the work obligation $q$, so one action may implement different declared operations on different components. A linked artifact $z\in\operatorname{Art}(G)$ may be a certificate premise but needs no operation label.

A certificate $c\in\kappa$ stores an identifier, internal and external premises, a claim, and an issuer; its claim must instantiate a schema declared by $D$. Internal premises refer to the current iCORE state, whereas external premises enter through declared environment-input nodes in $G$. The
certificate is valid when active, current, authorized in issuer and scope, and accepted by $V$; it remains active until validly revoked or superseded.
The complete record is given in Appendix~\ref{app:certificates}.

Certificate validity alone is insufficient: the certificate and its premises
must connect through $\Pi$ to the affected assertion in $Q$ and the responsible observable nodes in $G$. In particular, an assignment is auditable only when the responsible action, affected work and assignee, assignment operation, and capability evidence agree across $G$, $Q$, $\omega$, and
$\kappa$.

Together, $G$, $Q$, and $\Pi$ determine whether work is justified,
responsibility is suitably assigned, and both assessments evolve. Appendix~\ref{app:projection-insufficiency} shows why unlabeled interaction topology alone cannot determine these properties.

\section{iCORE for Work Soundness}
\label{sec:work-soundness}

iCORE uses \emph{work soundness} in the proof-theoretic sense: every active decision-relevant work-level assertion must have a finite valid derivation from recorded interaction, its declared operation, and active evidence under $V$. The preservation, reconstruction, and terminal-report results below are specification invariants induced by the update legality, recording, and terminal contracts; required-root completeness is enforced separately by Section~\ref{sec:terminal-soundness}, and semantic task correctness still depends on validator adequacy.

\subsection{Legal Cooperation Processes}
\label{sec:legal-process}

Using the declaration package from Section~\ref{sec:core}, a iCORE system is $\mathcal{S}=(D,\mathcal{X}_{D},\widehat{\mathcal{R}},X_0,\operatorname{Acc},\operatorname{Rej})$, where $\mathcal{X}_{D}$ is the feasible-state space, $\widehat{\mathcal{R}}$ is the family of candidate local-update schemas, $X_0\in\mathcal{X}_D$ is the initial state, and $\operatorname{Acc},\operatorname{Rej}:\mathcal{X}_D\to\{0,1\}$ are the terminal predicates. The relation $X\equiv_X Y$ identifies states differing only in fields declared inessential by $D$ while preserving, under transported typed-graph isomorphisms, every field affecting update enablement, obligations, assignments, evidence, or terminal decisions; it is a congruence for legal updates, as detailed in Appendix~\ref{app:local-updates}.

Membership in $\mathcal{X}_{D}$ constrains represented states, not the physical actions agents may attempt: an unsupported effect may be recorded in $G$ but cannot be committed in $Q$ except as an explicit blocked, rejected, invalid, superseded, or repair-pending record. Feasibility anchors work to observable input, keeps unresolved work accountable, supports assignments with observable assignment and capability evidence, and certifies aggregate coverage and compatibility; Appendix~\ref{app:feasibility} gives the complete conditions.

A candidate local update, induced by an observable agent action, tool result, validator decision, or monitor intervention, may change $G$, $Q$, and $\Pi$. Its finite declared local view contains the events read, their linked work, and the active evidence required by the operation; enablement uses neither future events nor unrecorded global information.

\begin{definition}[Legal iCORE update]
\label{def:legal-core-update}
A candidate local update $X\rightarrow X'$ is \emph{legal} if it (i) accesses only declared finite interfaces and depends only on its declared local view and $D$; (ii) applies an operation permitted by $D$; (iii) provides active $V$-valid evidence for every assertion it creates, activates, makes decision-relevant, or semantically modifies, and leaves no active assertion supported only by evidence that the update revokes or supersedes; and (iv) yields $X'\in\mathcal{X}_{D}$ while respecting $\equiv_X$.
\end{definition}

An \emph{admissible process} is any finite or infinite sequence $X_0\rightarrow X_1\rightarrow X_2\rightarrow\cdots$ of legal updates; different interaction orders, agent outputs, and monitor interventions may induce different admissible processes.

\subsection{Work Soundness}
\label{sec:work-soundness-property}

Using the correspondence established by $\Pi$, work soundness determines whether the evolving state of work in $Q$ is justified by observable interaction in $G$.

\begin{definition}[Work soundness]
\label{def:work-soundness}
For $q\in\WorkSet(X)$, let $\mathsf{A}_D(q;X)$ be its active decision-relevant work-level assertions and set $\mathsf{A}_D(X)=\bigcup_{q\in\WorkSet(X)}\mathsf{A}_D(q;X)$. A iCORE state $X$ is \emph{work-sound} when every $a\in\mathsf{A}_D(X)$ has a finite justification through observable initial-input or interaction nodes in $G$, the declared operation for each non-initial assertion, the corresponding links in $\Pi$, and evidence active and valid under $V$. A process is work-sound when every state it reaches is work-sound.
\end{definition}

Soundness therefore applies to every decision-relevant creation, assignment, transformation, resolution, and terminal-support assertion, not only the final report. For the evaluation set $\WorkSet(X)$ in Definition~\ref{def:relative-obligation}, $D$ predeclares positive weights $w_q$ based on declared information such as work type or covered-root importance; uniform weights are allowed. Let $J_D(q;X)=1$ exactly when all assertions in $\mathsf{A}_D(q;X)$ are justified. Define
\[
\snd_D(X)=
\begin{cases}
1, & \WorkSet(X)=\varnothing,\\[2mm]
\displaystyle
\frac{\sum_{q\in\WorkSet(X)}w_qJ_D(q;X)}
{\sum_{q\in\WorkSet(X)}w_q},
& \text{otherwise}.
\end{cases}
\]
Then $X$ is work-sound if and only if $\snd_D(X)=1$.

\begin{proposition}[Local-to-global preservation of work soundness]
\label{prop:local-work-soundness}
If $X_0$ is work-sound, every state reachable by an admissible process is work-sound.
\end{proposition}
The proof is given in Appendix~\ref{app:work-soundness-proofs}; thus, legality checked from finite local views yields a global guarantee over all active decision-relevant assertions in each reachable state.

Work soundness also makes the recorded work process reconstructable.

\begin{proposition}[Audit reconstruction]
\label{prop:audit-reconstruction}
If every legal update and its effects on $G$, $Q$, and $\Pi$ are recorded, the ordered record reconstructs every iCORE state up to $\equiv_X$. For any finite accepted process, it also gives each component in the terminal cover a finite explanation from the initial input through linked interactions in $G$, work transformations in $Q$, and their correspondence in $\Pi$.
\end{proposition}

The proof is given in Appendix~\ref{app:audit-reconstruction-proof}.

\subsection{Sound Terminal Outcomes}
\label{sec:terminal-soundness}

For $A\subseteq\operatorname{Accepted}(Q)$ and $R\subseteq\operatorname{ReqRoot}(Q)$, $A$ \emph{covers} $R$ when $R\subseteq\bigcup_{a\in A}\operatorname{Cov}_Q(a)$, and is a \emph{terminal cover} when it also has the joint compatibility evidence required by $V$. The relevant root set $\operatorname{EffRoot}(X)$ omits roots excluded, duplicated, rendered irrelevant, or superseded by active valid evidence.

Let $\operatorname{EffUnres}(X)$ contain the non-excluded unresolved components whose coverage intersects $\operatorname{EffRoot}(X)$, and let $\operatorname{FailRoot}(X)$ contain effective roots with active valid failure evidence not neutralized by an active valid repair. Then $\operatorname{Acc}(X)$ requires feasibility, no effective unresolved work, accepted work forming a certified terminal cover of all effective roots, no failed effective root or accepted component with unresolved or invalid support, and no live action or artifact that can change any effective root's discharge, failure, exclusion, or compatibility. Conversely, $\operatorname{Rej}(X)$ requires feasibility, no effective unresolved work, a nonempty effective-root set entirely contained in $\operatorname{FailRoot}(X)$ with no effective root discharged, and no live action or artifact that can repair, discharge, or reopen an effective root. Appendix~\ref{app:terminal-predicates} gives the exact predicates. A finite process is \emph{terminally complete} when its final state satisfies $\operatorname{Acc}(X)$ or $\operatorname{Rej}(X)$, or active valid evidence certifies that it is externally stuck.

\begin{definition}[Terminal-report soundness]
\label{def:terminal-report-soundness}
A work-sound process is \emph{terminal-report sound} when every \textsc{Accept} report emitted at a state $X$ satisfies $\operatorname{Acc}(X)$ and every \textsc{Reject} report emitted at a state $X$ satisfies $\operatorname{Rej}(X)$.
\end{definition}

Work soundness justifies individual work updates, while terminal-report soundness additionally requires the reported outcome to account for all required work.

\begin{corollary}[Terminal-report soundness]
\label{cor:terminal-report-soundness}
If $X_0$ is work-sound, every committed update is legal, and \textsc{Accept} and \textsc{Reject} are emitted only when $\operatorname{Acc}(X)$ and $\operatorname{Rej}(X)$ hold, respectively, then the process is terminal-report sound.
\end{corollary}

\paragraph{When process guarantees imply task correctness.}
These process guarantees imply semantic task correctness only when $V$ is adequate for the declared task class. An emitted decision is \emph{semantically correct} when its \textsc{Accept} output or \textsc{Reject} report satisfies the class's declared acceptance or rejection semantics. A \emph{$V$-certified failed-root set} $F\subseteq\operatorname{EffRoot}(X)$ contains roots with active valid failure evidence not neutralized by active valid repair. We call $V$ \emph{task-adequate} when every $V$-certified accepted terminal cover yields a semantically correct accepted output, every $V$-certified failed-root set satisfying the declared rejection rule yields a semantically correct rejection, and every process-relevant fault in coverage, compatibility, evidence, or terminal settlement is detectable under $V$.

\begin{proposition}[Task-level implication under validator adequacy]
\label{prop:task-level-implication}
If $V$ is task-adequate and a process is work-sound, terminal-report sound, and terminally complete, then every emitted \textsc{Accept} or \textsc{Reject} decision is semantically correct for the declared task class; an externally stuck report certifies only process status.
\end{proposition}

The proof is given in Appendix~\ref{app:task-implications}. For open-ended tasks whose quality is not fully captured by $V$, iCORE exposes unsupported process states but does not certify the remaining semantic quality of the agents' outputs.

\section{iCORE for Agent-Assignment Suitability}
\label{sec:assignment-stability}

Work soundness determines whether evaluated work is justified, not whether it is assigned suitably. We formalize \emph{agent-assignment stability} as the absence of a feasible alternative whose declared contribution advantage exceeds a fixed tolerance; here, ``stability'' means resistance to beneficial reassignment, not temporal invariance.

\subsection{Capability Values and Beneficial Reassignment}

For each $q\in\WorkSet(X)$, $D$ defines a feasible-agent set $\mathcal{F}_D(q,X)\subseteq\mathcal{M}_D$ using declared observable constraints such as tool access, modality, availability, authorization, or resource limits; feasibility does not predict success. The rule is evaluated from registered metadata and capability evidence in $G$. If no hard filter is available, it may conservatively set $\mathcal{F}_D(q,X)=\mathcal{M}_D$; by convention, a maximum over an empty feasible set is zero, indicating that no reassignment candidate is currently available.

The map $\alpha_Q(q)\in\mathcal{M}_D\cup\{\bot\}$ records the current assignee for unresolved work and the last assignee for settled work, using the statuses defined after Definition~\ref{def:obligation-graph}; $\bot$ denotes a blocked or unassigned component. Feasibility of $\mathcal{X}_D$ requires every non-$\bot$ assignee to belong to the feasible set in the applicable evaluation context.

To prevent retrospective evaluation, let $X_D^{\mathrm{eval}}(q)$ be $X$ for unresolved work and the reconstructible snapshot immediately before the assignment or reassignment establishing $\alpha_Q(q)$ for settled work, and abbreviate
$
\mathcal{F}_D(q,X):=\mathcal{F}_D\bigl(q,X_D^{\mathrm{eval}}(q)\bigr),
\qquad
v_D(m,q\mid X):=v_D\bigl(m,q\mid X_D^{\mathrm{eval}}(q)\bigr).
$
Thus, later availability changes and post-outcome oracle information cannot alter settled-work assessment.

For $m\in\mathcal{M}_D$, $D$ declares a normalized net contribution value $v_D(m,q\mid X)\in[0,1]$. It may be a scalarization, fixed before evaluation, of auditable criteria such as validated success estimates, latency, cost, or risk; it need not be an oracle estimate. If no defensible scalarization exists, report the criteria separately and do not invoke the scalar regret guarantee. Set $v_D(\bot,q\mid X)=0$. Every capability observation used by feasibility or contribution values must be recorded in $G$ and linked through $\Pi$ to the corresponding assignment in $Q$; unsupported self-reports and post-outcome oracle information are excluded.

For $m\in\mathcal{F}_D(q,X)$, define the counterfactual assignment gain
$
\Delta_D(m\leftarrow\alpha_Q(q);q,X)
=v_D(m,q\mid X)-v_D(\alpha_Q(q),q\mid X).
$
A gain above a declared tolerance $\epsilon\ge0$ is an actionable \emph{$\epsilon$-beneficial reassignment} for unresolved work. For settled work retained in $\WorkSet(X)$, it is a historical assignment-regret diagnostic computed in the assignment-time context and excludes later evidence.

\begin{definition}[Agent-assignment stability]
\label{def:assignment-stability}
A component $q\in\WorkSet(X)$ is \emph{$\epsilon$-stable} at $X$ when
$
    \max_{m\in\mathcal{F}_D(q,X)}
    \Delta_D(m\leftarrow\alpha_Q(q);q,X)
    \le\epsilon.
$
Using the same predeclared weights $w_q$ as work soundness, define
\[
\astab_{D,\epsilon}(X)=
\begin{cases}
1, & \\ \qquad\qquad \WorkSet(X)=\varnothing,\\[2mm]
\displaystyle
\frac{\sum_{q\in\WorkSet(X)}w_q
\mathbf{1}[q\text{ is $\epsilon$-stable at }X]}
{\sum_{q\in\WorkSet(X)}w_q},
& \\ \qquad\qquad \text{otherwise}.
\end{cases}
\]
The state is \emph{agent-assignment stable} when $\astab_{D,\epsilon}(X)=1$.
\end{definition}

The definition couples all three iCORE components: $Q$ supplies the component and assignee, $G$ the agents and observable capability evidence, and $\Pi$ the certified assignment correspondence.

\begin{theorem}[From Local Stability to Global Assignment-Regret Bound]
\label{thm:assignment-regret}
Assume $\WorkSet(X)\neq\varnothing$ and normalize its weights so that $\sum_{q\in\WorkSet(X)}w_q=1$. Let
$
    U_D(\alpha\mid X)=
    \sum_{q\in\WorkSet(X)}w_q
    v_D(\alpha(q),q\mid X).
$
For any componentwise feasible alternative assignment $\beta$, meaning $\beta(q)\in\mathcal{F}_D(q,X)\cup\{\bot\}$ for every $q$,
$
U_D(\beta\mid X)-U_D(\alpha_Q\mid X)
\le
1-(1-\epsilon)\astab_{D,\epsilon}(X).
$
In particular, if $X$ is agent-assignment stable, no feasible joint reassignment can improve the declared aggregate value by more than $\epsilon$.
\end{theorem}
The proof is given in Appendix~\ref{app:assignment-stability-proofs}. The left-hand side is the declared global assignment regret of replacing the recorded assignment $\alpha_Q$ by $\beta$. The theorem is local-to-global because componentwise stability bounds this regret for every feasible reassignment of the full evaluated work set; it does not require strategic agents or an equilibrium interpretation.

\subsection{iCORE-State Quality and Task Performance}
\label{sec:core-quality}

\begin{definition}[iCORE-state quality]
\label{def:core-quality}
The normalized iCORE-state quality is
$
\cq_{D,\epsilon}(X)
=\frac{1}{2}\bigl(\snd_D(X)+\astab_{D,\epsilon}(X)\bigr)\in[0,1].
$
We report $100\cq_{D,\epsilon}(X)$ together with both component scores; equal weighting is fixed before evaluation.
\end{definition}

Because both components are weight-normalized satisfaction rates, $\cq_{D,\epsilon}$ is invariant to common weight rescaling and supports comparison across work-set sizes; reporting the components separately prevents the average from hiding a severe soundness or assignment defect.

\begin{proposition}[Performance implication of iCORE-state quality]
\label{prop:quality-performance}
Let $\mathcal{W}=\WorkSet(X)\neq\varnothing$. Assume (i) $\Perf=\sum_{q\in\mathcal{W}}w_qY_q$ with $Y_q\ge0$ and $\sum_{q\in\mathcal{W}}w_q=1$; (ii) $V$ is task-adequate; (iii) every justified $q\in\mathcal{W}$ satisfies $\mathbb{E}[Y_q\mid X]\ge v_D(\alpha_Q(q),q\mid X)$; and (iv) every $q\in\mathcal{W}$ has a feasible agent with value at least $p_{\min}>\epsilon$. Then
$
\mathbb{E}[\Perf\mid X]
\ge 2(p_{\min}-\epsilon)
\left[\cq_{D,\epsilon}(X)-\frac{1}{2}\right]_+.
$
\end{proposition}

The proof is given in Appendix~\ref{app:quality-performance-proof}. Assumption~(i) restricts the bound, not the definition of $\cq_{D,\epsilon}$, to metrics decomposable over required roots or verifier-certified components. Without such a decomposition, adequate validators, lower calibration, or the feasible-agent floor, $\cq_{D,\epsilon}$ remains a process diagnostic rather than a performance certificate.

\section{iCORE-Audit}
\label{sec:core-audit}

iCORE-Audit is an event-triggered intervention layer around an existing multi-agent execution policy. It neither defines the task nor generates or selects ordinary work; $\operatorname{Schedule}_D$ in Algorithm~\ref{alg:core-audit} is only the host-policy interface. Each proposal is labeled by its intended effects on $G$, $Q$, and $\Pi$ and audited under Definition~\ref{def:legal-core-update}. Unsupported work is detected from certificate or event--work correspondence defects, assignment instability from Definition~\ref{def:assignment-stability}, and stalling from a progress predicate fixed in $D$, so no trigger relies on unrecorded monitor judgment.

After each ordinary attempt, iCORE-Audit recomputes $\snd_D(X)$, $\astab_{D,\epsilon}(X)$, and $\cq_{D,\epsilon}(X)$ from $\WorkSet(X)$, active certificates, assignments, and capability evidence. A \emph{hard flag} records a repair-requiring defect, such as an unsupported attempted assertion, a revocation that makes dependent work repair-pending, or a validator inconsistency; the unsupported semantic effect is not committed. Let $h$ count consecutive iterations since the last committed transition satisfying the predeclared $\operatorname{Progress}_D$ predicate. It resets only on such a transition and otherwise increases after a legal nonproductive step, a rejected proposal, or no enabled productive work. An intervention fires when the current state has no enabled productive ordinary update, $h\ge H$, a hard flag, an $\epsilon$-beneficial reassignment, or $\cq_{D,\epsilon}(X)<\tau$; $H$ and $\tau$ are fixed before execution.

A diagnosis $F$ identifies the affected obligations, missing or invalid evidence, current assignees, feasible alternatives, and relevant interaction history. The deterministic routing rule $\operatorname{ResponsibleAgent}_D(X,F)$ selects the agent authorized to propose the repair; ``responsible'' denotes repair authority rather than blame. For a work or evidence defect, this is the current assignee or the proposer whose attempted update produced the flag, according to a priority rule fixed in $D$. For an assignment defect, it is the recorded agent designated by $D$ as authorized to revise assignments, which may be the original assigner and need not be the currently assigned worker. The selected agent proposes a repair, verification, decomposition, or reassignment, but iCORE-Audit still labels and audits the candidate recovery. The recovery is committed only if it is legal and work-sound, satisfies the matched-set conditions below, and resolves every diagnosed hard flag whose resolution is required for continuation or acceptance; otherwise, only a legal failure-and-feedback record is added to the iCORE state.

To compare intervention quality without rewarding the silent deletion of difficult work, let $X^{-}$ denote the state at which $F$ is issued, let $X_f$ be the candidate repaired state, and let $X^{+}=\operatorname{Commit}_D(X_f)$ when the recovery is accepted. Before feedback is sent, $F$ freezes a matched audit set $\mathcal W_F$ containing the affected obligations and every active support component whose assertions or assignments may be changed by the recovery. The comparison is therefore between the problematic pre-intervention state and its proposed repair on the same frozen cohort:
$
\cq_{D,\epsilon}(X_f;\mathcal W_F)
\ge
\cq_{D,\epsilon}(X^{-};\mathcal W_F).
$
After commitment, the left-hand side is equivalently evaluated at $X^{+}$. A valid discharge, replacement, or supersession counts as satisfied, whereas removal without valid transport or settlement evidence remains in the frozen cohort and counts as unsound. The complete matched-set semantics are given in Appendix~\ref{app:matched-intervention-quality}.

Algorithm~\ref{alg:core-audit}, given in Appendix~\ref{app:core-audit-algorithm}, formalizes this loop. It checks the terminal predicates, requests an ordinary proposal through the host execution policy when productive work is enabled, audits the proposal, updates the patience counter $h$, and recomputes the iCORE diagnostics. When a trigger fires, it constructs $F$ and $\mathcal W_F$, routes the diagnosis to the declared repair authority, and commits the proposed recovery only after $\operatorname{AuditRecovery}_D$ returns no violation. Because rejected proposals, diagnoses, feedback, and accepted recoveries are represented by legal iCORE updates, a later auditor can reconstruct what information was supplied, which agent responded, and how the response changed the work and assignment state.

\begin{proposition}[Soundness and intervention quality]
\label{prop:core-audit-properties}
Assume that $X_0$ is work-sound, every committed update in Algorithm~\ref{alg:core-audit} is legal, and $\operatorname{AuditRecovery}_D$ enforces the matched-set conditions above. Then every committed state is work-sound. Moreover, for every committed intervention from $X^{-}$ to $X^{+}$ with diagnosis $F$,
$
\cq_{D,\epsilon}(X^{+};\mathcal W_F)
\ge
\cq_{D,\epsilon}(X^{-};\mathcal W_F).
$
\end{proposition}
The proof is given in the Appendix.%~\ref{app:core-audit-proof}.
If Proposition~\ref{prop:quality-performance}'s assumptions hold for both matched states on normalized $\mathcal{W}_F$ with the same $p_{\min}$ and $\epsilon$, its certified lower bound cannot decrease and increases strictly when post-intervention matched quality exceeds both its pre-intervention value and $1/2$. This is a conditional guarantee, not a claim that every intervention improves every open-ended outcome.

%% file: 05experiments.tex
\section{Experiments}
\label{sec:experiments}

We test three claims: whether the coupled state $X=(G,Q,\Pi)$ is necessary to reconstruct both work-soundness and assignment defects; whether state-grounded auditing improves passive observation of the same full state; and whether iCORE quality follows the conditional performance relation in Proposition~\ref{prop:quality-performance}. Realized terminal performance is reported separately from process quality.

\subsection{Protocol}

\paragraph{Applications and systems.}
We use six structured collaborative task (SCT) applications with obligation graphs that evolve during execution: Binary Coverage (BC), Decompose and Solve (DS), Compatible Merge (CM), Dependency Ladder (DL), Repair Verification (RV), and Capability Assignment (CA), with terminal performance defined as the weighted correctness of required roots. Each uses three agents, $\epsilon=0.05$, and $p_{\min}=0.65$.
The controlled suite evaluates all six under five cooperation failure conditions, allowing us to observe iCORE's ability to detect distinct classes of faults; the real-LLM suite evaluates BC, DS, and CA under clean and mixed faults with Qwen2.5-0.5B-Instruct. The main suites contain $540$ controlled and $108$ real-LLM episodes, while ablation and boundary tests bring the total to $882$ artifacts. Appendix~\ref{app:experiments-revised} gives the full matrix, replication settings, and runtime specifications.

We compare two variants of iCORE to several baselines. In MAS-Only, agents receive no monitor state; agents under Interaction-only (INT), Task-only (TASK), and LLM-Judge (JUDGE) observe $G$, $Q$, and the execution trace, respectively. iCORE-Observe (iCORE-O) and iCORE-Audit (iCORE-A) both observe $X=(G,Q,\Pi)$, but only iCORE-A returns structured diagnoses and applies the frozen-set preview guard. A shared proposal backend within each mode makes iCORE-A versus iCORE-O isolate intervention.

\paragraph{Metrics and inference.}
We report %trajectory 
iCORE quality $\cq_{D,\epsilon}$, the conditional performance from Proposition~\ref{prop:quality-performance}, and terminal performance on required roots. Proposal-probe detection is measured before feedback. Comparisons use matched episode keys, one-sided paired Wilcoxon tests with Holm correction, and paired-bootstrap $95\%$ confidence intervals.

\subsection{Reconstruction and Auditing Effects}

\begin{figure}[t]
    \centering

    \begin{minipage}[t]{0.49\linewidth}
        \centering
        \vspace{0pt}
        \includegraphics[width=\linewidth]{%
            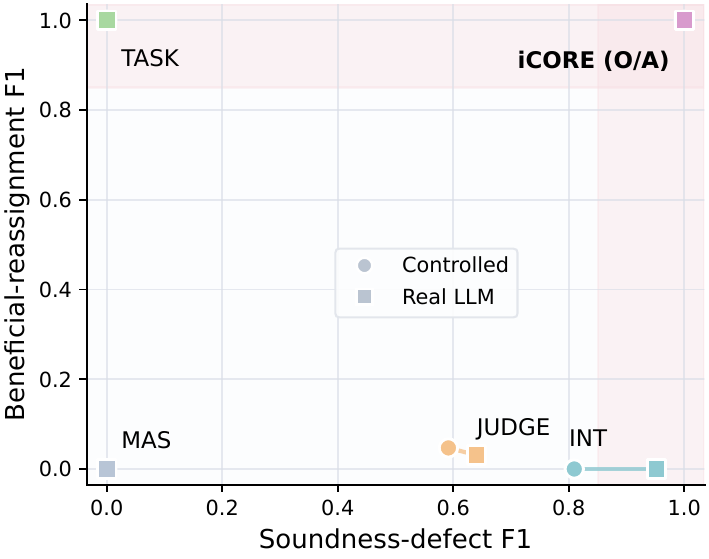}
        \vspace{-2mm}

        {\footnotesize\textbf{(a)} Representation validity}
    \end{minipage}
    \hfill
    \begin{minipage}[t]{0.49\linewidth}
        \centering
        \vspace{0pt}
        \includegraphics[width=\linewidth]{%
            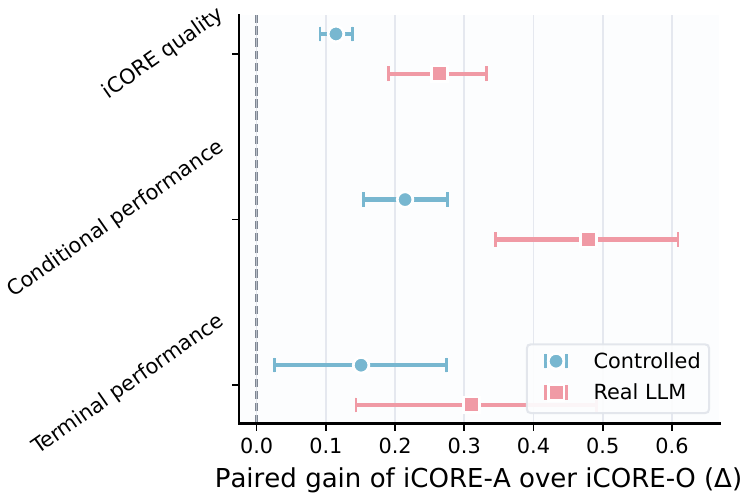}
        \vspace{-2mm}

        {\footnotesize\textbf{(b)} Paired auditing effect}
    \end{minipage}

    \vspace{-1mm}
    \caption{\textbf{Coupled-state reconstruction and auditing.}
\textbf{(a)} Representation validity across observation settings and execution modes. Markers distinguish controlled and real-LLM execution. iCORE (O/A) overlaps because both use FULL.
\textbf{(b)} Paired iCORE-A minus iCORE-O gains. Error bars are paired-bootstrap $95\%$ confidence intervals over $90$ controlled and $18$ real-LLM episode keys.}
    \label{fig:core-main-results-revised}
\end{figure}

Figure~\ref{fig:core-main-results-revised}(b) compares iCORE-A with iCORE-O on identical episode keys. Auditing improves trajectory quality, conditional performance, and terminal performance by $0.115/0.214/0.151$ in controlled execution and $0.264/0.480/0.310$ in real-LLM execution. Moreover, every recorded intervention is nondecreasing in matched iCORE quality ($512$ controlled and $108$ real-LLM records). This monotonicity is enforced by the preview guard, while the paired episode gains measure its downstream effect.

\subsection{Aggregate Results}

\begin{table}[t]
\centering
\scriptsize
\setlength{\tabcolsep}{2.2pt}
\renewcommand{\arraystretch}{1.03}
\resizebox{\linewidth}{!}{%
\begin{tabular}{@{}lccc|ccc@{}}
\toprule
& \multicolumn{3}{c|}{Controlled}
& \multicolumn{3}{c}{Real LLM} \\
\cmidrule(lr){2-4}
\cmidrule(lr){5-7}
System
& CQ & Cond. & Term.
& CQ & Cond. & Term. \\
\midrule
Best partial
& \underline{0.738}
& \underline{0.592}
& \underline{0.671}
& \underline{0.683}
& \underline{0.511}
& \underline{0.495} \\

iCORE-O
& 0.684
& 0.570
& 0.638
& 0.586
& 0.392
& 0.356 \\

iCORE-A
& \textbf{0.799}
& \textbf{0.785}
& \textbf{0.789}
& \textbf{0.850}
& \textbf{0.872}
& \textbf{0.667} \\
\bottomrule
\end{tabular}%
}
\caption{Main-suite results under controlled and real-LLM execution. ``Best partial'' is the columnwise best result among MAS-Only, INT, TASK, and JUDGE. Cond. and Term. denote conditional and terminal performance. Complete system-wise results appear in Appendix~\ref{app:aggregate-revised}.}
\label{tab:core-main-results-revised}
\end{table}

Table~\ref{tab:core-main-results-revised} compares iCORE-A with passive full-state observation and the strongest partial-state baseline. iCORE-A is best on all six mode--metric combinations, with particularly large gains in conditional performance under real-LLM execution. Its improvements over all five baselines remain significant after Holm correction (largest adjusted $p=8.71\times10^{-3}$ controlled and $3.41\times10^{-2}$ real LLM).

For every controlled quality bin with a positive certified bound, even the lower bootstrap endpoint of conditional performance remains above the bound in Proposition~\ref{prop:quality-performance}. The boundary suite shows why this conclusion requires adequate validators, calibrated capability values, and feasible agents, and the real-LLM suite is a finite candidate-selection stress test rather than an open-ended autonomous benchmark. Full results, boundary tests, and prompts appear in Appendix~\ref{app:experiments-revised}.

%% file: 06conclusion.tex
\section{Conclusion}
\label{sec:conclusion}

This paper introduced iCORE, a coupled representation $X=(G,Q,\Pi)$ for observable cooperation, evolving work obligations and assignments, and their certificate-backed correspondence. iCORE separates two complementary process properties: work soundness, which requires each decision-relevant work assertion to be justified, and agent-assignment stability, which requires active work to be assigned within a declared tolerance of the best feasible heterogeneous agent. Their quantitative satisfaction rates define iCORE-state quality, making it possible to compare states, interventions, and complete executions. Under decomposable tasks, adequate validators, and calibrated capability values, this score provides a lower bound on expected task performance; outside these assumptions, it remains a process-quality diagnostic rather than a certificate of open-ended answer quality. iCORE-Audit operationalizes the representation by allowing the ordinary agent process to run and returning state-grounded feedback when the process stalls or exposes a quality defect. Every diagnosis and recovery remains part of the auditable state, so quality improvement does not come at the cost of losing process accountability.

%% file: 07appendix.tex
\clearpage

\appendix

\section{Formal iCORE Semantics}
\label{app:core-semantics}

This appendix provides the detailed semantics and proofs underlying Sections~\ref{sec:core}--\ref{sec:core-audit}. It retains the complete execution model while the main text presents only the definitions needed to state the central results.

\subsection{Typed Observable Graphs}
\label{app:typed-graphs}

A typed observable graph is
\[
    H=(V_H,E_H,s_H,t_H,\theta_H,\lambda_H),
\]
where $V_H$ and $E_H$ are finite node and edge sets, $s_H,t_H:E_H\rightarrow V_H$ are source and target maps, and $\theta_H$ and $\lambda_H$ assign declared types and observable labels. We write $u\rightarrow_H v$ when an edge connects $u$ to $v$ and $u\rightsquigarrow_H v$ when there is a directed path from $u$ to $v$.

A field is \emph{semantic} when it may affect update enablement, feasibility, coverage, compatibility, or a terminal decision. Two typed graphs are observationally equivalent when a type-preserving graph isomorphism preserves every semantic field declared by $D$. Only representational fields, such as fresh identifiers, formatting, or inessential timestamps, may differ.

iCORE imposes an observational boundary. Private prompts, hidden chain-of-thought, and internal tool-selection heuristics may influence an agent, but they cannot support a work assertion or terminal decision unless their relevant effect is emitted as an observable action, artifact, label, or certificate premise represented in $G$, $Q$, or $\Pi$. Thus auditability concerns declared observables rather than access to an agent's hidden internal reasoning.

For $S\subseteq V_G$, define
\[
    \operatorname{Past}_G(S)
    =
    \{u\in V_G:
      u\rightsquigarrow_G s
      \text{ for some }s\in S\}.
\]
The live frontier $\operatorname{Live}(G)$ contains the actions and artifacts whose effects may still influence unresolved work or a terminal decision. A node leaves the live frontier only when its use, resolution, rejection, discard, invalidation, or supersession is observably recorded.

\subsection{Work Semantics}
\label{app:work-semantics}

The application declares how each relation in $\mathcal{R}_Q$ contributes to obligation, coverage, and support. We write
\[
    p\prec_Q^{\mathrm{obl}}q
\]
when $p$ must be resolved before $q$,
\[
    p\prec_Q^{\mathrm{cov}}q
\]
when $q$ may contribute to discharging $p$, and
\[
    p\prec_Q^{\mathrm{sup}}q
\]
when $p$ supports the existence or state of $q$.

The exact status vocabulary is
$
\mathcal{S}_Q=
\{
\texttt{open},
\texttt{ready},
\texttt{blocked},
\texttt{closed},\\
\texttt{submitted},
\texttt{accepted},
\texttt{rejected},
\texttt{invalid}
\}.
$

The required roots satisfy
$
    \operatorname{ReqRoot}(Q)
    \subseteq
    \{q\in V_Q:
      \nexists p\in V_Q
      \text{ with }
      p\prec_Q^{\mathrm{obl}}q\}.
$
The unresolved and accepted work sets are
$
    \operatorname{Unres}(Q)
    =
    \{q\in V_Q:
      \operatorname{stat}_Q(q)
      \in
      \{\texttt{open},\texttt{ready},
      \texttt{blocked},\texttt{submitted}\}\},
$
and
\[
    \operatorname{Accepted}(Q)
    =
    \{q\in V_Q:
      \operatorname{stat}_Q(q)=\texttt{accepted}\}.
\]

For a work component $q$, its root coverage is
\[
    \operatorname{Cov}_Q(q)
    =
    \{r\in\operatorname{ReqRoot}(Q):
      r\preceq_Q^{\mathrm{cov}}q\}.
\]
A set $A\subseteq\operatorname{Accepted}(Q)$ is a terminal cover of
$R\subseteq\operatorname{ReqRoot}(Q)$ when
\[
    R\subseteq
    \bigcup_{a\in A}\operatorname{Cov}_Q(a)
\]
and the joint compatibility evidence required by $V$ exists.

A required root is \emph{discharged} when a terminal cover exists for it. We write
\[
    \operatorname{Dis}(Q)
    =
    \{r\in\operatorname{ReqRoot}(Q):
      r\text{ is discharged}\}.
\]

The declaration package fixes a finite observable agent set $\mathcal{M}_D$. The map $\alpha_Q$ records the current assignee of unresolved work and retains the last recorded assignee of settled work. For a component $q$, let $\operatorname{Dis}_{-q}(X)$ contain the effective roots discharged by active accepted work whose support chain does not depend on $q$, and define
\[
    \operatorname{Ob}_X(q)
    =\operatorname{Cov}_Q(q)\cap
    \bigl(\operatorname{EffRoot}(X)\setminus
    \operatorname{Dis}_{-q}(X)\bigr).
\]
The evaluation set $\WorkSet(X)$ contains every unresolved component $q$ with
$\operatorname{Ob}_X(q)\neq\varnothing$, together with every component carrying
an active decision-relevant assertion for a iCORE diagnostic or terminal
predicate. Consequently, settled support components remain auditable rather
than disappearing from terminal quality scores. If $\WorkSet(X)=\varnothing$,
the task has no remaining or terminal-support work and both component scores
are defined as one by convention.

For each $q\in\WorkSet(X)$, $D$ specifies a feasible-agent set $\mathcal{F}_D(q,X)$ and a normalized net contribution value $v_D(m,q\mid X)\in[0,1]$. The formula and permissible evidence are fixed before evaluation. For settled work, $v_D$ may use only capability evidence observable no later than the relevant assignment or settlement event; post-outcome oracle information is excluded. Any capability observation used by $v_D$ must be represented in $G$, and any assignment or reassignment in $Q$ must be connected through $\Pi$ to the corresponding observable event and active valid evidence.

\subsection{Certificates and Correspondence}
\label{app:certificates}

A certificate is a finite record
\[
    c=
    \bigl(
    \operatorname{id}(c),
    \operatorname{prem}_{\mathrm{int}}(c),
    \operatorname{prem}_{\mathrm{ext}}(c),
    \operatorname{claim}(c),
    \operatorname{issuer}(c)
    \bigr).
\]
Internal premises reference information already represented in $X$. External premises reference sources whose authority and scope are declared by $V$.

An external premise enters a work-support chain only through a declared environment-input node in $G$ and an active valid certificate linking that node, through $\Pi$, to the affected work assertion. External evidence therefore does not bypass the observable support requirement.

A certificate is active when no active valid revocation or supersession certificate targets it. The predicate
\[
    \operatorname{ValidCert}_V(c;X)=1
\]
holds when $c\in\kappa$, $c$ is active, its internal premises are present and current, its external premises satisfy $V$, and its claim passes the declared validation procedure. The active certificate-dependency graph must be acyclic.

A work-level assertion includes any asserted node, relation, status, signature, coverage fact, compatibility fact, or terminal marker in $Q$. Every non-initial decision-relevant assertion must be connected through $\Pi$ to observable execution and active valid evidence.

\subsection{Feasibility}
\label{app:feasibility}

A well-typed state $X=(G,Q,(\pi,\omega,\kappa))$ is feasible when:

\begin{enumerate}
    \item every required root is connected to an observable initial input and supporting certificate;
    \item every non-initial work assertion is backed by responsible nodes in $G$, corresponding operation labels, and active valid certificates;
    \item every non-root work node has a finite support chain ending at a required root or declared external input;
    \item every unresolved work component remains connected to live execution or has an explicitly certified blocked, delegated, waiting, submitted, superseded, or repair-pending condition;
    \item coverage uses only eligible work, and every aggregate has the joint compatibility evidence required by $V$;
    \item every active assignment in $Q$ names a feasible agent or carries valid blocked/unassigned evidence, and every assignment or reassignment is supported by an observable event and capability evidence through $\Pi$;
    \item every terminal event in $G$ corresponds to the appropriate work state in $Q$, and every terminal work state has an observable witness;
    \item revoked or superseded evidence cannot continue supporting an active assertion without re-certification, repair, invalidation, rejection, or supersession; and
    \item no active assertion simultaneously depends on a certificate and an active valid revocation of that certificate.
\end{enumerate}

\subsection{iCORE Systems, Local Views, and Admissible Processes}
\label{app:local-updates}

\begin{definition}[iCORE system]
A iCORE system is
\[
    \mathcal S=
    (D,\mathcal X_D,\widehat{\mathcal R},X_0,
    \operatorname{Acc},\operatorname{Rej}),
\]
where $D$ fixes the observable semantics, agent set $\mathcal M_D$, operation vocabulary, validation specification $V$, capability-value interface, local-view interface, and declared equivalence; $\mathcal X_D$ is the set of feasible iCORE states satisfying the preceding conditions; $\widehat{\mathcal R}$ is the family of candidate local-update schemas; $X_0\in\mathcal X_D$ is the initial state; and
\[
    \operatorname{Acc},\operatorname{Rej}:
    \mathcal X_D\rightarrow\{0,1\}
\]
are the declared terminal predicates.
\end{definition}
The declared state equivalence $X\equiv_X X'$ holds when $X$ and $X'$ differ only in fields declared inessential by $D$ while agreeing, under transported graph isomorphisms, on every field that may affect update enablement, obligations, assignments, capability evidence, or terminal decisions. The relation is required to be a congruence for legal updates.

A rewrite occurrence $\rho$ at $X=(G,Q,\Pi)$ has finite declared interfaces
$
\operatorname{Read}_G(\rho;X),$
$
\operatorname{Write}_G(\rho;X),$
$
\operatorname{Read}_Q(\rho;X),$
$
\operatorname{Write}_Q(\rho;X),$
$
\operatorname{Read}_{\Pi}(\rho;X),$
$
\operatorname{Write}_{\Pi}(\rho;X).
$
Its local cooperation view is
\[
    \operatorname{Loc}_G(\rho;X)
    =
    \operatorname{Read}_G(\rho;X)
    \cup
    \operatorname{Anc}^{D}_G
    (\operatorname{Read}_G(\rho;X)),
\]
where the permitted ancestry rule and finite depth bound are fixed by $D$. Its mapped work view is
\[
\begin{aligned}
    \operatorname{View}_Q(\rho;\Pi,X)
    ={}&\operatorname{Read}_Q(\rho;X)\\
    &\cup
    \bigcup_{g\in\operatorname{Loc}_G(\rho;X)\cap V_G}
    \pi(g),
\end{aligned}
\]
together with the incident semantic relations, current statuses and signatures, and active certificates required by the declared operation. The complete local view is
\[
    L(\rho;X)=
    (\operatorname{Loc}_G(\rho;X),
     \operatorname{View}_Q(\rho;\Pi,X),
     \operatorname{Read}_{\Pi}(\rho;X)).
\]

\begin{definition}[Legal local update]
A candidate update schema induces a partial function
\[
    \widehat R_{\rho}:\mathcal X_D\rightharpoonup\mathcal X_D,
    \qquad
    \widehat R_{\rho}(X)=X'.
\]
An occurrence is enabled when every declared object and certificate premise is present and current in $L(\rho;X)$ and $D$ permits the proposed operation. The enabled update is legal when it (i) reads and writes only through its declared interfaces, (ii) depends only on $L(\rho;X)$ and the fixed declarations, (iii) supports every new semantic assertion with active valid evidence under $V$, (iv) produces $X'\in\mathcal X_D$, and (v) respects declared state equivalence.
\end{definition}

\begin{definition}[Admissible process and reachability]
An admissible process is a finite or infinite sequence
\[
    X_0\rightarrow X_1\rightarrow X_2\rightarrow\cdots
\]
in which every transition is a legal local update. A state $Y$ is reachable when $X_0\rightarrow^*Y$ along a finite admissible prefix.
\end{definition}

Anchoring, obligation settlement, certificate consistency, and event--work correspondence are state-space conditions encoded by feasibility rather than separate dynamic axioms. The remaining dynamic contract is online locality plus preservation: a legal update may use only its current declared view and must preserve feasibility and equivalence congruence. An incorrect proposal must therefore become an observable blocked, rejected, invalid, superseded, excluded, or repair-pending fact rather than silently breaking the model.

\subsection{Terminal Predicates}
\label{app:terminal-predicates}

Let $\operatorname{Off}_X(Q)$ contain work components marked by active valid evidence as duplicated, excluded, irrelevant, or superseded. Define
\[
    \operatorname{EffRoot}(X)
    =
    \operatorname{ReqRoot}(Q)
    \setminus
    \operatorname{Off}_X(Q)
\]
and
$
    \operatorname{EffUnres}(X)
    =
    \big\{
    q\in
    \operatorname{Unres}(Q)\setminus\operatorname{Off}_X(Q):
    \operatorname{Cov}_Q(q)\cap\operatorname{EffRoot}(X)
    \neq\varnothing
    \big\}.
$

The failed roots $\operatorname{FailRoot}(X)$ are the effective roots with active valid rejection, invalidity, impossibility, or unreachability evidence and no valid repair neutralizing that failure.

A state satisfies $\operatorname{Acc}(X)$ when:

\begin{enumerate}
    \item $X$ is feasible and $\operatorname{EffUnres}(X)=\varnothing$;
    \item accepted work forms a certified terminal cover of $\operatorname{EffRoot}(X)$;
    \item no effective root belongs to $\operatorname{FailRoot}(X)$;
    \item no accepted component depends on unresolved, invalid, rejected, contradictory, or unrepaired work; and
    \item no live productive artifact can change the discharge, failure, exclusion, or compatibility of an effective root.
\end{enumerate}

A state satisfies $\operatorname{Rej}(X)$ when:

\begin{enumerate}
    \item $X$ is feasible, $\operatorname{EffUnres}(X)=\varnothing$, and $\operatorname{EffRoot}(X)\neq\varnothing$;
    \item every effective root belongs to $\operatorname{FailRoot}(X)$;
    \item
    \[
        \operatorname{EffRoot}(X)
        \cap
        \operatorname{Dis}(Q)
        =
        \varnothing;
    \]
    and
    \item no live productive artifact can repair, accept, discharge, or reopen an effective root.
\end{enumerate}

\subsection{Productivity and Fairness}
\label{app:productivity}

A transition $X\rightarrow Y$ is productive when
$X\not\equiv_X Y$. The productive quotient is defined by
\[
    [X]_{\equiv_X}
    \rightarrow_p
    [Y]_{\equiv_X}
\]
when there exist $X'\equiv_X X$ and $Y'\equiv_X Y$ such that
$X'\rightarrow Y'$ is productive. We write $\rightarrow_p^*$ for its reflexive-transitive closure.

The system is productively terminating when no reachable state begins an infinite productive path. It is productively normalizing when every reachable state has a finite productive continuation to a productive normal form.

An occurrence is continuously enabled when, from some point onward, its declared objects and premises remain available. A productive execution is fair when it does not permanently starve a continuously enabled productive occurrence. The occurrence must eventually execute, or another recorded productive update must invalidate one of its premises by resolving, blocking, delegating, excluding, superseding, or repairing the relevant work.

\subsection{Why the Coupling Is Necessary}
\label{app:projection-insufficiency}

\begin{theorem}[Projection insufficiency]
\label{thm:projection-insufficiency}
Consider any class of iCORE states containing two states with identical unlabeled cooperation-graph topology but different obligation, assignment, certificate, or event--work correspondence information, such that one state is work-sound or agent-assignment stable and the other is not. No monitor restricted to the unlabeled topology of $G$ can be both sound and complete for the differing property on that class.
\end{theorem}

\subsection{Matched Intervention Quality}
\label{app:matched-intervention-quality}

Let $F$ be a diagnosis produced at state $X$. The matched audit set $\mathcal W_F$ is the finite weighted set containing every diagnosed component, every active component whose support or assignment is changed by the proposed recovery, and the active support ancestors required to evaluate those assertions. The set and weights are frozen before the agent receives feedback.

For a successor state $Y$, the declaration package transports each $q\in\mathcal W_F$ through recorded identity, refinement, replacement, supersession, or settlement links. Define $J_D^F(q;Y)=1$ when the transported component remains fully justified or has been validly discharged, rejected, replaced, or superseded with all required evidence; deletion without a valid transport or settlement record gives $J_D^F(q;Y)=0$. Define $A_{D,\epsilon}^F(q;Y)=1$ when the transported unresolved component is $\epsilon$-stable, or when a settled component was completed under an assignment that was $\epsilon$-stable according to capability evidence available no later than settlement. The matched scores are
$
\snd_D(Y;\mathcal W_F)
=
\frac{\sum_{q\in\mathcal W_F}w_qJ_D^F(q;Y)}
{\sum_{q\in\mathcal W_F}w_q},
\qquad
\astab_{D,\epsilon}(Y;\mathcal W_F)
=
\frac{\sum_{q\in\mathcal W_F}w_qA_{D,\epsilon}^F(q;Y)}
{\sum_{q\in\mathcal W_F}w_q},
$
with value one when $\mathcal W_F=\varnothing$, and
\[
\cq_{D,\epsilon}(Y;\mathcal W_F)
=
\frac{1}{2}
\left(
\snd_D(Y;\mathcal W_F)
+
\astab_{D,\epsilon}(Y;\mathcal W_F)
\right).
\]
Freezing the cohort prevents an intervention from increasing its score merely by deleting, renaming, or excluding a difficult obligation without valid evidence.

\subsection{iCORE-Audit Procedure}
\label{app:core-audit-algorithm}

Algorithm~\ref{alg:core-audit} gives the complete state-maintenance loop. The routine $\operatorname{AuditOrdinary}_D$ returns the set of legality or work-soundness violations in an ordinary proposal. If that set is nonempty, $\operatorname{RecordFailure}_D$ records the attempted proposal and its audit result without committing the unsupported effects on $Q$ or $\Pi$. Similarly, $\operatorname{AuditRecovery}_D$ returns the violations of legality, work soundness, diagnosed-flag resolution, or matched-set quality. All recording operations below are themselves legal iCORE updates.

Here $h$ counts consecutive iterations without a committed transition satisfying $\operatorname{Progress}_D$. The recovery audit accepts $X_f$ only when the candidate update is legal and work-sound, does not silently remove any obligation in $\mathcal W_F$, resolves each diagnosed hard flag required for continuation or acceptance, and satisfies
\[
\cq_{D,\epsilon}(X_f;\mathcal W_F)
\ge
\cq_{D,\epsilon}(X^{-};\mathcal W_F).
\]
A successful but nonproductive recovery therefore does not reset $h$.

\begin{algorithm*}[t]
\footnotesize
\caption{iCORE-Audit with State-Grounded Agent Feedback}
\label{alg:core-audit}
\begin{algorithmic}[1]
\REQUIRE Initial request $u_0$, declarations $D$, tolerance $\epsilon$, quality threshold $\tau$, patience $H$, budget $B$
\STATE $X\gets\operatorname{Init}_D(u_0)$; $h\gets0$
\FOR{$t=0,1,\ldots,B-1$}
    \IF{$\operatorname{Acc}(X)$}
        \STATE \textbf{return} \textsc{Accept}
    \ELSIF{$\operatorname{Rej}(X)$}
        \STATE \textbf{return} \textsc{Reject}
    \ENDIF
    \STATE $U\gets\operatorname{EnabledProd}_D(X)$
    \IF{$U\neq\varnothing$}
        \STATE $(m,\rho)\gets\operatorname{Schedule}_D(X,U)$
        \STATE $L\gets\operatorname{LocalView}_D(X,\rho)$
        \STATE $(o,C)\gets\operatorname{AgentStep}_D(m,L)$
        \STATE $\widetilde X\gets\operatorname{LabelUpdate}_D(X,\rho,m,o,C)$
        \STATE $E\gets\operatorname{AuditOrdinary}_D(X,\widetilde X)$
        \IF{$E=\varnothing$}
            \STATE $X^{+}\gets\operatorname{Commit}_D(\widetilde X)$
            \STATE $h\gets0$ if $\operatorname{Progress}_D(X,X^{+})$; otherwise $h\gets h+1$
            \STATE $X\gets X^{+}$
        \ELSE
            \STATE $X\gets\operatorname{RecordFailure}_D(X,\widetilde X,E)$; $h\gets h+1$
        \ENDIF
    \ELSE
        \STATE $h\gets h+1$
    \ENDIF
    \STATE $q\gets\cq_{D,\epsilon}(X)$
    \IF{$U=\varnothing \lor h\ge H \lor \operatorname{HardFlag}_D(X) \lor \operatorname{BeneficialReassign}_{D,\epsilon}(X) \lor q<\tau$}
        \STATE $X^{-}\gets X$; $F\gets\operatorname{Diagnose}_D(X^{-})$; $\mathcal W_F\gets\operatorname{MatchedSet}_D(X^{-},F)$
        \STATE $m_f\gets\operatorname{ResponsibleAgent}_D(X^{-},F)$
        \STATE $L_f\gets\operatorname{FeedbackView}_D(X^{-},F)$
        \STATE $(o_f,C_f)\gets\operatorname{AgentRevise}_D(m_f,L_f,F)$
        \STATE $X_f\gets\operatorname{LabelIntervention}_D(X^{-},m_f,o_f,C_f,F)$
        \STATE $E_f\gets\operatorname{AuditRecovery}_D(X^{-},X_f,F,\mathcal W_F)$
        \IF{$E_f=\varnothing$}
            \STATE $X^{+}\gets\operatorname{Commit}_D(X_f)$
            \STATE $h\gets0$ if $\operatorname{Progress}_D(X^{-},X^{+})$; otherwise $h\gets h+1$
            \STATE $X\gets X^{+}$
        \ELSE
            \STATE $X\gets\operatorname{RecordFeedback}_D(X^{-},X_f,F,E_f)$; $h\gets h+1$
            \IF{$\operatorname{CertifiedStuck}_D(X)$}
                \STATE \textbf{return} \textsc{Stuck}
            \ENDIF
        \ENDIF
    \ENDIF
\ENDFOR
\STATE \textbf{return} $\operatorname{CertifyBudget}_D(X)$ if valid; otherwise \textsc{Incomplete}
\end{algorithmic}
\end{algorithm*}
\FloatBarrier

\section{Supplementary Execution-Order Stability}
\label{app:execution-order-stability}

This section retains the earlier task-process notion of stability as a supplementary property. It is intentionally excluded from the main text, where stability refers only to heterogeneous-agent assignment. Here the question is whether different legal orders of productive updates reach equivalent semantic outcomes.

A legal update $X\rightarrow Y$ is productive when $X\not\equiv_X Y$ and stuttering otherwise. A reachable state is a productive normal form when no productive legal update remains enabled. Suppose two productive updates $\rho$ and $\eta$ are enabled at $X$. They are \emph{residual-sound independent} when both residual executions $\rho;\eta_\rho$ and $\eta;\rho_\eta$ are defined and reach equivalent states. Otherwise the one-step fork is a \emph{productive local critical pair}.

Two states $Y$ and $Z$ are productively joinable modulo $\equiv_X$, written $Y\downarrow_{p,\equiv_X}Z$, when they have productive descendants $Y'$ and $Z'$ with $Y'\equiv_X Z'$.

\begin{definition}[Execution-order stability]
\label{def:execution-order-stability}
A iCORE system is \emph{globally execution-order stable} when any two productive continuations from the same reachable state to productive normal forms terminate in equivalent iCORE states.
\end{definition}

\begin{theorem}[Local-to-global execution-order stability]
\label{thm:execution-order-local-to-global}
Assume that the reachable productive quotient terminates and every reachable one-step productive fork is either residual-sound independent or a productive local critical pair. Then the system is globally execution-order stable if and only if every reachable productive local critical pair is productively joinable modulo $\equiv_X$.
\end{theorem}

Joinability requires agreement across the full state: the branches must reconcile their cooperation events in $G$, work and assignment states in $Q$, and event--work/certificate correspondence in $\Pi$. Therefore an unlabeled cooperation topology alone cannot characterize execution-order stability. This supplementary property may be useful when reproducible ordering is itself an application requirement, but it is not used to define agent-assignment stability or iCORE-state quality in the main text.

% ===== CTRL+F: REVISED iCORE EXPERIMENT APPENDIX START =====
\section{Additional Experimental Details}
\label{app:experiments-revised}

This appendix documents the exact experiment matrix materialized in the supplied artifact. All $882$ scheduled keys listed below produced episode JSON files; ``completed'' elsewhere in the tables is a task-success indicator rather than a file-completion indicator.

\subsection{Experiment Matrix}

\begin{table*}[t]
\centering
\scriptsize
\setlength{\tabcolsep}{3.5pt}
\begin{tabular}{@{}p{0.075\textwidth}p{0.065\textwidth}p{0.145\textwidth}p{0.155\textwidth}p{0.33\textwidth}r@{}}
\toprule
Mode & Suite & Environments & Systems/variants & Conditions or boundaries & Episodes \\
\midrule
Controlled & Main & BC, DS, CM, DL, RV, CA & six systems & clean, missing evidence, misassignment, verifier delay, stall & 540 \\
Real LLM & Main & BC, DS, CA & six systems & clean, mixed & 108 \\
Controlled & Ablation & BC, DS, CM, DL, RV, CA & Full, SO, AO, $-\Pi$, NC & mixed & 90 \\
Controlled & Boundary & BC, DS, CM, DL, RV, CA & iCORE-O, iCORE-A & standard, DV, SC, WP & 144 \\
\midrule
\multicolumn{5}{r}{Total materialized episode artifacts} & 882 \\
\bottomrule
\end{tabular}
\caption{Executed matrix. DV, SC, and WP denote degraded validator, shuffled declared capabilities, and weak agent pool.}
\label{tab:core-matrix-revised}
\end{table*}

The controlled main, real-LLM main, ablation, and boundary suites contain $540$, $108$, $90$, and $144$ materialized episode artifacts, respectively.

\subsection{Full Aggregate Results}
\label{app:aggregate-revised}

\input{results/tables_core_revised/main_results_revised}

\subsection{Structured Collaborative Task Environments}
\label{app:environments-revised}

All environments use agents \texttt{agent\_1}, \texttt{agent\_2}, and \texttt{agent\_3}, with $\epsilon=0.05$ and $p_{\min}=0.65$. Every action becomes an event in $G$; work status, assignee, prerequisites, values, and weights reside in $Q$; assignment, result, merge, verification, and repair certificates are appended to $\Pi$. The implementation is a self-contained graph runtime. It does not import an external DIG runtime package.

\begin{table*}[t]
\centering
\scriptsize
\setlength{\tabcolsep}{4.0pt}
\renewcommand{\arraystretch}{1.10}
\begin{tabular}{@{}llllll@{}}
\toprule
Env. & Obligation structure & Targets/validator & Required roots & Max rounds & Primary mechanism \\
\midrule
BC & six independent atomic items & each value $1$, exact & six, weight $1$ each & 7 & complete each item exactly once with active evidence \\
DS & aggregate root plus four latent children & children $1,2,3,4$; root sum $10$ & root, weight $2$ & 9 & decompose, solve, certify, and merge \\
CM & three source items supporting one aggregate & $0.2,0.3,0.5\rightarrow1.0$ & aggregate, weight $2$ & 8 & retain all compatible support sources \\
DL & chain $d_0\to d_1\to d_2\to d_3$ & values $1,2,3,4$, exact & $d_3$, weight $1.5$ & 8 & respect certified prerequisites \\
RV & draft, check, repair, final chain & $7,7,8,8$; versioned evidence & final, weight $2$ & 9 & invalidate stale version-1 evidence and certify repair \\
CA & four independent heterogeneous items & $11,13,17,19$, exact & four, weight $1$ each & 7 & assign math/search/code/review to capable agents \\
\bottomrule
\end{tabular}
\caption{Environment specifications}
\label{tab:core-environments-revised}
\end{table*}

\paragraph{Binary Coverage (BC).}
The six work points $p_0,\ldots,p_5$ are independent required roots. A valid terminal state requires every point to be closed with a valid result certificate. Duplicate completion is permitted by the runtime but consumes an agent action and can be detected from recent events in $G$.

\paragraph{Decompose and Solve (DS).}
The root begins open while children $s_1,\ldots,s_4$ are latent. A legal decomposition activates the children. Their exact values must sum to $10$, and the aggregate certificate must retain the child identifiers. Direct root completion before decomposition lacks the required support relation.

\paragraph{Compatible Merge (CM).}
Three independent components with values $0.2$, $0.3$, and $0.5$ support an aggregate root of $1.0$. The merge is justified only if all support components are current and certified.

\paragraph{Dependency Ladder (DL).}
Only $d_0$ is initially open. The remaining nodes are blocked until all predecessors close. Completing a blocked node still creates an observable event, but its work assertion is unsupported.

\paragraph{Repair Verification (RV).}
The graph requires a value-$7$ draft and check, followed by a value-$8$ repair and final result. Check and repair items use versioned validators. Evidence with version $1$ is stale after the required version becomes $2$.

\paragraph{Capability Assignment (CA).}
Four independent roots encode math, search, code, and review. Ground-truth capability tables are fixed before evaluation. For example, the best agents for math, search, and code have probabilities $0.94$, $0.95$, and $0.96$; review is best assigned to \texttt{agent\_2} with probability $0.88$. A work item is stable when no declared feasible agent improves expected contribution by more than $\epsilon=0.05$.

\subsection{Systems and Observable Projections}

\begin{table*}[t]
\centering
\small
\setlength{\tabcolsep}{5.2pt}
\begin{tabular}{@{}llll@{}}
\toprule
System & Projection & Intervention & Operational distinction \\
\midrule
MAS-Only & none & no & shared worker policy without monitoring \\
Interaction-only & $G$ & yes & detects duplicate actions and stalls from event history \\
Task-only & $Q$ & yes & detects dependencies, incomplete submission, and beneficial reassignment \\
LLM-Judge & trace & yes & free-form trace judge; no direct $Q$ or $\Pi$ access \\
iCORE-Observe & $(G,Q,\Pi)$ & no & exact passive measurement of both iCORE components \\
iCORE-Audit & $(G,Q,\Pi)$ & yes & structured diagnosis, recommended repair, and frozen-set preview guard \\
\bottomrule
\end{tabular}
\caption{Evaluated systems. iCORE-Observe and iCORE-Audit intentionally share the same measurement projection.}
\label{tab:core-systems-revised}
\end{table*}

For every proposed action, the evaluator first creates an unreviewed proposal state and records measurement predictions on that same state. Only then does the selected system review or modify the action. This proposal-probe protocol prevents iCORE-Audit from receiving artificially easy labels after it repairs a defect.

iCORE-Audit diagnoses duplicate work, assignments more than $\epsilon$ below the best declared agent, blocked dependencies, missing evidence, incomplete submission, and unsupported closed work. It evaluates the recommended successor on the same frozen work set. A candidate that reduces frozen-set soundness or iCORE quality is replaced with \texttt{wait}; the diagnosis remains visible to the next decision. This explains why nondecreasing matched intervention quality is an enforced guard property rather than an unconstrained empirical coincidence.

\subsection{Controlled Proposal Backend}
\label{app:controlled-backend-revised}

Every controlled system uses the same \texttt{SharedPolicyMLP}: a $12\rightarrow32\rightarrow32\rightarrow1$ network with ReLU activations. Its twelve candidate features are enabled status, assigned-to-self, unassigned, prerequisites ready, aggregate ready, evidence available, current-agent capability, best capability, already complete, duplicate risk, audit recommendation, and terminal-submit readiness. The checkpoint is trained on $8{,}000$ synthetic feature vectors for $20$ epochs with batch size $256$, AdamW learning rate $2\times10^{-3}$, weight decay $10^{-4}$, and binary cross entropy. A deterministic keyed Gaussian jitter with standard deviation $0.025$ breaks exact score ties. If the checkpoint is absent, the implementation falls back to the fixed linear score recorded in \texttt{models.py}; the supplied run used the trained checkpoint.

The controlled JUDGE condition cannot call a language model. Its review surrogate selects the first diagnosed alternative with probability $0.70$ and otherwise selects a keyed random alternative. The real-LLM JUDGE condition instead uses the separate judge prompt below.

\subsection{Real-LLM Backend and Prompts}
\label{app:prompts-revised}

Reporting prompt templates is standard and useful for reproducibility in LLM-agent experiments. We therefore include the exact system strings and the JSON payload schemas used by the artifact. The LLM does not generate arbitrary tools or free-form work products. It selects one candidate from the finite list constructed by the environment. Consequently, this suite tests orchestration and monitor feedback under an LLM decision policy; it is not an open-ended natural-language problem-solving benchmark.

\paragraph{Worker system prompt.}
\begin{lstlisting}[numbers=none]
You are one cooperative worker in a multi-agent graph task. Choose exactly one listed action. Return JSON only as {"candidate_index": <integer>, "reasoning": "brief"}.

<environment description from the YAML specification>
\end{lstlisting}

\paragraph{Worker user payload.}
\begin{lstlisting}[numbers=none]
{
  "observation": {
    "agent": <agent id>,
    "round": <integer>,
    "max_rounds": <integer>,
    "own_assignments": <list>,
    "messages": <last five messages>,
    "feedback": <monitor message or null>
  },
  "candidate_actions": [
    {"index": <integer>, "action": <structured action>,
     "description": <candidate description>}, ...
  ]
}
\end{lstlisting}

\paragraph{Trace-judge system prompt.}
\begin{lstlisting}[numbers=none]
You are a process judge. Select the safest action index using only the provided observable trace summary. Return JSON only as {"candidate_index": <integer>}.
\end{lstlisting}

\paragraph{Trace-judge user payload.}
\begin{lstlisting}[numbers=none]
{
  "trace_summary": {"round": <integer>, "events": <last 12 events>},
  "proposed_action": <structured action>,
  "alternatives": <up to eight structured actions>
}
\end{lstlisting}

The local model is \texttt{Qwen/Qwen2.5-0.5B-Instruct}, accessed through the OpenAI-compatible vLLM endpoint. Decoding uses temperature $0$, at most $160$ completion tokens, and JSON-object response format. The server log records vLLM $0.26.0$, bfloat16 weights, maximum sequence length $4096$, tensor parallel size $1$, asynchronous scheduling, prefix caching, chunked prefill, and GPU-memory utilization $0.85$. It records one CUDA rank with $47.37$ GiB visible memory, but not the GPU product name; we therefore do not infer a hardware SKU from the log. The checkpoint occupies approximately $0.93$ GiB. The default endpoint is \texttt{http://127.0.0.1:8000/v1}.

\subsection{Fault Injection and Theory-Boundary Conditions}
\label{app:faults-revised}

The controlled main suite uses the following conditions.
\begin{itemize}
    \item \textbf{Clean:} no explicit action fault.
    \item \textbf{Missing evidence:} completion, merge, repair, or verification loses its evidence with probability $0.65$.
    \item \textbf{Misassignment:} assignment or reassignment is redirected to the least capable declared agent with probability $0.75$.
    \item \textbf{Verifier delay:} relevant certificates become available two rounds later.
    \item \textbf{Stall:} a proposed action becomes \texttt{wait} with probability $0.30$.
\end{itemize}
The real-LLM mixed condition combines missing evidence, misassignment, and stall with a common probability scale of $0.55$, yielding probabilities $0.65\times0.55$, $0.75\times0.55$, and $0.30\times0.55$. Fault draws use stable keyed randomness, so system-specific calls do not shift the perturbation schedule.

The boundary suite changes assumptions rather than merely increasing random noise. \textbf{DV} accepts a wrong result with probability $0.65$, breaking validator adequacy. \textbf{SC} cyclically permutes declared capability values while retaining the original ground-truth success probabilities, breaking calibration. \textbf{WP} scales both true and declared capabilities by $0.55$, weakening the feasible-agent floor. These variants are used to delimit Proposition~\ref{prop:quality-performance}, not to claim robustness under arbitrary assumption violations.

\subsection{Metrics}
\label{app:metrics-revised}

For an active or frozen work set, the implementation computes weighted fractions
\begin{align*}
\snd_D(X)
&=\frac{\sum_w \omega_w\,\mathbf{1}[w\text{ is justified}]}{\sum_w\omega_w},\\
\astab_{D,\epsilon}(X)
&=\frac{\sum_w \omega_w\,\mathbf{1}[\operatorname{gain}(w)\leq\epsilon]}{\sum_w\omega_w}.
\end{align*}
and $\cq_{D,\epsilon}(X)=(\snd_D(X)+\astab_{D,\epsilon}(X))/2$. The paper macros define the same quantities; the displayed indicator notation here is only an implementation-level restatement.

\textbf{Trajectory metrics} average state scores over all recorded rounds. \textbf{Conditional performance} sums the ground-truth success probability of each currently assigned, justified component and normalizes by active work weight. It is the empirical conditional expectation corresponding to the proposition. \textbf{Terminal performance} is the realized weighted correctness of required roots. \textbf{Terminal violation} indicates a submitted but invalid terminal state. \textbf{Completion} indicates that all required roots were closed before the round limit. The certified bound is
\[
2\max(p_{\min}-\epsilon,0)\,[\cq_{D,\epsilon}(X)-1/2]_+.
\]
A single terminal realization below this expectation bound is not a theorem violation; the appropriate empirical comparison is between conditional-performance means and the bound within quality bins.

\subsection{Measurement Reconstruction}
\label{app:measurement-revised}

\input{results/tables_core_revised/measurement_results_revised}

Table~\ref{tab:core-measurement-revised} gives all proposal-probe results. FULL-state reconstruction is exact in both modes. TASK reconstructs assignment defects but cannot validate certificate-backed work assertions. INT detects many event-visible soundness failures but has no access to capability comparisons. JUDGE is noisy on both components. MAS predicts no defects and is included as a zero-information reference.

\subsection{Quality--Performance Relationship}
\label{app:quality-performance-revised}

\begin{figure}[t]
    \centering
    \includegraphics[width=\linewidth]{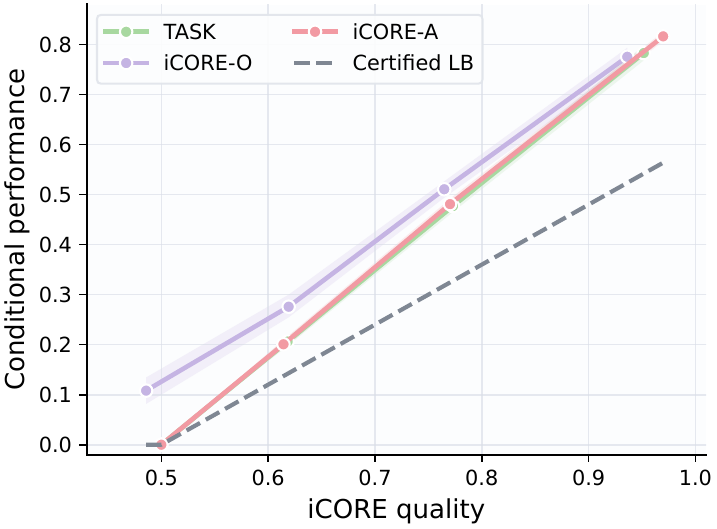}
    \caption{Controlled conditional performance by iCORE-quality bin. Shaded regions are bootstrap $95\%$ confidence intervals; the dashed curve is the mean certified lower bound.}
    \label{fig:core-quality-performance-revised}
\end{figure}

The controlled figure uses trajectory states from TASK, iCORE-O, and iCORE-A. At the zero-bound quality bin, both the performance estimate and bound are zero for iCORE-A. For every reported bin with positive certified bound, the lower endpoint of the performance interval remains above that bound. This is the assumption-aligned empirical check of Proposition~\ref{prop:quality-performance}; terminal performance is analyzed separately.

\subsection{Environment-Level Paired Gains}
\label{app:environment-gain-revised}

\begin{figure}[t]
    \centering
    \includegraphics[width=\linewidth]{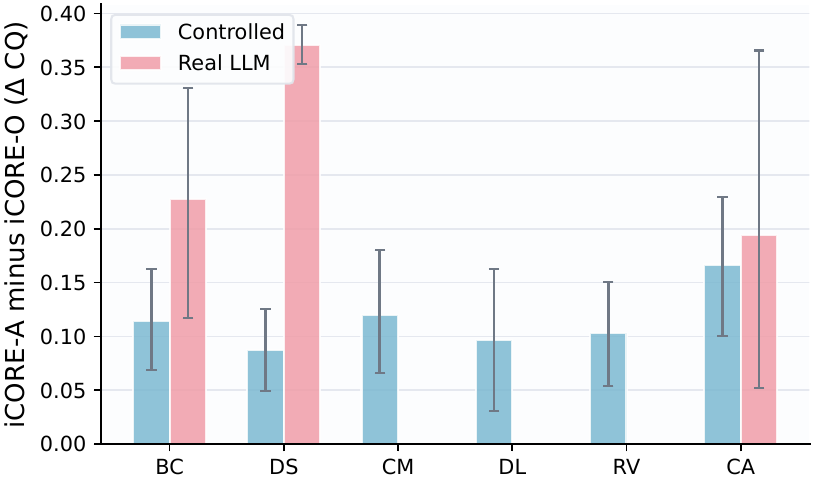}
    \caption{Paired trajectory-CQ gain of iCORE-Audit over iCORE-Observe by environment. Error bars are bootstrap $95\%$ confidence intervals. Real-LLM bars exist only for BC, DS, and CA.}
    \label{fig:core-environment-gain-revised}
\end{figure}

The mean controlled gain is positive in all six environments: $0.114$ (BC), $0.087$ (DS), $0.120$ (CM), $0.097$ (DL), $0.103$ (RV), and $0.167$ (CA). Real-LLM gains are $0.228$ (BC), $0.371$ (DS), and $0.194$ (CA). These are paired process-quality effects and do not imply that every environment reaches a valid terminal submission.

\subsection{Ablations}
\label{app:ablation-revised}

\input{results/tables_core_revised/ablation_results_revised}

\begin{figure}[t]
    \centering
    \includegraphics[width=\linewidth]{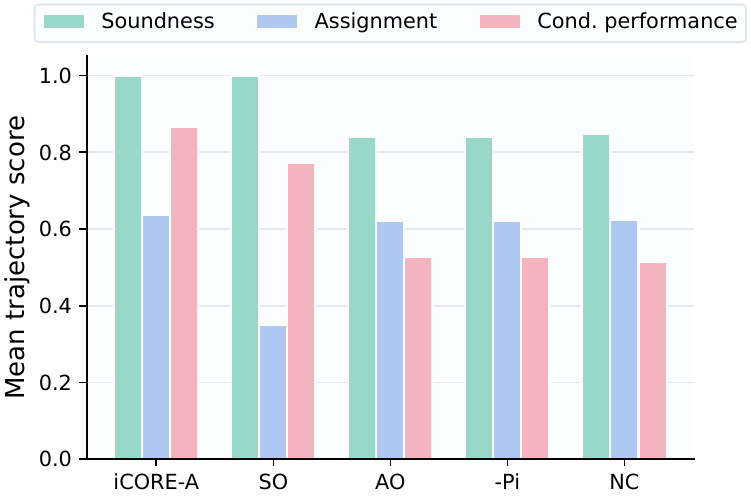}
    \caption{Controlled mixed-condition ablations. SO retains only soundness diagnoses, AO retains only beneficial-reassignment diagnoses, $-\Pi$ removes ledger visibility, and NC disables certificate validation in review.}
    \label{fig:core-ablation-components-revised}
\end{figure}

Full iCORE-A achieves soundness $1.000$, assignment stability $0.637$, CQ $0.818$, conditional performance $0.865$, and terminal performance $1.000$. SO preserves soundness but reduces assignment stability to $0.350$. AO preserves more assignment stability ($0.621$) but allows soundness to fall to $0.840$. The $-\Pi$ and NC variants reduce CQ to $0.730$ and $0.735$ and terminal performance to $0.694$ and $0.657$. AO and $-\Pi$ have identical aggregate scores in this matrix; this tie should not be presented as evidence that the two mechanisms are equivalent beyond the current generator.

\subsection{Boundary of the Performance Claim}
\label{app:boundary-revised}

\input{results/tables_core_revised/boundary_results_revised}

\begin{figure}[t]
    \centering
    \includegraphics[width=0.88\linewidth]{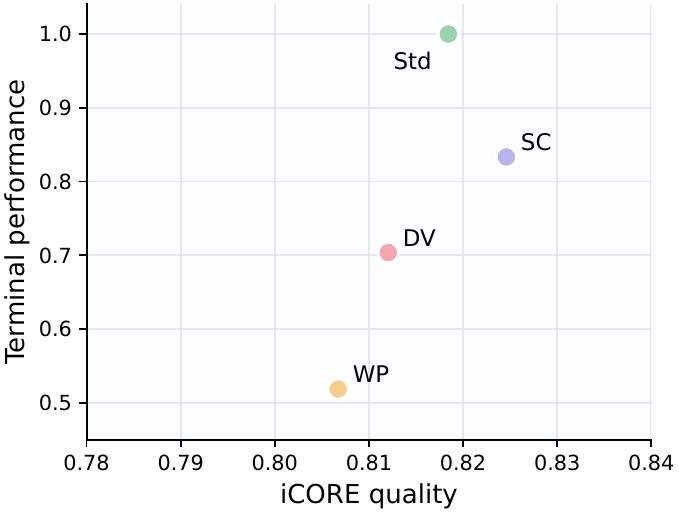}
    \caption{iCORE-Audit quality and realized terminal performance under standard and theory-boundary conditions. DV, SC, and WP denote degraded validator, shuffled declared capabilities, and weak pool.}
    \label{fig:core-boundary-stress-revised}
\end{figure}

The boundary suite illustrates why iCORE quality requires the proposition's external adequacy assumptions. DV leaves CQ near the standard value ($0.812$ versus $0.818$) while terminal performance falls from $1.000$ to $0.704$, because invalid evidence can be accepted. WP similarly retains CQ $0.807$ while terminal performance falls to $0.519$, because all available agents are weak. SC produces CQ $0.825$ and terminal performance $0.833$ by decoupling declared and true capabilities. These are expected boundary failures, not counterexamples to a claim made without those assumptions.

\subsection{Execution Cost}
\label{app:cost-revised}

\input{results/tables_core_revised/runtime_results_revised}

\begin{figure}[t]
    \centering
    \includegraphics[width=\linewidth]{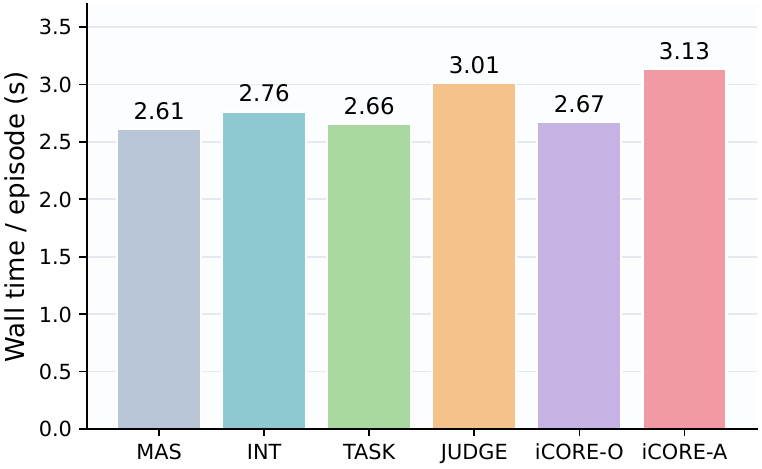}
    \caption{Mean real-LLM wall time per episode. The shared local server, model, and decoding configuration are fixed across systems.}
    \label{fig:core-real-runtime-revised}
\end{figure}

iCORE-A requires $0.026$ seconds per controlled episode versus $0.018$ for iCORE-O, an approximately $49\%$ relative increase but only $0.009$ seconds absolute. In the local real-LLM run, the corresponding values are $3.131$ and $2.672$ seconds, an approximately $17\%$ increase. iCORE-A makes more model decisions because recovery can extend a trace, but successful structured repair reduces mean verifier calls from $8.22$ to $4.33$ and certificates from $12.11$ to $9.67$ relative to iCORE-O. Wall time depends on the local server, compilation cache, and hardware and is therefore reported as an implementation cost rather than a general latency claim.

\subsection{Statistical Tests}
\label{app:statistics-revised}

\input{results/tables_core_revised/statistical_tests_revised}

Each target--baseline comparison uses the same materialized episode keys across systems and is paired on environment, condition, and boundary. Controlled tests have $90$ pairs; real-LLM tests have $18$. We use the one-sided alternative that iCORE-A is greater and Holm-adjust all $15$ comparisons within each mode. The compact table reports the smallest mean gain and largest adjusted $p$-value across the five baselines for each metric. The full baseline-specific CSV is included with the generated package.

\subsection{Reproducibility and Limitations}
\label{app:limitations-revised}

The artifact stores complete per-episode JSON objects containing $G$, $Q$, $\Pi$, trajectories, proposal probes, predictions, interventions, metrics, and costs. Figure source CSVs and table source CSVs are included. The plotting and table scripts supplied with the artifact write to new output directories and never overwrite the original experiment code or results.

Two limitations constrain interpretation. First, the real-LLM test covers only BC, DS, and CA with a $0.5$B model, so its scope across task families and model scales remains limited. Second, the LLM selects from environment-generated candidate actions whose structured fields include task values and evidence templates; it does not independently solve the underlying numerical or semantic task. These limitations are compatible with the paper's narrower claim: a coupled representation makes the target process properties measurable, and state-grounded auditing can use that measurement to guard and repair graph-structured cooperation.

\section{Proofs}
\label{app:proofs}

\subsection{Proof of Theorem~\ref{thm:projection-insufficiency}} 
\label{app:projection-insufficiency-proof} 

\begin{lemma}[Indistinguishability on a projection fiber]
\label{lem:projection-fiber-indistinguishability}
Let $\mathcal{C}$ be a class of states, let
\[
    \varrho:\mathcal{C}\rightarrow\mathcal{Y}
\]
be an observation map, and let
\[
    \mathsf{P}:\mathcal{C}\rightarrow\{0,1\}
\]
be a Boolean property. Suppose that there exist $X^{+},X^{-}\in\mathcal{C}$ such that
\[
    \varrho(X^{+})=\varrho(X^{-}),
    \qquad
    \mathsf{P}(X^{+})=1,
    \qquad
    \mathsf{P}(X^{-})=0.
\]
Then no monitor whose decision depends only on $\varrho(X)$ can be both sound and complete for $\mathsf{P}$ on $\mathcal{C}$.
\end{lemma}

\begin{proof}
A monitor restricted to the observation $\varrho(X)$ is a map
\[
    M:\mathcal{C}\rightarrow\{0,1\}
\]
for which there exists a decision rule
\[
    \widehat{M}:\mathcal{Y}\rightarrow\{0,1\}
\]
satisfying
\[
    M(X)=\widehat{M}\bigl(\varrho(X)\bigr)
\]
for every $X\in\mathcal{C}$. We use the convention that $M(X)=1$ means that the monitor declares that $\mathsf{P}$ holds at $X$. Soundness therefore requires
\[
    M(X)=1
    \quad\Longrightarrow\quad
    \mathsf{P}(X)=1
\]
for every $X\in\mathcal{C}$, whereas completeness requires
\[
    \mathsf{P}(X)=1
    \quad\Longrightarrow\quad
    M(X)=1
\]
for every $X\in\mathcal{C}$.

Assume, for contradiction, that such a monitor is both sound and complete. Since
\[
    \varrho(X^{+})=\varrho(X^{-}),
\]
the restriction of $M$ to $\varrho$ gives
\[
\begin{aligned}
    M(X^{+})
    &=
    \widehat{M}\bigl(\varrho(X^{+})\bigr)\\
    &=
    \widehat{M}\bigl(\varrho(X^{-})\bigr)\\
    &=
    M(X^{-}).
\end{aligned}
\]
Thus the monitor must return the same Boolean value on $X^{+}$ and $X^{-}$. Denote this common value by $b\in\{0,1\}$.

If $b=1$, then
\[
    M(X^{-})=1.
\]
Soundness would imply
\[
    \mathsf{P}(X^{-})=1,
\]
contradicting the assumption
\[
    \mathsf{P}(X^{-})=0.
\]
If instead $b=0$, then
\[
    M(X^{+})=0.
\]
However, since
\[
    \mathsf{P}(X^{+})=1,
\]
completeness requires
\[
    M(X^{+})=1,
\]
which is again a contradiction. Since the common output must be either zero or one and both possibilities contradict one of the two required guarantees, no monitor depending only on $\varrho(X)$ can be both sound and complete for $\mathsf{P}$.
\end{proof}

\begin{proof}[Proof of Theorem~\ref{thm:projection-insufficiency}]
Let $\mathcal{C}$ be the class of iCORE states appearing in the theorem. For a iCORE state
\[
    X=(G,Q,\Pi),
\]
with
\[
    G=
    (V_G,E_G,s_G,t_G,\theta_G,\lambda_G,\tau_G,\mu_G),
\]
define its unlabeled cooperation-graph topology by
\[
    \operatorname{Top}(X)
    =
    \left[
    (V_G,E_G,s_G,t_G)
    \right]_{\cong},
\]
where $[\cdot]_{\cong}$ denotes the directed-graph isomorphism class obtained after discarding vertex and edge identifiers together with all types, observable labels, logical times, metadata, agent identities, capability evidence, verifier information, and monitor annotations. In particular, $\operatorname{Top}(X)$ retains only the incidence structure of the cooperation graph and contains no information from the obligation graph $Q$ or the audit map $\Pi$.

By assumption, $\mathcal{C}$ contains two states with identical unlabeled cooperation-graph topology for which one of the properties named in the theorem has different truth values. Choose such states and denote them by
\[
    X^{+}=(G^{+},Q^{+},\Pi^{+})
    \qquad\text{and}\qquad
    X^{-}=(G^{-},Q^{-},\Pi^{-}),
\]
where $X^{+}$ satisfies the differing property and $X^{-}$ does not. The assumed identity of their unlabeled topologies means
\[
    \operatorname{Top}(X^{+})
    =
    \operatorname{Top}(X^{-}).
\]

If the differing property is work soundness, define
\[
    \mathsf{P}_{\mathrm{ws}}(X)
    =
    \mathbf{1}
    \bigl[
        X\text{ is work-sound}
    \bigr].
\]
By Definition~\ref{def:work-soundness}, this is equivalently
\[
    \mathsf{P}_{\mathrm{ws}}(X)
    =
    \mathbf{1}
    \bigl[
        \snd_D(X)=1
    \bigr].
\]
The choice of $X^{+}$ and $X^{-}$ then gives
\[
    \mathsf{P}_{\mathrm{ws}}(X^{+})=1
    \qquad\text{and}\qquad
    \mathsf{P}_{\mathrm{ws}}(X^{-})=0.
\]
Consequently, the property $\mathsf{P}_{\mathrm{ws}}$ is not constant on the fiber
\[
    \operatorname{Top}^{-1}
    \bigl(
        \operatorname{Top}(X^{+})
    \bigr).
\]

This difference cannot be recovered from the unlabeled topology. Indeed, work soundness requires that, for every decision-relevant component
$q\in\WorkSet(X)$, every active decision-relevant assertion associated with $q$ be justified by observable interaction, a declared work operation, and evidence valid under $V$. In terms of the state representation, evaluating the indicator
\[
    J_D(q;X)
\]
requires the semantic status, obligation, coverage, and support information stored in $Q$, together with the event--work correspondence $\pi$, operation label $\omega$, and certificate information $\kappa$ stored in
\[
    \Pi=(\pi,\omega,\kappa).
\]
Accordingly,
\[
    \snd_D(X)
    =
    \frac{
        \sum_{q\in\WorkSet(X)}
        w_q J_D(q;X)
    }{
        \sum_{q\in\WorkSet(X)}w_q
    }
\]
when $\WorkSet(X)\neq\varnothing$ depends on information that is absent from $\operatorname{Top}(X)$. Two states can therefore have the same directed incidence pattern in $G$ while differing in whether a work assertion is linked to the correct event, whether the corresponding operation is declared, whether the required certificate exists, whether that certificate is valid and active, or whether the asserted work is an actual decision-relevant obligation. The hypothesis of the theorem selects precisely such a pair for which these erased distinctions change the truth value of work soundness.

If the differing property is instead agent-assignment stability, define
\[
    \mathsf{P}_{\mathrm{as}}(X)
    =
    \mathbf{1}
    \bigl[
        X\text{ is agent-assignment stable}
    \bigr].
\]
By Definition~\ref{def:assignment-stability}, this is equivalently
\[
    \mathsf{P}_{\mathrm{as}}(X)
    =
    \mathbf{1}
    \bigl[
        \astab_{D,\epsilon}(X)=1
    \bigr].
\]
The selected states satisfy
\[
    \mathsf{P}_{\mathrm{as}}(X^{+})=1
    \qquad\text{and}\qquad
    \mathsf{P}_{\mathrm{as}}(X^{-})=0.
\]
Hence $\mathsf{P}_{\mathrm{as}}$ is likewise not constant on a fiber of $\operatorname{Top}$.

More explicitly, for each $q\in\WorkSet(X)$, assignment stability requires evaluation of
\[
    \max_{m\in\mathcal{F}_D(q,X)}
    \Delta_D
    \bigl(
        m\leftarrow\alpha_Q(q);
        q,X
    \bigr)
    \leq\epsilon,
\]
where
\[
\begin{aligned}
    \Delta_D
    \bigl(
        m\leftarrow\alpha_Q(q);
        q,X
    \bigr)
    &=
    v_D(m,q\mid X)
    -
    v_D(\alpha_Q(q),q\mid X).
\end{aligned}
\]
Determining the left-hand side requires the work component $q$, its current or recorded assignee $\alpha_Q(q)$, the feasible-agent set $\mathcal{F}_D(q,X)$, the declared contribution values $v_D(m,q\mid X)$, and the observable capability evidence and assignment correspondence validating those quantities. The assignee is stored in $Q$, while the agreement among the responsible event in $G$, the affected work in $Q$, the assignment operation, and its supporting evidence is represented through $\Pi$. None of these quantities is determined by the unlabeled incidence structure of $G$. Thus identical unlabeled cooperation topology is compatible with different values of
\[
    \max_{m\in\mathcal{F}_D(q,X)}
    \left[
        v_D(m,q\mid X)
        -
        v_D(\alpha_Q(q),q\mid X)
    \right]
\]
and therefore with different truth values of agent-assignment stability, exactly as assumed by the theorem.

Now let $M$ be any monitor restricted to the unlabeled topology of $G$. By the meaning of this restriction, there must exist a function $\widehat{M}$ such that
\[
    M(X)
    =
    \widehat{M}
    \bigl(
        \operatorname{Top}(X)
    \bigr)
\]
for every $X\in\mathcal{C}$. Therefore,
\[
\begin{aligned}
    M(X^{+})
    &=
    \widehat{M}
    \bigl(
        \operatorname{Top}(X^{+})
    \bigr)\\
    &=
    \widehat{M}
    \bigl(
        \operatorname{Top}(X^{-})
    \bigr)\\
    &=
    M(X^{-}).
\end{aligned}
\]
For whichever property differs between $X^{+}$ and $X^{-}$, take
\[
    \mathsf{P}
    =
    \mathsf{P}_{\mathrm{ws}}
    \qquad\text{or}\qquad
    \mathsf{P}
    =
    \mathsf{P}_{\mathrm{as}},
\]
respectively. We then have
\[
    \operatorname{Top}(X^{+})
    =
    \operatorname{Top}(X^{-}),
    \qquad
    \mathsf{P}(X^{+})=1,
    \qquad
    \mathsf{P}(X^{-})=0.
\]
Applying Lemma~\ref{lem:projection-fiber-indistinguishability} with
\[
    \varrho=\operatorname{Top}
\]
shows that no such monitor can be both sound and complete for $\mathsf{P}$ on $\mathcal{C}$.

The same conclusion also holds for a randomized topology-restricted monitor under zero-error soundness and completeness. If
\[
    a(T)
    =
    \Pr\!\left[
        M\text{ declares that }\mathsf{P}\text{ holds}
        \,\middle|\,
        \operatorname{Top}(X)=T
    \right],
\]
then the common topology
\[
    T=
    \operatorname{Top}(X^{+})
    =
    \operatorname{Top}(X^{-})
\]
forces the same acceptance probability on both states. Completeness at $X^{+}$ requires $a(T)=1$, while soundness at $X^{-}$ requires $a(T)=0$, which is impossible. Hence neither deterministic nor zero-error randomized access to the unlabeled cooperation-graph topology suffices to obtain both guarantees. This proves the theorem.
\end{proof}

\subsection{Proof of Proposition~\ref{prop:local-work-soundness}} 
\label{app:work-soundness-proofs} 

For a iCORE state $X=(G,Q,\Pi)$, denote by $\mathsf{A}_D(X)$ the set of active decision-relevant work-level assertions in $Q$. For $a\in\mathsf{A}_D(X)$, write $\operatorname{Just}_D(a;X)=1$ when there exists a finite justification for $a$ consisting of observable interaction recorded in $G$, the corresponding operation labels from the declared vocabulary $\Omega$, the event--work correspondence recorded by $\Pi$, and certificates that are active and valid under $V$ in state $X$. With this notation, Definition~\ref{def:work-soundness} is equivalently
\[
\begin{aligned}
    X\text{ is work-sound}
    \quad\Longleftrightarrow\quad
    &\operatorname{Just}_D(a;X)=1\\
    &\text{for every }a\in\mathsf{A}_D(X).
\end{aligned}
\]

\begin{lemma}[One-step preservation of work soundness]
\label{lem:one-step-work-soundness}
If $X$ is work-sound and $X\rightarrow X'$ is a legal iCORE update, then $X'$ is work-sound.
\end{lemma}

\begin{proof}
Fix an arbitrary assertion
\[
    a\in\mathsf{A}_D(X').
\]
We prove that
\[
    \operatorname{Just}_D(a;X')=1.
\]
First suppose that $a$ was already active and decision-relevant in $X$ and that the update changes neither the semantic content of $a$, its active status, its decision relevance, nor any semantic object used by one of its valid justification witnesses. Since $X$ is work-sound, Definition~\ref{def:work-soundness} gives
\[
    \operatorname{Just}_D(a;X)=1.
\]
Consequently, there exist a finite set of observable actions or artifacts
\[
    H_a\subseteq V_G,
\]
a finite set of declared operation labels
\[
    O_a\subseteq\Omega,
\]
the corresponding event--work links represented in $\Pi$, and a finite set of certificates
\[
    C_a\subseteq\kappa
\]
such that these objects jointly justify $a$ in $X$. More explicitly, the nodes in $H_a$ record the observable actions or artifacts responsible for the assertion, the labels in $O_a$ identify the declared work operations by which those observable events affect the relevant work components, the corresponding entries of $\Pi$ connect those events and operations to $a$, and every certificate $c\in C_a$ satisfies
\[
    \operatorname{ValidCert}_V(c;X)=1.
\]

A legal update reads and writes only through its declared interfaces. By the present case assumption, none of the semantic objects belonging to this justification is changed by the update. Therefore every node in $H_a$, every operation label in $O_a$, every relevant event--work correspondence in $\Pi$, and every certificate in $C_a$ remains present in $X'$ with the same semantic meaning. Moreover, every $c\in C_a$ remains active, its internal premises remain present and current, its external premises continue to satisfy $V$, and its claim continues to pass the declared validation procedure. If any of these activity, currentness, or validation conditions changed, then the support of $a$ would be affected and the assertion would belong to the complementary case considered below. Hence
\[
    \operatorname{ValidCert}_V(c;X')=1
    \qquad
    \text{for every }c\in C_a.
\]
The same finite collection
\[
    (H_a,O_a,C_a)
\]
together with the unchanged correspondence entries in $\Pi$ is therefore a valid justification witness for $a$ in $X'$. It follows that
\[
    \operatorname{Just}_D(a;X')=1.
\]

It remains to consider the complementary case, in which $a$ is newly created, newly activated, newly decision-relevant, semantically modified, or affected because some part of its former justification is written, revoked, superseded, invalidated, or replaced by the update. Suppose first that $a$ is newly created, newly activated, newly decision-relevant, or semantically modified. A work-level assertion includes an asserted node, relation, status, signature, coverage fact, compatibility fact, or terminal marker. Thus creating such an assertion, changing its semantic value, or making it newly active and decision-relevant produces a new decision-relevant assertion in the successor state. By Definition~\ref{def:legal-core-update}, legality requires the update to support every such assertion using an operation permitted by $D$ and evidence valid under $V$. Therefore the update must provide valid supporting evidence for $a$.

Legality additionally requires
\[
    X'\in\mathcal{X}_D.
\]
By the feasibility conditions defining $\mathcal{X}_D$, every non-initial work assertion must be backed by responsible observable nodes in $G$, corresponding operation labels, and active valid certificates through $\Pi$, while every required initial root must remain connected to observable initial input and supporting evidence. Consequently, the support required by legality cannot consist only of an unrecorded or external assertion. It must be represented in the successor state by observable interaction in $G$, an operation from $\Omega$, the corresponding event--work links in $\Pi$, and certificates valid under $V$. Hence
\[
    \operatorname{Just}_D(a;X')=1.
\]

Now suppose that $a$ existed before the update but a certificate, correspondence entry, operation label, premise, or other semantic object in its former justification is revoked, superseded, invalidated, replaced, or otherwise made noncurrent. The former witness can no longer be used automatically, because the definition of certificate validity requires all certificates and premises in the witness to remain active and current. Nevertheless, legality still requires
\[
    X'\in\mathcal{X}_D.
\]
The feasibility conditions defining $\mathcal{X}_D$ require that revoked or superseded evidence cannot continue to support an active assertion without re-certification, repair, invalidation, rejection, or supersession. They also prohibit an active assertion from depending simultaneously on a certificate and an active valid revocation of that certificate. Therefore, after the update, precisely one of the following must occur: the old assertion is made observably inactive by a valid invalidation, rejection, or supersession record, or the assertion remains active and is supplied with replacement, repaired, or re-certified support that is active and valid in $X'$.

The first alternative is impossible for the fixed assertion $a$, because
\[
    a\in\mathsf{A}_D(X')
\]
means that $a$ is active and decision-relevant in the successor state. Hence the second alternative must hold. There therefore exist responsible observable nodes in $G'$, corresponding declared operation labels, event--work links in $\Pi'$, and active certificates valid under $V$ that jointly support $a$ in $X'$. Thus
\[
    \operatorname{Just}_D(a;X')=1.
\]

The assertion $a\in\mathsf{A}_D(X')$ was arbitrary. We have proved
\[
    \forall a\in\mathsf{A}_D(X'),
    \qquad
    \operatorname{Just}_D(a;X')=1.
\]
Definition~\ref{def:work-soundness} therefore implies that $X'$ is work-sound.
\end{proof}

\begin{proof}[Proof of Proposition~\ref{prop:local-work-soundness}]
Let
\[
    X_0\rightarrow X_1\rightarrow X_2\rightarrow\cdots
\]
be an arbitrary admissible process. By the definition of admissibility, every transition
\[
    X_n\rightarrow X_{n+1}
\]
is a legal iCORE update. We prove by induction on $n\in\mathbb{N}$ that every state $X_n$ is work-sound.

The base case is immediate from the hypothesis:
\[
    X_0\text{ is work-sound}.
\]
Now let $n\geq 0$ and assume that
\[
    X_n\text{ is work-sound}.
\]
Since the process is admissible, the transition
\[
    X_n\rightarrow X_{n+1}
\]
is legal. Applying Lemma~\ref{lem:one-step-work-soundness} with
\[
    X=X_n
    \qquad\text{and}\qquad
    X'=X_{n+1}
\]
gives
\[
    X_{n+1}\text{ is work-sound}.
\]
The induction therefore yields
\[
    \forall n\in\mathbb{N},
    \qquad
    X_n\text{ is work-sound}.
\]

Equivalently, for every $n$ and every component
\[
    q\in\WorkSet(X_n),
\]
all active decision-relevant assertions associated with $q$ are justified. Hence
\[
    J_D(q;X_n)=1
    \qquad
    \text{for every }q\in\WorkSet(X_n).
\]
When
\[
    \WorkSet(X_n)=\varnothing,
\]
the definition of the work-soundness score directly gives
\[
    \snd_D(X_n)=1.
\]
When
\[
    \WorkSet(X_n)\neq\varnothing,
\]
we obtain
\[
\begin{aligned}
    \snd_D(X_n)
    &=
    \frac{
        \sum_{q\in\WorkSet(X_n)}
        w_qJ_D(q;X_n)
    }{
        \sum_{q\in\WorkSet(X_n)}
        w_q
    }\\
    &=
    \frac{
        \sum_{q\in\WorkSet(X_n)}
        w_q\cdot 1
    }{
        \sum_{q\in\WorkSet(X_n)}
        w_q
    }\\
    &=
    \frac{
        \sum_{q\in\WorkSet(X_n)}
        w_q
    }{
        \sum_{q\in\WorkSet(X_n)}
        w_q
    }\\
    &=1.
\end{aligned}
\]
The denominator is strictly positive because the index set is nonempty and every predeclared importance weight satisfies
\[
    w_q>0.
\]
Thus every state occurring in the admissible process is work-sound.

Finally, let $Y$ be any state reached from $X_0$ by an admissible process. By the definition of reachability, there exists a finite admissible prefix
\[
    X_0\rightarrow X_1\rightarrow\cdots\rightarrow X_n=Y
\]
for some $n\in\mathbb{N}$. The induction result gives that $X_n$ is work-sound, and therefore
\[
    Y\text{ is work-sound}.
\]
The same reasoning applies when the complete admissible process is infinite, because every state reached along an infinite process occurs at some finite index. Hence every state reached by an admissible process from $X_0$ is work-sound.
\end{proof}

\subsection{Proof of Proposition~\ref{prop:audit-reconstruction}} \label{app:audit-reconstruction-proof} 

For an admissible process
\[
    X_0\xrightarrow{\rho_0}X_1
    \xrightarrow{\rho_1}X_2
    \xrightarrow{\rho_2}\cdots,
    \qquad
    X_i=(G_i,Q_i,\Pi_i),
\]
write the ordered audit entry of the transition
$X_i\xrightarrow{\rho_i}X_{i+1}$ as
\[
    \mathfrak r_i
    =
    \bigl(
        \rho_i,
        \Delta_i^G,
        \Delta_i^Q,
        \Delta_i^{\Pi}
    \bigr),
\]
where $\rho_i$ identifies the legal update occurrence and
$\Delta_i^G$, $\Delta_i^Q$, and $\Delta_i^{\Pi}$ record all of its
semantic effects on the three components of the iCORE state. Thus these
effects include every created or removed node and edge, every changed
semantic label, status, signature, or assignment, every created,
revoked, or superseded certificate, and every changed event--work or
operation correspondence. Fields declared inessential by $D$, such as
a choice of fresh identifiers or an inessential timestamp, need not be
reproduced literally. The order of the entries is the commit order of
the process, equivalently the order in which the corresponding state
transitions occur.

\begin{lemma}[Replay uniqueness modulo declared equivalence]
\label{lem:audit-replay-uniqueness}
Let $\mathfrak R=(\mathfrak r_0,\mathfrak r_1,\ldots)$ be the ordered
audit record of an admissible process from the fixed initial state
$X_0$. Any replay of $\mathfrak R$ from a representative
$\widehat X_0\equiv_X X_0$ produces states $\widehat X_i$ satisfying
\[
    \widehat X_i\equiv_X X_i
\]
for every recorded index $i$. Consequently, any two state sequences
consistent with the same record are pointwise equivalent under
$\equiv_X$.
\end{lemma}

\begin{proof}
We argue by induction on the number of replayed entries. At index zero,
the replay is initialized with a representative satisfying
\[
    \widehat X_0\equiv_X X_0,
\]
so the claim holds before any update is replayed. Suppose that, for some
$i\geq 0$, the replay has already produced a state $\widehat X_i$ for
which
\[
    \widehat X_i\equiv_X X_i.
\]
By the definition of the declared state equivalence, there are
transported type-preserving graph isomorphisms between the $G$- and
$Q$-components of these states, together with the induced transport of
$\Pi$, and these transports preserve every field that may affect update
enablement, feasibility, obligations, assignments, capability evidence,
certificate validity, or terminal decisions. Transport the recorded
occurrence $\rho_i$ along these isomorphisms and denote the transported
occurrence by $\widehat\rho_i$. Because the semantic read objects,
statuses, signatures, operation labels, correspondence entries, and
certificate premises are preserved, the complete local views agree up
to the same transport:
\[
    L(\widehat\rho_i;\widehat X_i)
    \cong
    L(\rho_i;X_i).
\]
The occurrence $\rho_i$ is legal at $X_i$, because the original process
is admissible. Hence every object and certificate premise required by
$\rho_i$ is present and current in $L(\rho_i;X_i)$, the operation is
permitted by $D$, every new semantic assertion has active valid
support, and the recorded successor is
\[
    X_{i+1}
    =
    \widehat R_{\rho_i}(X_i)
    \in\mathcal X_D.
\]
The transported equality of local views implies that the same
enablement and validation tests hold for $\widehat\rho_i$ at
$\widehat X_i$. Moreover, declared state equivalence is required to be
a congruence for legal updates. Therefore the transported occurrence is
legal and its successor satisfies
\[
    \widehat X_{i+1}
    =
    \widehat R_{\widehat\rho_i}(\widehat X_i)
    \equiv_X
    \widehat R_{\rho_i}(X_i)
    =
    X_{i+1}.
\]
The recorded effects
$\Delta_i^G$, $\Delta_i^Q$, and $\Delta_i^{\Pi}$ determine all semantic
writes of this occurrence. Since an update occurrence induces a partial
function on $\mathcal X_D$, a replay consistent with the entry cannot
choose a different semantic successor; it may differ only in fields
that $D$ declares inessential. This proves the induction step, and thus
\[
    \widehat X_i\equiv_X X_i
\]
for every finite index $i$. The argument applies without change to an
infinite record, because every state at a finite index is determined by
a finite prefix. Finally, if $(\widehat X_i)_i$ and $(\widetilde X_i)_i$
are two replays of the same record, then
\[
    \widehat X_i\equiv_X X_i
    \qquad\text{and}\qquad
    \widetilde X_i\equiv_X X_i,
\]
and symmetry and transitivity of $\equiv_X$ give
\[
    \widehat X_i\equiv_X\widetilde X_i
\]
for every $i$.
\end{proof}

\begin{lemma}[Finite backward explanation of terminal-cover work]
\label{lem:finite-backward-explanation}
Let
\[
    X_0\xrightarrow{\rho_0}X_1
    \xrightarrow{\rho_1}\cdots
    \xrightarrow{\rho_{N-1}}X_N
\]
be a finite admissible process with $\operatorname{Acc}(X_N)=1$, and
let $A\subseteq\operatorname{Accepted}(Q_N)$ be the certified terminal
cover used by the acceptance predicate. For every component $a\in A$,
there exists a finite explanation substructure consisting of recorded
interaction objects from the cooperation graphs, recorded work objects
and transformations from the obligation graphs, and the corresponding
restriction of the audit maps, such that the explanation contains the
accepted status and terminal-cover role of $a$, connects every
non-initial work assertion in the explanation to responsible observable
interaction and active valid evidence, and reaches an observable initial
input through the required-root coverage and support structure.
\end{lemma}

\begin{proof}
Fix $a\in A$. Since $\operatorname{Acc}(X_N)=1$, the terminal-state
conditions in Appendix~\ref{app:terminal-predicates} imply that $X_N$ is
feasible, that $a$ has status $\texttt{accepted}$, that the accepted
components in $A$ jointly cover every effective required root with the
compatibility evidence required by $V$, and that no accepted component
depends on unresolved, invalid, rejected, contradictory, or unrepaired
work. Let $\mathsf S_N(a)$ be the finite set containing the assertion
\[
    \operatorname{stat}_{Q_N}(a)=\texttt{accepted},
\]
the active result and signature assertions of $a$ used by the terminal
decision, the coverage assertions by which $a$ participates in the
certified terminal cover, and the compatibility assertions involving
$a$ that are used by the joint terminal-cover certificate. This set is
finite because $Q_N$, $\Pi_N$, and every certificate record are finite.

We first expose the work-side connection to the input. Because $a$
belongs to the certified terminal cover, its terminal-cover evidence
identifies the effective required root or roots whose discharge uses
$a$. For each such root $r$, the coverage semantics provide a finite
coverage path in $Q_N$ from $r$ to $a$; schematically,
\[
    r=q_0
    \preceq_{Q_N}^{\mathrm{cov}}
    q_1
    \preceq_{Q_N}^{\mathrm{cov}}
    \cdots
    \preceq_{Q_N}^{\mathrm{cov}}
    q_{\ell}=a.
\]
If an intermediate work component occurring in the selected terminal
support is not itself a required root, feasibility supplies a finite
support chain ending at a required root or a declared input node. We add
the nodes, semantic relations, statuses, signatures, coverage facts,
and support facts on these finitely many selected chains to
$\mathsf S_N(a)$. Every required root reached in this way is connected,
by feasibility, to an observable initial-input node and a supporting
certificate. Any auxiliary external premise is not treated as hidden
support: by Appendix~\ref{app:certificates}, it enters only through a
declared environment-input node in $G$ and an active valid certificate.
Thus the work-side support of $a$ is anchored at the declared observable
input boundary and, through the required-root branch of the terminal
cover, at the initial input.

It remains to show that every assertion on these work chains can be
expanded into a finite recorded cross-layer explanation. We construct a
backward audit slice. At a generic transition
\[
    X_i\xrightarrow{\rho_i}X_{i+1},
\]
consider a finite set $\mathsf S_{i+1}$ of semantic objects and
assertions in $X_{i+1}$ that are currently required by the explanation.
Partition it as
\[
    \mathsf S_{i+1}
    =
    \mathsf U_i\cup\mathsf W_i,
\]
where $\mathsf U_i$ contains the objects that were already present in
$X_i$ with the same semantic value and remained unchanged by
$\rho_i$, while $\mathsf W_i$ contains the objects created, activated,
semantically modified, re-certified, repaired, or otherwise supplied
with their current support by the transition. Objects in
$\mathsf U_i$ can simply be carried to the preceding state. For an
object or assertion $z\in\mathsf W_i$, the update is legal, so the
current semantic value of $z$ cannot have been introduced without the
support required by $D$ and $V$. In particular, the relevant part of
the recorded update contains observable action or artifact nodes
$g\in V_{G_{i+1}}$, a declared operation label
$\omega_{i+1}(g)\in\Omega$, an event--work correspondence
\[
    q\in\pi_{i+1}(g),
\]
and certificates in $\kappa_{i+1}$ that are active and valid for the
claim being introduced or maintained. For every such certificate $c$,
\[
    \operatorname{ValidCert}_V(c;X_{i+1})=1.
\]
We add these recorded $G$-objects, $Q$-effects, operation labels,
correspondence entries, and certificates to the explanation. We then
place into $\mathsf S_i$ the objects and premises from the pre-update
local view $L(\rho_i;X_i)$ that are needed to enable and validate this
relevant part of the update. Formally, it is sufficient to set
\[
    \mathsf S_i
    =
    \mathsf U_i
    \cup
    \bigcup_{z\in\mathsf W_i}
    \operatorname{Prem}_i(z),
\]
where $\operatorname{Prem}_i(z)$ denotes the subset of
$L(\rho_i;X_i)$ used by the declared operation and validators to support
$z$. If the implementation does not separately mark the minimal subset,
we may take
\[
    \operatorname{Prem}_i(z)=L(\rho_i;X_i),
\]
which is still finite by the finite-interface and finite-ancestry
requirements.

This construction also covers changes of evidence. If support formerly
used by a required assertion is revoked, superseded, or made noncurrent,
feasibility prohibits that assertion from remaining active unless it is
re-certified, repaired, invalidated, rejected, or superseded. Since the
assertions in the selected terminal explanation remain active and
support an accepted component, the invalidation or rejection alternative
cannot be the terminal support used for $a$. The backward slice therefore
follows the recorded replacement, repair, re-certification, or valid
supersession chain. Internal certificate premises are added to the slice;
external premises terminate at their declared environment-input nodes.
The active certificate-dependency graph is acyclic, so expanding the
certificates occurring within one state or one update cannot create an
infinite regress. Likewise, the cooperation graph is a finite directed
acyclic graph, so adding the causal ancestors of a responsible event adds
only finitely many observable actions and artifacts.

Starting with the finite seed $\mathsf S_N(a)$ and applying the preceding
construction successively for
\[
    i=N-1,N-2,\ldots,0
\]
produces finite sets $\mathsf S_i$ at every stage. The update index
strictly decreases whenever the construction asks for pre-update
premises, so after exactly $N$ backward transitions no further process
history remains to be traversed. The remaining objects lie in the fixed
initial state $X_0$. Feasibility of $X_0$ connects every required root to
an observable initial input and gives every non-root initial work object
a finite support chain ending at such a root or at a declared input
node. Adding these initial anchors closes the explanation.

Let $I_a\subseteq\{0,1,\ldots,N-1\}$ be the set of update indices whose
entries were included by the backward slice, and let $\mathcal E_a$ be
the union of the selected initial objects, selected local-view premises,
selected update effects, selected causal ancestors in $G$, selected work
chains in $Q$, and selected restrictions of $\Pi$ and $\kappa$. Its size
is bounded by
\[
    |\mathcal E_a|
    \leq
    |X_0|
    +
    \sum_{i\in I_a}
    \bigl(
        |L(\rho_i;X_i)|
        +|\Delta_i^G|
        +|\Delta_i^Q|
        +|\Delta_i^{\Pi}|
    \bigr),
\]
and the right-hand side is finite because $I_a$ is finite and every
state, local interface, update effect, and certificate record is finite.
By construction, every selected work transformation is paired through
the selected audit-map entries with the observable event and declared
operation that justified it, every selected certificate premise is
represented and validated, the explanation begins at the observable
input anchors, and it ends at the accepted terminal assertions of $a$.
Thus $\mathcal E_a$ is the required finite explanation.
\end{proof}

\begin{proof}[Proof of Proposition~\ref{prop:audit-reconstruction}]
Let
\[
    X_0\xrightarrow{\rho_0}X_1
    \xrightarrow{\rho_1}X_2
    \xrightarrow{\rho_2}\cdots
\]
be the legal process represented by the record. The declaration package
fixes $X_0$, and the hypothesis supplies the ordered entry
$\mathfrak r_i$ for every committed transition. Lemma~\ref{lem:audit-replay-uniqueness}
therefore shows, by replaying the entries from $X_0$, that the record
determines every state $X_i$ up to $\equiv_X$. It also shows uniqueness:
any other state sequence compatible with the same initial state and the
same recorded updates and effects is pointwise equivalent to the
original sequence. Hence the record reconstructs the complete sequence
of iCORE states up to the declared observational equivalence. For an
infinite process, this statement is understood pointwise, since every
finite-index state is reconstructed from a finite record prefix.

Now suppose that the process is finite and accepted, so that for some
$N<\infty$ its final state satisfies
\[
    \operatorname{Acc}(X_N)=1.
\]
Let $A$ be the certified terminal cover used by the acceptance predicate
and fix an arbitrary $a\in A$. Lemma~\ref{lem:finite-backward-explanation}
constructs a finite explanation $\mathcal E_a$ whose $G$-part contains
the responsible observable interaction events and artifacts, whose
$Q$-part contains the relevant work components and transformations, and
whose $\Pi$-part records the operation, event--work, and certificate
correspondence connecting those two parts. The explanation is anchored
at the observable initial input and terminates at the accepted status,
coverage role, and compatibility support of $a$. Since $a$ was arbitrary,
every accepted component in the terminal cover has such a finite
explanation. This proves both claims.
\end{proof}

\subsection{Proof of Corollary~\ref{cor:terminal-report-soundness}} 
\label{app:terminal-report-soundness-proof} 

\begin{proof}
Consider the sequence of states produced by the committed updates of the execution. This sequence has the form
\[
    X_0
    \xrightarrow{\rho_0}
    X_1
    \xrightarrow{\rho_1}
    X_2
    \xrightarrow{\rho_2}
    \cdots,
\]
where the index set of states is either
\[
    \mathcal I=\{0,1,\ldots,N\}
\]
for some finite execution, or
\[
    \mathcal I=\mathbb N
\]
for an infinite execution. For every index \(n\in\mathcal I\) such that
\(n+1\in\mathcal I\), the transition
\[
    X_n\xrightarrow{\rho_n}X_{n+1}
\]
is a committed update. By hypothesis, every committed update is legal. Hence
\[
    X_n\rightarrow X_{n+1}
    \quad\text{is a legal iCORE update}
\]
for every such \(n\). It follows directly from the definition of an admissible
process that the committed-state sequence is an admissible process beginning
at \(X_0\). Candidate updates that are rejected, blocked, aborted, or otherwise
not committed do not occur as transitions in this sequence and therefore do
not change any state reached by the process.

The initial state is work-sound by hypothesis. Proposition~\ref{prop:local-work-soundness}
states that every state reached from a work-sound initial state by an admissible
process is work-sound. Applying that proposition to the committed-state
sequence gives
\[
    X_n\text{ is work-sound}
    \qquad
    \text{for every }n\in\mathcal I.
\]
Equivalently, for every \(n\in\mathcal I\) and every active
decision-relevant work-level assertion
\[
    a\in\mathsf{A}_D(X_n),
\]
there exists a justification consisting of observable interaction recorded
in \(G_n\), declared operation labels from \(\Omega\), the corresponding
event--work relations recorded by \(\Pi_n\), and certificates that are active
and valid under \(V\) in \(X_n\). In the notation used in
Appendix~\ref{app:work-soundness-proofs},
\[
    \operatorname{Just}_D(a;X_n)=1
    \qquad
    \text{for every }
    a\in\mathsf{A}_D(X_n).
\]
Thus every state reached by the execution is work-sound, and Definition~\ref{def:work-soundness}
therefore implies that the entire committed process is work-sound.

It remains to verify the conditions governing the emitted terminal reports.
Let \(e\) be an arbitrary terminal-report occurrence in the execution, and
let \(X_{n(e)}\) denote the current committed state when \(e\) is emitted.
Such a state exists even when no update has yet been committed, in which case
\(n(e)=0\) and the current state is \(X_0\). Because \(X_{n(e)}\) belongs to
the committed-state sequence, the preceding argument gives
\[
    X_{n(e)}\text{ is work-sound}.
\]

Suppose first that the report occurrence \(e\) emits
\(\textsc{Accept}\). The reporting hypothesis states that the execution emits
\(\textsc{Accept}\) only when the acceptance predicate holds at the current
state. Therefore,
\[
    e=\textsc{Accept}
    \quad\Longrightarrow\quad
    \operatorname{Acc}\bigl(X_{n(e)}\bigr).
\]
In particular, the execution cannot emit \(\textsc{Accept}\) at a state for
which
\[
    \operatorname{Acc}\bigl(X_{n(e)}\bigr)
\]
is false. Notice that this conclusion is not inferred merely from work
soundness: work soundness justifies the active decision-relevant assertions,
whereas \(\operatorname{Acc}\) additionally requires complete settlement of
the effective required roots, a certified compatible terminal cover, the
absence of unresolved relevant work and supported failures, and the absence
of a live productive artifact that can change the outcome. These additional
terminal conditions are supplied precisely by the assumed acceptance guard.

Suppose instead that \(e\) emits \(\textsc{Reject}\). The corresponding
reporting hypothesis states that the execution emits \(\textsc{Reject}\) only
when the rejection predicate holds at the current state. Hence,
\[
    e=\textsc{Reject}
    \quad\Longrightarrow\quad
    \operatorname{Rej}\bigl(X_{n(e)}\bigr).
\]
Consequently, the execution cannot emit \(\textsc{Reject}\) at a state for
which
\[
    \operatorname{Rej}\bigl(X_{n(e)}\bigr)
\]
is false. As in the acceptance case, this implication is imposed by the
terminal-report guard rather than by work soundness alone.

The report occurrence \(e\) was arbitrary. We have therefore established
that every terminal-report occurrence satisfies
\[
    e=\textsc{Accept}
    \quad\Longrightarrow\quad
    \operatorname{Acc}\bigl(X_{n(e)}\bigr)
\]
and
\[
    e=\textsc{Reject}
    \quad\Longrightarrow\quad
    \operatorname{Rej}\bigl(X_{n(e)}\bigr).
\]
Together with the fact that every state in the committed process is
work-sound, these are exactly the conditions in
Definition~\ref{def:terminal-report-soundness}. Hence the process is
terminal-report sound. The argument applies equally to finite and infinite
executions, because every report in an infinite execution occurs after a
finite committed prefix and is therefore associated with some state
\(X_n\) covered by the preceding reasoning.
\end{proof}

\subsection{Proof of Proposition~\ref{prop:task-level-implication}} \label{app:task-implications} 

\begin{proof}
Let
\[
    \mathfrak p:
    X_0\rightarrow X_1\rightarrow\cdots\rightarrow X_N
\]
be the process under consideration. The process is finite because terminal completeness is defined for a finite process. Write
\[
    X_i=(G_i,Q_i,\Pi_i)
    \qquad (0\leq i\leq N).
\]
Since the process is work-sound, Definition~\ref{def:work-soundness} gives
\[
    X_i\text{ is work-sound}
    \qquad
    \text{for every }i\in\{0,1,\ldots,N\}.
\]
Equivalently, using the notation introduced in Appendix~\ref{app:work-soundness-proofs}, for every active decision-relevant assertion
\[
    b\in\mathsf{A}_D(X_i)
\]
we have
\[
    \operatorname{Just}_D(b;X_i)=1.
\]
Thus every accepted or rejected status, every coverage or compatibility claim, every failure claim, every exclusion, duplication, irrelevance, or supersession claim used to determine the effective required roots, and every terminal-support assertion that is active and decision-relevant at any reported state is connected through \(\Pi_i\) to observable interaction in \(G_i\), a declared operation, and evidence active and valid under \(V\). In particular, the terminal predicates cannot be supported by an unrecorded work assertion or by evidence that is inactive, revoked, superseded, noncurrent, or invalid under \(V\).

Let \(e\) be an arbitrary emitted \textsc{Accept} or \textsc{Reject} report, and let \(X_{n(e)}\) be the current committed state when \(e\) is emitted. Because \(e\) occurs after a finite prefix of \(\mathfrak p\), its index satisfies
\[
    n(e)\in\{0,1,\ldots,N\}.
\]
For readability, set
\[
    X^e=X_{n(e)},
    \qquad
    Q^e=Q_{n(e)}.
\]
The preceding work-soundness conclusion applies to \(X^e\).

Suppose first that
\[
    e=\textsc{Accept}.
\]
Terminal-report soundness implies that an \textsc{Accept} report may be emitted only at a state satisfying the declared acceptance predicate. Hence
\[
    e=\textsc{Accept}
    \quad\Longrightarrow\quad
    \operatorname{Acc}(X^e).
\]
Unfolding \(\operatorname{Acc}(X^e)\) from Appendix~\ref{app:terminal-predicates}, there exists a set
\[
    A_e\subseteq\operatorname{Accepted}(Q^e)
\]
that forms a certified terminal cover of \(\operatorname{EffRoot}(X^e)\). Therefore
\[
    \operatorname{EffRoot}(X^e)
    \subseteq
    \bigcup_{a\in A_e}\operatorname{Cov}_{Q^e}(a),
\]
and the joint compatibility evidence required by \(V\) exists for the components in \(A_e\). The other clauses of the acceptance predicate give
\[
    \operatorname{EffUnres}(X^e)=\varnothing
\]
and
\[
    \operatorname{EffRoot}(X^e)
    \cap
    \operatorname{FailRoot}(X^e)
    =\varnothing.
\]
They also ensure that no accepted component depends on unresolved, invalid, rejected, contradictory, or unrepaired work, and ensure that no live productive artifact can change the discharge, failure, exclusion, or compatibility of an effective root. Consequently, the cover is terminal rather than merely provisional: every effective required root is discharged by accepted work, no effective required root simultaneously carries an unneutralized supported failure, no unresolved relevant work remains, and no still-live productive artifact can alter the reported outcome.

It remains to verify that the terminal cover just described is genuinely certified under \(V\), rather than merely having the correct set-theoretic shape. Every assertion that a component belongs to \(\operatorname{Accepted}(Q^e)\), every active coverage assertion used in \(\operatorname{Cov}_{Q^e}(a)\), every compatibility assertion used to combine the members of \(A_e\), and every active assertion used to remove a root from
\[
    \operatorname{EffRoot}(X^e)
    =
    \operatorname{ReqRoot}(Q^e)
    \setminus
    \operatorname{Off}_{X^e}(Q^e)
\]
is decision-relevant to the \textsc{Accept} report. Since \(X^e\) is work-sound, each such assertion has a justification through observable interaction, a declared operation, the correspondence recorded by \(\Pi_{n(e)}\), and active evidence valid under \(V\). Hence no required root is silently omitted: a root is absent from \(\operatorname{EffRoot}(X^e)\) only when its duplicated, excluded, irrelevant, or superseded status is supported by active valid evidence. It follows that \(A_e\) is a \(V\)-certified accepted terminal cover in precisely the sense used by task adequacy.

By the first clause of task adequacy, every \(V\)-certified accepted terminal cover yields a semantically correct accepted output for the declared task class. Applying that clause to \(A_e\) shows that the output certified by \(A_e\) is semantically correct. Therefore the emitted decision
\[
    e=\textsc{Accept}
\]
is semantically correct for the declared task class. Notice that validator adequacy is essential here: work soundness proves that the acceptance evidence is valid according to \(V\), while task adequacy is what permits the further inference from validity under \(V\) to correctness under the task semantics.

Suppose instead that
\[
    e=\textsc{Reject}.
\]
Terminal-report soundness now gives
\[
    e=\textsc{Reject}
    \quad\Longrightarrow\quad
    \operatorname{Rej}(X^e).
\]
Define
\[
    F_e=\operatorname{EffRoot}(X^e).
\]
Unfolding the rejection predicate yields
\[
    F_e\neq\varnothing,
    \qquad
    \operatorname{EffUnres}(X^e)=\varnothing,
\]
and
\[
    F_e
    \subseteq
    \operatorname{FailRoot}(X^e).
\]
Because \(F_e=\operatorname{EffRoot}(X^e)\), the last inclusion states that every effective required root has active valid rejection, invalidity, impossibility, or unreachability evidence and has no valid repair neutralizing that failure. The rejection predicate also gives
\[
    F_e\cap\operatorname{Dis}(Q^e)=\varnothing
\]
and states that no live productive artifact can repair, accept, discharge, or reopen any root in \(F_e\). Thus the reported failure is complete over the effective requirement set, is not contradicted by a valid discharge of any of those roots, and cannot still be changed by remaining productive work.

Each active failure assertion used to place a root in \(\operatorname{FailRoot}(X^e)\), together with each exclusion or supersession assertion used to determine \(F_e\), is decision-relevant to the \textsc{Reject} report. Work soundness of \(X^e\) therefore supplies, for every \(r\in F_e\), observable interaction in \(G_{n(e)}\), a declared work operation, the corresponding entries in \(\Pi_{n(e)}\), and active evidence valid under \(V\) supporting the asserted failure of \(r\). Moreover, the definition of \(\operatorname{FailRoot}(X^e)\) excludes any root whose failure has been neutralized by a valid repair. Hence \(F_e\) is a \(V\)-certified failed-root set, not merely a set of roots carrying unsupported failure labels.

By the second clause of task adequacy, every \(V\)-certified failed-root set of this kind justifies rejection for the declared task class. Applying that clause to \(F_e\) proves that the decision
\[
    e=\textsc{Reject}
\]
is semantically correct. The detectability clause in task adequacy is consistent with both cases above: any process-relevant defect in required-root coverage, compatibility, evidence, or terminal settlement must be detectable under \(V\), and therefore cannot be hidden behind a purportedly certified terminal cover or failed-root set. Without this clause, a process could be internally sound relative to an incomplete validator while still missing a task-relevant defect.

The report occurrence \(e\) was arbitrary. Let \(\operatorname{SemCorrect}_D(e)\) denote that the report \(e\) is semantically correct for the task class declared by \(D\). We have proved
\[
    e=\textsc{Accept}
    \quad\Longrightarrow\quad
    \operatorname{SemCorrect}_D(e)
\]
and
\[
    e=\textsc{Reject}
    \quad\Longrightarrow\quad
    \operatorname{SemCorrect}_D(e)
\]
for every emitted \textsc{Accept} or \textsc{Reject} report.

It remains to identify what terminal completeness permits when neither semantic decision is reported. Since \(\mathfrak p\) is terminally complete, its final state \(X_N\) satisfies
\[
    \operatorname{Acc}(X_N)
    \;\lor\;
    \operatorname{Rej}(X_N)
    \;\lor\;
    \operatorname{CertStuck}_V(\mathfrak p,X_N),
\]
where \(\operatorname{CertStuck}_V(\mathfrak p,X_N)\) abbreviates the existence of valid evidence certifying that the process is externally stuck at \(X_N\). The acceptance and rejection predicates cannot both hold. Indeed, \(\operatorname{Rej}(X_N)\) would imply
\[
    \operatorname{EffRoot}(X_N)\neq\varnothing
\]
and
\[
    \operatorname{EffRoot}(X_N)
    \subseteq
    \operatorname{FailRoot}(X_N),
\]
whereas \(\operatorname{Acc}(X_N)\) would imply
\[
    \operatorname{EffRoot}(X_N)
    \cap
    \operatorname{FailRoot}(X_N)
    =\varnothing.
\]
These relations would force
\[
    \operatorname{EffRoot}(X_N)=\varnothing,
\]
contradicting the nonemptiness required by \(\operatorname{Rej}(X_N)\). Thus an accepted terminal decision and a rejected terminal decision are mutually exclusive.

If terminal completeness is witnessed instead by external stuckness, then by definition there exists valid evidence certifying that the cooperation process cannot presently continue because of an external process condition. Denote one such certificate by \(c_{\mathrm{stuck}}\). Its validity means
\[
    \operatorname{ValidCert}_V
    (c_{\mathrm{stuck}};X_N)=1,
\]
and its claim is that the process is externally stuck at \(X_N\). Because the process is work-sound, any active decision-relevant terminal-status assertion supported by this certificate is also connected through \(\Pi_N\) to the relevant observable events and declared operation. This validates the reported process status. It does not, however, provide either a \(V\)-certified accepted terminal cover or a \(V\)-certified failed-root set. The definition of task adequacy licenses semantic acceptance only from the former and semantic rejection only from the latter. External stuckness may arise from unavailable tools, permissions, resources, communication, or other process conditions without establishing that the task output is correct or that the task itself is semantically impossible. Therefore no \textsc{Accept} or \textsc{Reject} conclusion follows from external stuckness alone, and an externally stuck report is valid only as the certified status of the cooperation process. This proves the proposition.
\end{proof}

\subsection{Proof of Theorem~\ref{thm:assignment-regret}} \label{app:assignment-stability-proofs} 

\begin{proof}
Let
\[
    \mathcal{W}
    :=
    \WorkSet(X).
\]
By assumption, $\mathcal{W}\neq\varnothing$, and the component weights are normalized so that
\[
    \sum_{q\in\mathcal{W}} w_q = 1.
\]
Define the set of $\epsilon$-stable components by
\[
    \mathcal{S}_{\epsilon}(X)
    :=
    \left\{
        q\in\mathcal{W}
        :
        \max_{m\in\mathcal{F}_D(q,X)}
        \Delta_D
        \bigl(
            m\leftarrow\alpha_Q(q);
            q,X
        \bigr)
        \leq \epsilon
    \right\},
\]
and define its complement in $\mathcal{W}$ by
\[
    \mathcal{R}_{\epsilon}(X)
    :=
    \mathcal{W}\setminus\mathcal{S}_{\epsilon}(X).
\]
Hence,
\[
    \mathcal{W}
    =
    \mathcal{S}_{\epsilon}(X)
    \,\dot{\cup}\,
    \mathcal{R}_{\epsilon}(X).
\]
By the definition of the assignment-stability score and the normalization
$\sum_{q\in\mathcal{W}}w_q=1$, we have
\[
\begin{aligned}
    \astab_{D,\epsilon}(X)
    &=
    \frac{
        \sum_{q\in\mathcal{W}}
        w_q
        \mathbf{1}
        \left[
            q\in\mathcal{S}_{\epsilon}(X)
        \right]
    }{
        \sum_{q\in\mathcal{W}}w_q
    }\\
    &=
    \sum_{q\in\mathcal{W}}
    w_q
    \mathbf{1}
    \left[
        q\in\mathcal{S}_{\epsilon}(X)
    \right]\\
    &=
    \sum_{q\in\mathcal{S}_{\epsilon}(X)}w_q.
\end{aligned}
\]
It follows that
\[
\begin{aligned}
    \sum_{q\in\mathcal{R}_{\epsilon}(X)}w_q
    &=
    \sum_{q\in\mathcal{W}}w_q
    -
    \sum_{q\in\mathcal{S}_{\epsilon}(X)}w_q\\
    &=
    1-\astab_{D,\epsilon}(X).
\end{aligned}
\]

Now fix an arbitrary componentwise feasible alternative assignment $\beta$, so that
\[
    \beta(q)\in\mathcal{F}_D(q,X)\cup\{\bot\}
    \qquad
    \text{for every }q\in\mathcal{W}.
\]
For each $q\in\mathcal{W}$, define the componentwise utility gain
\[
    g_q
    :=
    v_D(\beta(q),q\mid X)
    -
    v_D(\alpha_Q(q),q\mid X).
\]
Using the definition of $U_D$, the aggregate gain produced by $\beta$ can be written exactly as
\[
\begin{aligned}
    &U_D(\beta\mid X)
    -
    U_D(\alpha_Q\mid X)\\
    &=
    \sum_{q\in\mathcal{W}}
    w_qv_D(\beta(q),q\mid X)
    -
    \sum_{q\in\mathcal{W}}
    w_qv_D(\alpha_Q(q),q\mid X)\\
    &=
    \sum_{q\in\mathcal{W}}
    w_q
    \left(
        v_D(\beta(q),q\mid X)
        -
        v_D(\alpha_Q(q),q\mid X)
    \right)\\
    &=
    \sum_{q\in\mathcal{W}}w_qg_q\\
    &=
    \sum_{q\in\mathcal{S}_{\epsilon}(X)}
    w_qg_q
    +
    \sum_{q\in\mathcal{R}_{\epsilon}(X)}
    w_qg_q.
\end{aligned}
\]

Consider any component
\[
    q\in\mathcal{S}_{\epsilon}(X).
\]
If
\[
    \beta(q)\in\mathcal{F}_D(q,X),
\]
then, by the definition of the counterfactual assignment gain,
\[
\begin{aligned}
    g_q
    &=
    v_D(\beta(q),q\mid X)
    -
    v_D(\alpha_Q(q),q\mid X)\\
    &=
    \Delta_D
    \bigl(
        \beta(q)\leftarrow\alpha_Q(q);
        q,X
    \bigr).
\end{aligned}
\]
Because $\beta(q)$ is an element of $\mathcal{F}_D(q,X)$, its gain cannot exceed the maximum gain over that feasible set. Therefore,
\[
\begin{aligned}
    g_q
    &=
    \Delta_D
    \bigl(
        \beta(q)\leftarrow\alpha_Q(q);
        q,X
    \bigr)\\
    &\leq
    \max_{m\in\mathcal{F}_D(q,X)}
    \Delta_D
    \bigl(
        m\leftarrow\alpha_Q(q);
        q,X
    \bigr).
\end{aligned}
\]
Since $q\in\mathcal{S}_{\epsilon}(X)$, the defining inequality of
$\epsilon$-stability gives
\[
    \max_{m\in\mathcal{F}_D(q,X)}
    \Delta_D
    \bigl(
        m\leftarrow\alpha_Q(q);
        q,X
    \bigr)
    \leq\epsilon.
\]
Combining the preceding two inequalities yields
\[
    g_q\leq\epsilon.
\]

If instead
\[
    \beta(q)=\bot,
\]
then the convention $v_D(\bot,q\mid X)=0$ gives
\[
\begin{aligned}
    g_q
    &=
    v_D(\bot,q\mid X)
    -
    v_D(\alpha_Q(q),q\mid X)\\
    &=
    -v_D(\alpha_Q(q),q\mid X).
\end{aligned}
\]
Since declared component values are nonnegative,
\[
    v_D(\alpha_Q(q),q\mid X)\geq0,
\]
and hence
\[
    g_q\leq0.
\]
Because $\epsilon\geq0$, it follows that
\[
    g_q\leq0\leq\epsilon.
\]
Thus, regardless of whether $\beta(q)$ is a feasible agent or $\bot$,
\[
    g_q\leq\epsilon
    \qquad
    \text{for every }q\in\mathcal{S}_{\epsilon}(X).
\]

Now consider any component
\[
    q\in\mathcal{R}_{\epsilon}(X).
\]
Although no $\epsilon$-stability guarantee is available for such a component, the normalization of the declared values implies
\[
    0
    \leq
    v_D(\beta(q),q\mid X)
    \leq1.
\]
Indeed, this follows directly from the range of $v_D$ when
$\beta(q)\in\mathcal{F}_D(q,X)$, while for $\beta(q)=\bot$ it follows from
$v_D(\bot,q\mid X)=0$. Similarly,
\[
    0
    \leq
    v_D(\alpha_Q(q),q\mid X)
    \leq1.
\]
Therefore,
\[
\begin{aligned}
    g_q
    &=
    v_D(\beta(q),q\mid X)
    -
    v_D(\alpha_Q(q),q\mid X)\\
    &\leq
    1-0\\
    &=1.
\end{aligned}
\]
Hence,
\[
    g_q\leq1
    \qquad
    \text{for every }q\in\mathcal{R}_{\epsilon}(X).
\]

Substituting the two componentwise bounds into the decomposition of the aggregate gain, and using $w_q>0$, gives
\[
\begin{aligned}
    U_D(\beta\mid X)
    -
    U_D(\alpha_Q\mid X)
    &=
    \sum_{q\in\mathcal{S}_{\epsilon}(X)}
    w_qg_q
    +
    \sum_{q\in\mathcal{R}_{\epsilon}(X)}
    w_qg_q\\
    &\leq
    \sum_{q\in\mathcal{S}_{\epsilon}(X)}
    w_q\epsilon
    +
    \sum_{q\in\mathcal{R}_{\epsilon}(X)}
    w_q\\
    &=
    \epsilon
    \sum_{q\in\mathcal{S}_{\epsilon}(X)}
    w_q
    +
    \sum_{q\in\mathcal{R}_{\epsilon}(X)}
    w_q.
\end{aligned}
\]
Using
\[
    \sum_{q\in\mathcal{S}_{\epsilon}(X)}w_q
    =
    \astab_{D,\epsilon}(X)
\]
and
\[
    \sum_{q\in\mathcal{R}_{\epsilon}(X)}w_q
    =
    1-\astab_{D,\epsilon}(X),
\]
we obtain
\[
\begin{aligned}
    U_D(\beta\mid X)
    -
    U_D(\alpha_Q\mid X)
    &\leq
    \epsilon\astab_{D,\epsilon}(X)
    +
    1-\astab_{D,\epsilon}(X)\\
    &=
    1
    -
    \astab_{D,\epsilon}(X)
    +
    \epsilon\astab_{D,\epsilon}(X)\\
    &=
    1
    -
    (1-\epsilon)\astab_{D,\epsilon}(X).
\end{aligned}
\]
This proves the first claim.

Finally, suppose that $X$ is agent-assignment stable. By definition,
\[
    \astab_{D,\epsilon}(X)=1.
\]
Substituting this identity into the established bound gives
\[
\begin{aligned}
    U_D(\beta\mid X)
    -
    U_D(\alpha_Q\mid X)
    &\leq
    1-(1-\epsilon)\cdot1\\
    &=\epsilon.
\end{aligned}
\]
Since $\beta$ was an arbitrary componentwise feasible joint reassignment, no feasible joint reassignment can improve the declared aggregate value by more than $\epsilon$.
\end{proof}

\subsection{Proof of Proposition~\ref{prop:quality-performance}} \label{app:quality-performance-proof} 

\begin{proof}
Fix a iCORE state $X$ satisfying the assumptions of the proposition, and write
\[
    \mathcal W=\WorkSet(X).
\]
Because assumption~(i) requires normalized weights over $\mathcal W$, the evaluated work set is nonempty and the weights satisfy
\[
    w_q>0
    \quad\text{for every }q\in\mathcal W,
    \qquad
    \sum_{q\in\mathcal W}w_q=1.
\]
By the weighted decomposition of task performance in assumption~(i), there exist nonnegative random component contributions $Y_q\geq0$, indexed by $q\in\mathcal W$, such that
\[
    \Perf
    =
    \sum_{q\in\mathcal W}w_qY_q.
\]
For every $q\in\mathcal W$, define
\[
    J_q
    =
    J_D(q;X)
    \in\{0,1\}
\]
and
\[
    S_q
    =
    \mathbf{1}
    \!\left[
        q\text{ is $\epsilon$-stable at }X
    \right]
    \in\{0,1\}.
\]
Equivalently, by Definition~\ref{def:assignment-stability},
\[
    S_q
    =
    \mathbf{1}
    \!\left[
        \max_{m\in\mathcal F_D(q,X)}
        \Delta_D(m\leftarrow\alpha_Q(q);q,X)
        \leq\epsilon
    \right].
\]
Since the state $X$ is fixed, the quantities $J_q$, $S_q$, $\alpha_Q(q)$, $\mathcal F_D(q,X)$, and $v_D(m,q\mid X)$ are deterministic under the conditional expectation given $X$.

We first relate the expected task contribution of each component to the declared contribution value of its current or recorded assignee. Suppose that $J_q=1$. Then every active decision-relevant assertion associated with $q$ is justified under the declared validator. Since $V$ is task-adequate, such a justified component is semantically valid for the task-performance decomposition. If $\alpha_Q(q)\neq\bot$, the calibration assumption therefore gives
\[
    \mathbb E[Y_q\mid X]
    \geq
    v_D(\alpha_Q(q),q\mid X).
\]
If $\alpha_Q(q)=\bot$, then the convention
\[
    v_D(\bot,q\mid X)=0
\]
and the nonnegativity of $Y_q$ give
\[
    \mathbb E[Y_q\mid X]
    \geq
    0
    =
    v_D(\alpha_Q(q),q\mid X).
\]
Thus, whenever $J_q=1$,
\[
    \mathbb E[Y_q\mid X]
    \geq
    v_D(\alpha_Q(q),q\mid X).
\]
When $J_q=0$, the nonnegativity of $Y_q$ implies
\[
    \mathbb E[Y_q\mid X]\geq0,
\]
while
\[
    J_qv_D(\alpha_Q(q),q\mid X)=0.
\]
Consequently, for every $q\in\mathcal W$, regardless of whether $q$ is justified,
\[
    \mathbb E[Y_q\mid X]
    \geq
    J_qv_D(\alpha_Q(q),q\mid X).
\]
Using the weighted decomposition of $\Perf$ and linearity of conditional expectation, we obtain
\[
\begin{aligned}
    \mathbb E[\Perf\mid X]
    &=
    \mathbb E\!\left[
        \sum_{q\in\mathcal W}w_qY_q
        \,\middle|\,
        X
    \right] \\
    &=
    \sum_{q\in\mathcal W}
    w_q\mathbb E[Y_q\mid X] \\
    &\geq
    \sum_{q\in\mathcal W}
    w_qJ_qv_D(\alpha_Q(q),q\mid X).
\end{aligned}
\]

We next derive a uniform lower bound on the declared value of every $\epsilon$-stable component. Fix an arbitrary $q\in\mathcal W$ for which $S_q=1$. Assumption~(iv) states that $q$ has at least one feasible agent whose declared value is at least $p_{\min}$. Hence, there exists an agent
\[
    m_q^\star\in\mathcal F_D(q,X)
\]
such that
\[
    v_D(m_q^\star,q\mid X)
    \geq
    p_{\min}.
\]
Since $S_q=1$, component $q$ is $\epsilon$-stable, and therefore
\[
    \max_{m\in\mathcal F_D(q,X)}
    \Delta_D(m\leftarrow\alpha_Q(q);q,X)
    \leq
    \epsilon.
\]
By the definition of the reassignment gain,
\[
    \Delta_D(m\leftarrow\alpha_Q(q);q,X)
    =
    v_D(m,q\mid X)
    -
    v_D(\alpha_Q(q),q\mid X).
\]
Substituting this definition into the stability condition yields
\[
    \max_{m\in\mathcal F_D(q,X)}
    \left\{
        v_D(m,q\mid X)
        -
        v_D(\alpha_Q(q),q\mid X)
    \right\}
    \leq
    \epsilon.
\]
Because $m_q^\star\in\mathcal F_D(q,X)$, its reassignment gain cannot exceed the maximum over the feasible-agent set. Thus,
\[
\begin{aligned}
    &v_D(m_q^\star,q\mid X)
    -
    v_D(\alpha_Q(q),q\mid X)\\
    &\leq
    \max_{m\in\mathcal F_D(q,X)}
    \left\{
        v_D(m,q\mid X)
        -
        v_D(\alpha_Q(q),q\mid X)
    \right\} \\
    &\leq
    \epsilon.
\end{aligned}
\]
Rearranging this inequality gives
\[
    v_D(\alpha_Q(q),q\mid X)
    \geq
    v_D(m_q^\star,q\mid X)-\epsilon.
\]
Using the lower bound
\[
    v_D(m_q^\star,q\mid X)\geq p_{\min},
\]
we conclude that
\[
    v_D(\alpha_Q(q),q\mid X)
    \geq
    p_{\min}-\epsilon.
\]
In particular, an $\epsilon$-stable component cannot be unassigned under the assumptions of the proposition. Indeed, if $\alpha_Q(q)=\bot$, then
\[
    v_D(\alpha_Q(q),q\mid X)
    =
    v_D(\bot,q\mid X)
    =
    0,
\]
and hence
\[
\begin{aligned}
    \Delta_D(m_q^\star\leftarrow\bot;q,X)
    &=
    v_D(m_q^\star,q\mid X)
    -
    v_D(\bot,q\mid X) \\
    &=
    v_D(m_q^\star,q\mid X) \\
    &\geq
    p_{\min} \\
    &>
    \epsilon,
\end{aligned}
\]
which would contradict the $\epsilon$-stability of $q$.

We have therefore shown that
\[
    S_q=1
    \quad\Longrightarrow\quad
    v_D(\alpha_Q(q),q\mid X)
    \geq
    p_{\min}-\epsilon.
\]
If $S_q=0$, then
\[
    (p_{\min}-\epsilon)S_q=0,
\]
and the range condition
\[
    v_D(\alpha_Q(q),q\mid X)\in[0,1]
\]
again gives
\[
    v_D(\alpha_Q(q),q\mid X)
    \geq
    (p_{\min}-\epsilon)S_q.
\]
Consequently, the inequality
\[
    v_D(\alpha_Q(q),q\mid X)
    \geq
    (p_{\min}-\epsilon)S_q
\]
holds for every $q\in\mathcal W$. Multiplying both sides by the nonnegative indicator $J_q$ gives
\[
    J_qv_D(\alpha_Q(q),q\mid X)
    \geq
    (p_{\min}-\epsilon)J_qS_q.
\]
Substituting this componentwise bound into the preceding lower bound on conditional task performance yields
\[
\begin{aligned}
    \mathbb E[\Perf\mid X]
    &\geq
    \sum_{q\in\mathcal W}
    w_qJ_qv_D(\alpha_Q(q),q\mid X) \\
    &\geq
    (p_{\min}-\epsilon)
    \sum_{q\in\mathcal W}w_qJ_qS_q.
\end{aligned}
\]

It remains to lower-bound the total normalized weight of the components that are simultaneously justified and $\epsilon$-stable. Since $J_q,S_q\in\{0,1\}$, one has
\[
    (1-J_q)(1-S_q)\geq0.
\]
Expanding the left-hand side gives
\[
    1-J_q-S_q+J_qS_q\geq0,
\]
and therefore
\[
    J_qS_q
    \geq
    J_q+S_q-1.
\]
Multiplying by $w_q>0$ and summing over $q\in\mathcal W$ gives
\[
\begin{aligned}
    \sum_{q\in\mathcal W}w_qJ_qS_q
    &\geq
    \sum_{q\in\mathcal W}
    w_q(J_q+S_q-1) \\
    &=
    \sum_{q\in\mathcal W}w_qJ_q
    +
    \sum_{q\in\mathcal W}w_qS_q
    -
    \sum_{q\in\mathcal W}w_q.
\end{aligned}
\]
Since the weights are normalized,
\[
    \sum_{q\in\mathcal W}w_q=1.
\]
Moreover, by the definition of work soundness and the fact that $\mathcal W\neq\varnothing$,
\[
    \snd_D(X)
    =
    \frac{
        \sum_{q\in\mathcal W}w_qJ_D(q;X)
    }{
        \sum_{q\in\mathcal W}w_q
    }
    =
    \sum_{q\in\mathcal W}w_qJ_q.
\]
Similarly, by the definition of agent-assignment stability,
\[
    \astab_{D,\epsilon}(X)
    =
    \frac{
        \sum_{q\in\mathcal W}w_q
        \mathbf{1}[q\text{ is $\epsilon$-stable at }X]
    }{
        \sum_{q\in\mathcal W}w_q
    }
    =
    \sum_{q\in\mathcal W}w_qS_q.
\]
Hence,
\[
    \sum_{q\in\mathcal W}w_qJ_qS_q
    \geq
    \snd_D(X)
    +
    \astab_{D,\epsilon}(X)
    -
    1.
\]
On the other hand, every term $w_qJ_qS_q$ is nonnegative, so
\[
    \sum_{q\in\mathcal W}w_qJ_qS_q
    \geq
    0.
\]
Combining these two inequalities gives
\[
    \sum_{q\in\mathcal W}w_qJ_qS_q
    \geq
    \max\left\{
        0,\,
        \snd_D(X)
        +
        \astab_{D,\epsilon}(X)
        -
        1
    \right\}.
\]
Equivalently,
\[
    \sum_{q\in\mathcal W}w_qJ_qS_q
    \geq
    \left[
        \snd_D(X)
        +
        \astab_{D,\epsilon}(X)
        -
        1
    \right]_+.
\]
Therefore,
\[
\begin{aligned}
    \mathbb E[\Perf\mid X]
    &\geq
    (p_{\min}-\epsilon)
    \sum_{q\in\mathcal W}w_qJ_qS_q \\
    &\geq
    (p_{\min}-\epsilon)
    \left[
        \snd_D(X)
        +
        \astab_{D,\epsilon}(X)
        -
        1
    \right]_+.
\end{aligned}
\]
By Definition~\ref{def:core-quality},
\[
    \cq_{D,\epsilon}(X)
    =
    \frac{1}{2}
    \left(
        \snd_D(X)
        +
        \astab_{D,\epsilon}(X)
    \right),
\]
and therefore
\[
    \snd_D(X)
    +
    \astab_{D,\epsilon}(X)
    =
    2\cq_{D,\epsilon}(X).
\]
Substituting this identity into the preceding bound gives
\[
\begin{aligned}
    \mathbb E[\Perf\mid X]
    &\geq
    (p_{\min}-\epsilon)
    \left[
        2\cq_{D,\epsilon}(X)-1
    \right]_+ \\
    &=
    (p_{\min}-\epsilon)
    \left[
        2\left(
            \cq_{D,\epsilon}(X)-\frac{1}{2}
        \right)
    \right]_+.
\end{aligned}
\]
Since $2>0$, the positive-part operator satisfies
\[
    [2z]_+=2[z]_+
\]
for every $z\in\mathbb R$. Applying this identity with
\[
    z
    =
    \cq_{D,\epsilon}(X)-\frac{1}{2}
\]
finally yields
\[
    \mathbb E[\Perf\mid X]
    \geq
    2(p_{\min}-\epsilon)
    \left[
        \cq_{D,\epsilon}(X)-\frac{1}{2}
    \right]_+.
\]
This is the claimed performance implication.
\end{proof}

\subsection{Proof of Proposition~\ref{prop:core-audit-properties}} \label{app:core-audit-proof} 

\begin{proof}
Let
\[
    X^{(0)},X^{(1)},X^{(2)},\ldots
\]
denote the actual sequence of states retained by
Algorithm~\ref{alg:core-audit}, with
\[
    X^{(0)}=X_0.
\]
Tentative states such as $\widetilde X$ and intervention proposals such
as $X_f$ that fail their corresponding audits are not members of this
sequence, because they are never committed. The failure and feedback
records retained by the algorithm are themselves represented as legal
iCORE updates and do not install the rejected candidate assertions.
Hence every transition
\[
    X^{(k)}\longrightarrow X^{(k+1)}
\]
in the retained sequence is a legal iCORE update under the hypotheses of
the proposition.

We first prove that every retained state is work-sound. We proceed by
induction on $k$. The initial state is work-sound by assumption, and
therefore
\[
    X^{(0)}=X_0
\]
is work-sound. Now fix an arbitrary $k\geq 0$ and suppose that
\[
    X^{(k)}
\]
is work-sound. Since
\[
    X^{(k)}\longrightarrow X^{(k+1)}
\]
is a legal iCORE update, Lemma~\ref{lem:one-step-work-soundness} applies
with
\[
    X=X^{(k)}
    \qquad\text{and}\qquad
    X'=X^{(k+1)}.
\]
It follows that
\[
    X^{(k+1)}
\]
is work-sound. The induction therefore gives
\[
    X^{(k)}
    \text{ is work-sound for every retained index }k.
\]
In particular, the conclusion holds after both an ordinary commitment
and an intervention commitment. For a committed intervention, the same
conclusion is also certified directly by
$\operatorname{AuditRecovery}_D$, because the recovery audit accepts an
intervention candidate only when the candidate is legal and work-sound.
Thus every committed state produced by Algorithm~\ref{alg:core-audit}
is work-sound.

It remains to establish that quality is nondecreasing on the matched
audit set whenever an intervention is committed. Consider an arbitrary
committed intervention. Let $X^{-}$ denote the state immediately before
the diagnosis, and write
\[
    F
    =
    \operatorname{Diagnose}_D(X^{-}),
    \qquad
    \mathcal W_F
    =
    \operatorname{MatchedSet}_D(X^{-},F).
\]
Let $X_f$ be the labeled intervention state proposed in response to
$F$. Since this intervention is committed, its recovery audit returns
no violation:
\[
    \operatorname{AuditRecovery}_D
    \bigl(
        X^{-},
        X_f,
        F,
        \mathcal W_F
    \bigr)
    =
    \varnothing.
\]
Let
\[
    X^{+}
    =
    \operatorname{Commit}_D(X_f)
\]
denote the state retained after the successful intervention.
Commitment changes the audited proposal from tentative to retained but
does not alter the semantic contents of $G$, $Q$, or $\Pi$ on which the
matched-set indicators depend. Hence, for every
$q\in\mathcal W_F$,
\[
    J_D^F(q;X^{+})
    =
    J_D^F(q;X_f)
\]
and
\[
    A_{D,\epsilon}^F(q;X^{+})
    =
    A_{D,\epsilon}^F(q;X_f).
\]
Consequently,
\[
    \cq_{D,\epsilon}(X^{+};\mathcal W_F)
    =
    \cq_{D,\epsilon}(X_f;\mathcal W_F).
\]

By construction, the set $\mathcal W_F$ and all of its weights are
frozen before the diagnosis is supplied to the responsible agent.
Therefore the intervention cannot alter the cohort or the weights on
which the pre-intervention and post-intervention scores are evaluated.
Moreover, every $q\in\mathcal W_F$ is transported from $X^{-}$ to
$X^{+}$ through the identity, refinement, replacement, supersession, or
settlement relations declared by $D$. For brevity, define
\[
\begin{aligned}
    J_q^{-}
    &=
    J_D^F(q;X^{-}),
    &
    J_q^{+}
    &=
    J_D^F(q;X^{+}),\\
    A_q^{-}
    &=
    A_{D,\epsilon}^F(q;X^{-}),
    &
    A_q^{+}
    &=
    A_{D,\epsilon}^F(q;X^{+}).
\end{aligned}
\]

Suppose first that
\[
    \mathcal W_F=\varnothing.
\]
By the empty-set convention in the matched-set definitions,
\[
    \snd_D(X^{-};\mathcal W_F)=1
\]
and
\[
    \astab_{D,\epsilon}(X^{-};\mathcal W_F)=1.
\]
The same convention gives
\[
    \snd_D(X^{+};\mathcal W_F)=1
\]
and
\[
    \astab_{D,\epsilon}(X^{+};\mathcal W_F)=1.
\]
It follows that
\[
\begin{aligned}
    \cq_{D,\epsilon}(X^{-};\mathcal W_F)
    &=
    \frac{1}{2}
    \left(
        \snd_D(X^{-};\mathcal W_F)
        +
        \astab_{D,\epsilon}(X^{-};\mathcal W_F)
    \right)\\
    &=
    \frac{1}{2}(1+1)\\
    &=1
\end{aligned}
\]
and, similarly,
\[
\begin{aligned}
    \cq_{D,\epsilon}(X^{+};\mathcal W_F)
    &=
    \frac{1}{2}
    \left(
        \snd_D(X^{+};\mathcal W_F)
        +
        \astab_{D,\epsilon}(X^{+};\mathcal W_F)
    \right)\\
    &=
    \frac{1}{2}(1+1)\\
    &=1.
\end{aligned}
\]
Therefore
\[
    \cq_{D,\epsilon}(X^{+};\mathcal W_F)
    =
    \cq_{D,\epsilon}(X^{-};\mathcal W_F),
\]
so the required nondecrease holds with equality when the matched audit
set is empty.

Now suppose that
\[
    \mathcal W_F\neq\varnothing.
\]
Define the frozen total matched weight by
\[
    W_F
    =
    \sum_{q\in\mathcal W_F}w_q.
\]
Every predeclared importance weight satisfies
\[
    w_q>0.
\]
Since $\mathcal W_F$ is nonempty, it follows that
\[
    W_F>0.
\]
Using the matched-set definitions, the pre-intervention soundness score
is
\[
    \snd_D(X^{-};\mathcal W_F)
    =
    \frac{
        \sum_{q\in\mathcal W_F}
        w_qJ_q^{-}
    }{
        W_F
    },
\]
and the corresponding assignment-stability score is
\[
    \astab_{D,\epsilon}(X^{-};\mathcal W_F)
    =
    \frac{
        \sum_{q\in\mathcal W_F}
        w_qA_q^{-}
    }{
        W_F
    }.
\]
Therefore
\[
\begin{aligned}
    \cq_{D,\epsilon}(X^{-};\mathcal W_F)
    &=
    \frac{1}{2}
    \left(
        \snd_D(X^{-};\mathcal W_F)
        +
        \astab_{D,\epsilon}(X^{-};\mathcal W_F)
    \right)\\
    &=
    \frac{1}{2}
    \left(
        \frac{
            \sum_{q\in\mathcal W_F}w_qJ_q^{-}
        }{
            W_F
        }
        +
        \frac{
            \sum_{q\in\mathcal W_F}w_qA_q^{-}
        }{
            W_F
        }
    \right)\\
    &=
    \frac{1}{2W_F}
    \left(
        \sum_{q\in\mathcal W_F}w_qJ_q^{-}
        +
        \sum_{q\in\mathcal W_F}w_qA_q^{-}
    \right)\\
    &=
    \frac{1}{2W_F}
    \sum_{q\in\mathcal W_F}
    w_q
    \bigl(
        J_q^{-}+A_q^{-}
    \bigr).
\end{aligned}
\]
The same expansion for the committed successor gives
\[
\begin{aligned}
    \cq_{D,\epsilon}(X^{+};\mathcal W_F)
    &=
    \frac{1}{2}
    \left(
        \snd_D(X^{+};\mathcal W_F)
        +
        \astab_{D,\epsilon}(X^{+};\mathcal W_F)
    \right)\\
    &=
    \frac{1}{2}
    \left(
        \frac{
            \sum_{q\in\mathcal W_F}w_qJ_q^{+}
        }{
            W_F
        }
        +
        \frac{
            \sum_{q\in\mathcal W_F}w_qA_q^{+}
        }{
            W_F
        }
    \right)\\
    &=
    \frac{1}{2W_F}
    \sum_{q\in\mathcal W_F}
    w_q
    \bigl(
        J_q^{+}+A_q^{+}
    \bigr).
\end{aligned}
\]

Because
\[
    \operatorname{AuditRecovery}_D
    \bigl(
        X^{-},
        X_f,
        F,
        \mathcal W_F
    \bigr)
    =
    \varnothing
\]
and because the recovery audit enforces the matched-set conditions of
Algorithm~\ref{alg:core-audit}, the accepted intervention satisfies
\[
    \cq_{D,\epsilon}(X_f;\mathcal W_F)
    \geq
    \cq_{D,\epsilon}(X^{-};\mathcal W_F).
\]
Since the matched semantic indicators of $X_f$ and $X^{+}$ are
identical, this inequality is equivalent to
\[
    \cq_{D,\epsilon}(X^{+};\mathcal W_F)
    \geq
    \cq_{D,\epsilon}(X^{-};\mathcal W_F).
\]
Substituting the two expanded expressions yields
\[
    \frac{1}{2W_F}
    \sum_{q\in\mathcal W_F}
    w_q
    \bigl(
        J_q^{+}+A_q^{+}
    \bigr)
    \geq
    \frac{1}{2W_F}
    \sum_{q\in\mathcal W_F}
    w_q
    \bigl(
        J_q^{-}+A_q^{-}
    \bigr).
\]
Since
\[
    2W_F>0,
\]
multiplication by $2W_F$ preserves the direction of the inequality and
gives
\[
    \sum_{q\in\mathcal W_F}
    w_q
    \bigl(
        J_q^{+}+A_q^{+}
    \bigr)
    \geq
    \sum_{q\in\mathcal W_F}
    w_q
    \bigl(
        J_q^{-}+A_q^{-}
    \bigr).
\]
Subtracting the right-hand side from the left-hand side gives
\[
    \sum_{q\in\mathcal W_F}
    w_q
    \left[
        \bigl(J_q^{+}-J_q^{-}\bigr)
        +
        \bigl(A_q^{+}-A_q^{-}\bigr)
    \right]
    \geq
    0.
\]
Dividing by the positive quantity $2W_F$ now yields
\[
\begin{aligned}
    &\cq_{D,\epsilon}(X^{+};\mathcal W_F)
    -
    \cq_{D,\epsilon}(X^{-};\mathcal W_F)\\
    &\qquad=
    \frac{1}{2W_F}
    \sum_{q\in\mathcal W_F}
    w_q
    \left[
        \bigl(J_q^{+}-J_q^{-}\bigr)
        +
        \bigl(A_q^{+}-A_q^{-}\bigr)
    \right]\\
    &\qquad\geq 0.
\end{aligned}
\]
Hence
\[
    \cq_{D,\epsilon}(X^{+};\mathcal W_F)
    \geq
    \cq_{D,\epsilon}(X^{-};\mathcal W_F).
\]

The transport and no-silent-removal conditions ensure that this
inequality represents an actual quality comparison rather than a change
of evaluation cohort. In particular, if some
$q\in\mathcal W_F$ is deleted without a valid identity, refinement,
replacement, supersession, or settlement record, then the matched-set
definition assigns
\[
    J_D^F(q;X^{+})=0
\]
instead of removing $q$ and its weight from the score. A component that
is discharged, rejected, replaced, or superseded receives
\[
    J_D^F(q;X^{+})=1
\]
only when the corresponding transport or settlement operation is
supported by all evidence required by $D$. Likewise, a settled component
receives
\[
    A_{D,\epsilon}^F(q;X^{+})=1
\]
only when it was completed under an $\epsilon$-stable assignment
according to capability evidence available no later than settlement.
Thus the intervention cannot obtain the preceding inequality merely by
deleting, renaming, hiding, or retrospectively relabeling difficult
work.

The committed intervention was arbitrary. Therefore every committed
intervention satisfies
\[
    \cq_{D,\epsilon}(X^{+};\mathcal W_F)
    \geq
    \cq_{D,\epsilon}(X^{-};\mathcal W_F).
\]
Hence iCORE-Audit is quality-nondecreasing on the matched audit set at
every committed intervention. Together with the work-soundness
induction, this proves both conclusions of the proposition.
\end{proof}

\subsection{Proof of Theorem~\ref{thm:execution-order-local-to-global}} 
\label{app:execution-order-stability-proof}

\begin{lemma}[Termination and local joinability imply global joinability]
\label{lem:terminating-local-joinability}
Let $(\mathcal A,\Rightarrow)$ be an abstract rewrite system such that $\Rightarrow$ terminates. Then every element of $\mathcal A$ has a $\Rightarrow$-normal descendant. If, in addition, every one-step fork is joinable, in the sense that whenever
\[
    A\Rightarrow B,
    \qquad
    A\Rightarrow C,
\]
there exists $D\in\mathcal A$ satisfying
\[
    B\Rightarrow^{*}D,
    \qquad
    C\Rightarrow^{*}D.
\]
Then every pair of descendants of a common source is joinable: for all $A,B,C\in\mathcal A$ such that
\[
    A\Rightarrow^{*}B,
    \qquad
    A\Rightarrow^{*}C,
\]
there exists $H\in\mathcal A$ such that
\[
    B\Rightarrow^{*}H,
    \qquad
    C\Rightarrow^{*}H.
\]
Under this additional local-joinability assumption, any two normal descendants of the same element are equal.
\end{lemma}

\begin{proof}
We first verify existence of normal descendants. Fix $A\in\mathcal A$. If no normal element were reachable from $A$, then $A$ itself would not be normal, so there would exist $A_1$ with
\[
    A\Rightarrow A_1.
\]
The element $A_1$ would likewise have no normal descendant, because any normal descendant of $A_1$ would also be a normal descendant of $A$. Hence there would exist $A_2$ with
\[
    A_1\Rightarrow A_2.
\]
Repeating the same argument would produce an infinite sequence
\[
    A\Rightarrow A_1\Rightarrow A_2\Rightarrow A_3\Rightarrow\cdots,
\]
contradicting termination. Therefore every $A\in\mathcal A$ has at least one normal descendant.

We now prove global joinability by well-founded induction. Define a strict descendant relation $\prec$ on $\mathcal A$ by
\[
    B\prec A
    \quad\Longleftrightarrow\quad
    A\Rightarrow^{+}B,
\]
where $\Rightarrow^{+}$ is the transitive closure of $\Rightarrow$. Since $\Rightarrow$ terminates, there is no infinite chain
\[
    A_0\Rightarrow^{+}A_1\Rightarrow^{+}A_2\Rightarrow^{+}\cdots,
\]
and consequently $\prec$ is well founded. For $A\in\mathcal A$, let $\mathsf P(A)$ denote the assertion that every two descendants of $A$ are joinable. Assume inductively that $\mathsf P(A')$ holds for every $A'\prec A$, and consider arbitrary reductions
\[
    A\Rightarrow^{*}B,
    \qquad
    A\Rightarrow^{*}C.
\]
If the first reduction has length zero, then $B=A$, and choosing $H=C$ gives
\[
    B=A\Rightarrow^{*}C=H,
    \qquad
    C\Rightarrow^{*}C=H.
\]
The same argument applies if the second reduction has length zero. We may therefore assume that both reductions are nonempty and write their first steps as
\[
    A\Rightarrow A_1\Rightarrow^{*}B,
    \qquad
    A\Rightarrow A_2\Rightarrow^{*}C.
\]
By the assumed joinability of every one-step fork, there exists $D$ such that
\[
    A_1\Rightarrow^{*}D,
    \qquad
    A_2\Rightarrow^{*}D.
\]
Because $A\Rightarrow A_1$, we have $A_1\prec A$. Both $B$ and $D$ are descendants of $A_1$, so the induction hypothesis $\mathsf P(A_1)$ yields an element $E$ satisfying
\[
    B\Rightarrow^{*}E,
    \qquad
    D\Rightarrow^{*}E.
\]
Similarly, $A_2\prec A$, and both $C$ and $D$ are descendants of $A_2$. The induction hypothesis $\mathsf P(A_2)$ therefore yields an element $F$ satisfying
\[
    C\Rightarrow^{*}F,
    \qquad
    D\Rightarrow^{*}F.
\]
The element $D$ is itself a strict descendant of $A$, because
\[
    A\Rightarrow A_1\Rightarrow^{*}D.
\]
Thus $D\prec A$. Since $E$ and $F$ are both descendants of $D$, the induction hypothesis $\mathsf P(D)$ gives an element $H$ for which
\[
    E\Rightarrow^{*}H,
    \qquad
    F\Rightarrow^{*}H.
\]
Combining the reductions gives
\[
    B\Rightarrow^{*}E\Rightarrow^{*}H
    \qquad\text{and}\qquad
    C\Rightarrow^{*}F\Rightarrow^{*}H,
\]
so $B$ and $C$ are joinable. This proves $\mathsf P(A)$. Well-founded induction now implies $\mathsf P(A)$ for every $A\in\mathcal A$.

Finally, suppose that $N_1$ and $N_2$ are normal descendants of the same element $A$. Global joinability provides $H$ such that
\[
    N_1\Rightarrow^{*}H,
    \qquad
    N_2\Rightarrow^{*}H.
\]
Because $N_1$ is normal, the first reduction must have length zero, and hence $H=N_1$. Because $N_2$ is normal, the second reduction must also have length zero, and hence $H=N_2$. Therefore
\[
    N_1=N_2,
\]
which proves uniqueness of the normal descendant.
\end{proof}

\begin{proof}[Proof of Theorem~\ref{thm:execution-order-local-to-global}]
Let
\[
    \mathcal X_p^{\mathrm{reach}}
    =
    \bigl\{
        [X]_{\equiv_X}: X\text{ is reachable}
    \bigr\}
\]
be the set of reachable equivalence classes, and abbreviate the productive quotient relation by
\[
    A\Rightarrow B
    \quad\Longleftrightarrow\quad
    A\rightarrow_p B,
    \qquad
    A,B\in\mathcal X_p^{\mathrm{reach}}.
\]
The theorem assumes that this relation terminates.

We first record how the quotient relation is connected to actual legal updates. Suppose
\[
    [S]_{\equiv_X}\Rightarrow[T]_{\equiv_X}.
\]
By the definition of the productive quotient, there exist representatives $S'$ and $T'$ such that
\[
    S'\equiv_X S,
    \qquad
    T'\equiv_X T,
    \qquad
    S'\rightarrow T',
    \qquad
    S'\not\equiv_X T'.
\]
Let $\rho$ be the legal update occurrence realizing $S'\rightarrow T'$. Since declared state equivalence is a congruence for legal updates, $\rho$ can be transported from $S'$ to any equivalent representative $\widehat S\equiv_X S'$. The transported occurrence is enabled and legal at $\widehat S$, and it produces a successor $\widehat T$ satisfying
\[
    \widehat S\rightarrow\widehat T,
    \qquad
    \widehat T\equiv_X T'.
\]
This transported update remains productive. Indeed, if $\widehat S\equiv_X\widehat T$, then
\[
    S'\equiv_X\widehat S
    \equiv_X\widehat T
    \equiv_X T',
\]
which would imply $S'\equiv_X T'$ by transitivity and contradict the productivity of $S'\rightarrow T'$. Consequently, every quotient edge can be realized by a productive legal update from any representative of its source class.

Applying this observation successively gives the following path-lifting property. For every finite quotient path
\[
    A_0\Rightarrow A_1\Rightarrow\cdots\Rightarrow A_k
\]
and every state $X_0'$ satisfying
\[
    [X_0']_{\equiv_X}=A_0,
\]
there exists a sequence of productive legal updates
\[
    X_0'\rightarrow X_1'\rightarrow\cdots\rightarrow X_k'
\]
such that
\[
    [X_i']_{\equiv_X}=A_i
    \qquad
    \text{for every }i\in\{0,1,\ldots,k\}.
\]
This follows by induction on $i$: once $X_i'$ has been constructed in the class $A_i$, the next quotient edge can be transported to $X_i'$ by the preceding congruence argument, producing $X_{i+1}'$ in the class $A_{i+1}$.

The quotient normal forms are exactly the equivalence classes of productive normal forms. To see one direction, suppose that $N$ is a productive normal form but that $[N]_{\equiv_X}$ has an outgoing quotient edge. A witness for that edge would be a productive update from some state $N'\equiv_X N$, and congruence would transport that update to a productive update enabled at $N$, contradicting normality. Conversely, if $[N]_{\equiv_X}$ is quotient-normal but a productive legal update $N\rightarrow N'$ is enabled, then
\[
    [N]_{\equiv_X}\Rightarrow[N']_{\equiv_X}
\]
would be an outgoing quotient edge, again a contradiction. Hence
\[
    N\text{ is a productive normal form}
    \quad\Longleftrightarrow\quad
    [N]_{\equiv_X}\text{ is $\Rightarrow$-normal}.
\]

Assume first that the system is globally execution-order stable. Let
\[
    X\xrightarrow{\rho}Y,
    \qquad
    X\xrightarrow{\eta}Z
\]
be an arbitrary reachable productive local critical pair. Since the quotient relation terminates, the normal-form-existence conclusion of Lemma~\ref{lem:terminating-local-joinability} guarantees that both $[Y]_{\equiv_X}$ and $[Z]_{\equiv_X}$ have quotient-normal descendants. Thus there exist quotient-normal classes $A_Y$ and $A_Z$ such that
\[
    [Y]_{\equiv_X}\Rightarrow^{*}A_Y,
    \qquad
    [Z]_{\equiv_X}\Rightarrow^{*}A_Z.
\]
By the path-lifting property, these quotient reductions can be realized by productive legal continuations
\[
    Y\rightarrow^{*}N_Y,
    \qquad
    Z\rightarrow^{*}N_Z,
\]
with
\[
    [N_Y]_{\equiv_X}=A_Y,
    \qquad
    [N_Z]_{\equiv_X}=A_Z.
\]
Because $A_Y$ and $A_Z$ are quotient-normal, the established normal-form correspondence implies that $N_Y$ and $N_Z$ are productive normal forms. Concatenating the two first updates with these continuations gives two productive continuations from the same reachable state $X$:
\[
    X\xrightarrow{\rho}Y\rightarrow^{*}N_Y,
    \qquad
    X\xrightarrow{\eta}Z\rightarrow^{*}N_Z.
\]
Global execution-order stability therefore gives
\[
    N_Y\equiv_X N_Z.
\]
Equivalently,
$
    [Y]_{\equiv_X}\Rightarrow^{*}[N_Y]_{\equiv_X},
    \qquad
    [Z]_{\equiv_X}\Rightarrow^{*}[N_Z]_{\equiv_X},
    \qquad
    [N_Y]_{\equiv_X}=[N_Z]_{\equiv_X}.
$
This is exactly
\[
    Y\downarrow_{p,\equiv_X}Z.
\]
Since the reachable productive local critical pair was arbitrary, every reachable productive local critical pair is productively joinable modulo $\equiv_X$.

For the converse, assume that every reachable productive local critical pair is productively joinable modulo $\equiv_X$. We show that every one-step fork of the reachable productive quotient is joinable. Consider arbitrary quotient edges
\[
    A\Rightarrow B,
    \qquad
    A\Rightarrow C,
\]
where $A\in\mathcal X_p^{\mathrm{reach}}$. Choose a reachable state $X$ with
\[
    A=[X]_{\equiv_X}.
\]
By the definition of the two quotient edges, each edge has a productive legal-update witness whose source is equivalent to $X$. Transporting both witnesses to the common representative $X$ by congruence yields productive legal updates
\[
    X\xrightarrow{\rho}Y,
    \qquad
    X\xrightarrow{\eta}Z
\]
such that
\[
    B=[Y]_{\equiv_X},
    \qquad
    C=[Z]_{\equiv_X}.
\]
If $B=C$, then the quotient fork is trivially joinable by taking its common descendant to be $B=C$. Suppose therefore that $B\neq C$. By the fork-classification assumption of the theorem, the reachable one-step productive fork at $X$ is either residual-sound independent or a productive local critical pair.

If the two updates are residual-sound independent, then both residual executions are defined. Hence there exist states $U$ and $V$ such that
\[
    X\xrightarrow{\rho}Y
    \xrightarrow{\eta_{\rho}}U,
    \qquad
    X\xrightarrow{\eta}Z
    \xrightarrow{\rho_{\eta}}V,
    \qquad
    U\equiv_X V.
\]
A residual update may itself be productive or stuttering. If $Y\rightarrow U$ is productive, then
\[
    [Y]_{\equiv_X}\Rightarrow[U]_{\equiv_X};
\]
if it is stuttering, then $Y\equiv_X U$, so
\[
    [Y]_{\equiv_X}=[U]_{\equiv_X},
\]
so a zero-step quotient reduction connects the same two classes. In either case,
\[
    [Y]_{\equiv_X}\Rightarrow^{*}[U]_{\equiv_X}.
\]
The identical argument gives
\[
    [Z]_{\equiv_X}\Rightarrow^{*}[V]_{\equiv_X}.
\]
Since $U\equiv_X V$, we have
\[
    [U]_{\equiv_X}=[V]_{\equiv_X}.
\]
Thus $B=[Y]_{\equiv_X}$ and $C=[Z]_{\equiv_X}$ are joinable.

If instead the fork is a productive local critical pair, the present assumption gives
\[
    Y\downarrow_{p,\equiv_X}Z.
\]
By definition, there exist productive descendants $Y'$ and $Z'$ satisfying
\[
    [Y]_{\equiv_X}\Rightarrow^{*}[Y']_{\equiv_X},
    \qquad
    [Z]_{\equiv_X}\Rightarrow^{*}[Z']_{\equiv_X},
    \qquad
    Y'\equiv_X Z'.
\]
Therefore
\[
    [Y']_{\equiv_X}=[Z']_{\equiv_X},
\]
so the quotient branches are again joinable. We have proved that every one-step fork of $(\mathcal X_p^{\mathrm{reach}},\Rightarrow)$ is joinable.

The quotient relation terminates by hypothesis. Lemma~\ref{lem:terminating-local-joinability} therefore applies and shows that any two quotient descendants of a common reachable class are joinable. Let $X$ be any reachable state, and consider arbitrary productive continuations
\[
    X\rightarrow^{*}N_1,
    \qquad
    X\rightarrow^{*}N_2
\]
to productive normal forms. Passing to equivalence classes gives
\[
    [X]_{\equiv_X}\Rightarrow^{*}[N_1]_{\equiv_X},
    \qquad
    [X]_{\equiv_X}\Rightarrow^{*}[N_2]_{\equiv_X}.
\]
The normal-form correspondence shows that both $[N_1]_{\equiv_X}$ and $[N_2]_{\equiv_X}$ are quotient-normal. The uniqueness conclusion of Lemma~\ref{lem:terminating-local-joinability} therefore gives
\[
    [N_1]_{\equiv_X}=[N_2]_{\equiv_X}.
\]
Equality of the quotient classes is equivalent to
\[
    N_1\equiv_X N_2.
\]
Since $X$ and the two productive continuations were arbitrary, the system is globally execution-order stable. This proves the converse implication and completes the equivalence.
\end{proof}

%% file: results/tables_core_revised/main_results_revised.tex
\begin{table*}[t]
\centering
\small
\setlength{\tabcolsep}{4.6pt}
\renewcommand{\arraystretch}{1.08}
\begin{tabular}{@{}lccc|ccc@{}}
\toprule
& \multicolumn{3}{c|}{Controlled (540 episodes)} & \multicolumn{3}{c}{Real LLM (108 episodes)} \\
\cmidrule(lr){2-4}\cmidrule(lr){5-7}
System & CQ & Cond. Perf. & Term. Perf. & CQ & Cond. Perf. & Term. Perf. \\
\midrule
MAS-Only & 0.684 & 0.570 & 0.638 & 0.586 & 0.392 & 0.356 \\
Interaction-only & 0.680 & 0.579 & 0.662 & 0.586 & 0.392 & 0.356 \\
Task-only & \underline{0.738} & \underline{0.592} & 0.654 & \underline{0.683} & \underline{0.511} & \underline{0.495} \\
LLM-Judge & 0.686 & 0.581 & \underline{0.671} & 0.586 & 0.392 & 0.356 \\
iCORE-Observe & 0.684 & 0.570 & 0.638 & 0.586 & 0.392 & 0.356 \\
iCORE-Audit & \textbf{0.799} & \textbf{0.785} & \textbf{0.789} & \textbf{0.850} & \textbf{0.872} & \textbf{0.667} \\
\bottomrule
\end{tabular}
\caption{Complete aggregate main-suite means. CQ is trajectory iCORE quality; Cond. Perf. is the calibrated conditional expectation used in Proposition~\ref{prop:quality-performance}; Term. Perf. is realized performance on required roots. Bold and underlined entries denote the best and second-best distinct means in each column.}
\label{tab:core-main-results-full-revised}
\end{table*}

%% file: results/tables_core_revised/measurement_results_revised.tex
\begin{table*}[t]
\centering
\small
\setlength{\tabcolsep}{5.2pt}
\begin{tabular}{@{}lccc|ccc@{}}
\toprule
& \multicolumn{3}{c|}{Controlled} & \multicolumn{3}{c}{Real LLM} \\
\cmidrule(lr){2-4}\cmidrule(lr){5-7}
System & Sound. F1 & Reassign. F1 & $|\widehat{\mathrm{CQ}}-\mathrm{CQ}|$ & Sound. F1 & Reassign. F1 & $|\widehat{\mathrm{CQ}}-\mathrm{CQ}|$ \\
\midrule
MAS-Only & 0.000 & 0.000 & 0.291 & 0.000 & 0.000 & 0.412 \\
Interaction-only & 0.810 & 0.000 & 0.234 & 0.952 & 0.000 & 0.337 \\
Task-only & 0.000 & 1.000 & 0.071 & 0.000 & 1.000 & 0.052 \\
LLM-Judge & 0.592 & 0.047 & 0.196 & 0.640 & 0.031 & 0.298 \\
iCORE-Observe & 1.000 & 1.000 & 0.000 & 1.000 & 1.000 & 0.000 \\
iCORE-Audit & 1.000 & 1.000 & 0.000 & 1.000 & 1.000 & 0.000 \\
\bottomrule
\end{tabular}
\caption{Proposal-probe reconstruction results across controlled and real-LLM execution.}
\label{tab:core-measurement-revised}
\end{table*}

%% file: results/tables_core_revised/ablation_results_revised.tex
\begin{table}[t]
\centering
\small
\setlength{\tabcolsep}{3.8pt}
\begin{tabular}{@{}lcccc@{}}
\toprule
Variant & Sound. & Assign. & CQ & Term. Perf. \\
\midrule
Full & 1.000 & 0.637 & 0.818 & 1.000 \\
Sound.-only & 1.000 & 0.350 & 0.675 & 0.944 \\
Assign.-only & 0.840 & 0.621 & 0.730 & 0.694 \\
No $\Pi$ & 0.840 & 0.621 & 0.730 & 0.694 \\
No cert. val. & 0.847 & 0.622 & 0.735 & 0.657 \\
\bottomrule
\end{tabular}
\caption{Controlled mixed-condition ablations over six environments (90 episodes total).}
\label{tab:core-ablation-revised}
\end{table}

%% file: results/tables_core_revised/boundary_results_revised.tex
\begin{table}[t]
\centering
\small
\setlength{\tabcolsep}{3.8pt}
\begin{tabular}{@{}lcccc@{}}
\toprule
Boundary & CQ & Cond. Perf. & Term. Perf. & Completed \\
\midrule
Standard & 0.818 & 0.865 & 1.000 & 1.000 \\
Degraded val. & 0.812 & 0.865 & 0.704 & 0.944 \\
Shuffled cap. & 0.825 & 0.723 & 0.833 & 0.778 \\
Weak pool & 0.807 & 0.442 & 0.519 & 0.333 \\
\bottomrule
\end{tabular}
\caption{iCORE-Audit in the controlled theory-boundary suite (144 episodes across iCORE-Observe and iCORE-Audit).}
\label{tab:core-boundary-revised}
\end{table}

%% file: results/tables_core_revised/runtime_results_revised.tex
\begin{table*}[t]
\centering
\scriptsize
\setlength{\tabcolsep}{4.0pt}
\begin{tabular}{@{}llrrrrrr@{}}
\toprule
Mode & System & Wall (s) & Model calls & Prompt tok. & Completion tok. & Verifier calls & Certificates \\
\midrule
Controlled & MAS-Only & 0.016 & 16.30 & 0.0 & 0.0 & 4.26 & 8.81 \\
 & Interaction-only & 0.018 & 16.62 & 0.0 & 0.0 & 4.52 & 9.16 \\
 & Task-only & 0.017 & 15.96 & 0.0 & 0.0 & 4.39 & 9.02 \\
 & LLM-Judge & 0.020 & 16.52 & 0.0 & 0.0 & 4.71 & 9.37 \\
 & iCORE-Observe & 0.018 & 16.30 & 0.0 & 0.0 & 4.26 & 8.81 \\
 & iCORE-Audit & 0.026 & 17.47 & 0.0 & 0.0 & 3.59 & 7.87 \\
\midrule
Real LLM & MAS-Only & 2.608 & 14.67 & 5776.4 & 764.8 & 8.22 & 12.11 \\
 & Interaction-only & 2.763 & 14.72 & 5799.9 & 762.6 & 8.28 & 12.17 \\
 & Task-only & 2.658 & 15.28 & 5981.6 & 731.6 & 7.22 & 11.33 \\
 & LLM-Judge & 3.010 & 14.78 & 5858.3 & 762.4 & 8.22 & 12.11 \\
 & iCORE-Observe & 2.672 & 14.67 & 5776.4 & 758.9 & 8.22 & 12.11 \\
 & iCORE-Audit & 3.131 & 18.17 & 6543.3 & 877.2 & 4.33 & 9.67 \\
\bottomrule
\end{tabular}
\caption{Per-episode execution-cost means. Controlled token counts are zero because the shared MLP does not tokenize prompts. Wall time is implementation- and hardware-specific and is not used as a quality metric.}
\label{tab:core-runtime-revised}
\end{table*}

%% file: results/tables_core_revised/statistical_tests_revised.tex
\begin{table}[t]
\centering
\small
\setlength{\tabcolsep}{4.0pt}
\begin{tabular}{@{}llrrr@{}}
\toprule
Mode & Metric & Pairs & Min. gain & Max. $p_{\mathrm H}$ \\
\midrule
Controlled & CQ & 90 & 0.060 & 1.13e-09 \\
Controlled & Cond. Perf. & 90 & 0.193 & 2.73e-07 \\
Controlled & Term. Perf. & 90 & 0.118 & 8.71e-03 \\
Real LLM & CQ & 18 & 0.168 & 7.30e-03 \\
Real LLM & Cond. Perf. & 18 & 0.361 & 2.30e-03 \\
Real LLM & Term. Perf. & 18 & 0.171 & 3.41e-02 \\
\bottomrule
\end{tabular}
\caption{One-sided paired Wilcoxon tests for iCORE-Audit against each of the five baselines. Holm correction is applied jointly over 15 comparisons within each execution mode. The table reports the least favorable adjusted result.}
\label{tab:core-stats-revised}
\end{table}

%% file: ref.bib
@article{wang2024survey,
  title={A survey on large language model based autonomous agents},
  author={Wang, Lei and Ma, Chen and Feng, Xueyang and Zhang, Zeyu and Yang, Hao and Zhang, Jingsen and Chen, Zhiyuan and Tang, Jiakai and Chen, Xu and Lin, Yankai and others},
  journal={Frontiers of Computer Science},
  volume={18},
  number={6},
  pages={186345},
  year={2024},
  publisher={Springer}
}

@article{guo2024large,
  title={Large language model based multi-agents: A survey of progress and challenges},
  author={Guo, Taicheng and Chen, Xiuying and Wang, Yaqi and Chang, Ruidi and Pei, Shichao and Chawla, Nitesh V and Wiest, Olaf and Zhang, Xiangliang},
  journal={arXiv preprint arXiv:2402.01680},
  year={2024}
}

@article{wu2023autogen,
  title={Autogen: Enabling next-gen llm applications via multi-agent conversation},
  author={Wu, Qingyun and Bansal, Gagan and Zhang, Jieyu and Wu, Yiran and Li, Beibin and Zhu, Erkang and Jiang, Li and Zhang, Xiaoyun and Zhang, Shaokun and Liu, Jiale and others},
  journal={arXiv preprint arXiv:2308.08155},
  year={2023}
}

@article{li2023camel,
  title={Camel: Communicative agents for" mind" exploration of large language model society},
  author={Li, Guohao and Hammoud, Hasan and Itani, Hani and Khizbullin, Dmitrii and Ghanem, Bernard},
  journal={Advances in neural information processing systems},
  volume={36},
  pages={51991--52008},
  year={2023}
}

@inproceedings{hong2024metagpt,
  title={MetaGPT: Meta programming for a multi-agent collaborative framework},
  author={Hong, Sirui and Zhuge, Mingchen and Chen, Jonathan and Zheng, Xiawu and Cheng, Yuheng and Wang, Jinlin and Zhang, Ceyao and Yau, Steven and Lin, Zijuan and Zhou, Liyang and others},
  booktitle={International Conference on Learning Representations},
  volume={2024},
  pages={23247--23275},
  year={2024}
}

@inproceedings{qian2024chatdev,
  title={Chatdev: Communicative agents for software development},
  author={Qian, Chen and Liu, Wei and Liu, Hongzhang and Chen, Nuo and Dang, Yufan and Li, Jiahao and Yang, Cheng and Chen, Weize and Su, Yusheng and Cong, Xin and others},
  booktitle={Proceedings of the 62nd annual meeting of the association for computational linguistics (volume 1: Long papers)},
  pages={15174--15186},
  year={2024}
}

@inproceedings{park2023generative,
  title={Generative agents: Interactive simulacra of human behavior},
  author={Park, Joon Sung and O'Brien, Joseph and Cai, Carrie Jun and Morris, Meredith Ringel and Liang, Percy and Bernstein, Michael S},
  booktitle={Proceedings of the 36th annual acm symposium on user interface software and technology},
  pages={1--22},
  year={2023}
}

@article{yao2022react,
  title={React: Synergizing reasoning and acting in language models},
  author={Yao, Shunyu and Zhao, Jeffrey and Yu, Dian and Du, Nan and Shafran, Izhak and Narasimhan, Karthik and Cao, Yuan},
  journal={arXiv preprint arXiv:2210.03629},
  year={2022}
}

@article{schick2023toolformer,
  title={Toolformer: Language models can teach themselves to use tools},
  author={Schick, Timo and Dwivedi-Yu, Jane and Dess{\`\i}, Roberto and Raileanu, Roberta and Lomeli, Maria and Hambro, Eric and Zettlemoyer, Luke and Cancedda, Nicola and Scialom, Thomas},
  journal={Advances in neural information processing systems},
  volume={36},
  pages={68539--68551},
  year={2023}
}

@article{shen2023hugginggpt,
  title={Hugginggpt: Solving ai tasks with chatgpt and its friends in hugging face},
  author={Shen, Yongliang and Song, Kaitao and Tan, Xu and Li, Dongsheng and Lu, Weiming and Zhuang, Yueting},
  journal={Advances in Neural Information Processing Systems},
  volume={36},
  pages={38154--38180},
  year={2023}
}

@inproceedings{wang2023plan,
  title={Plan-and-solve prompting: Improving zero-shot chain-of-thought reasoning by large language models},
  author={Wang, Lei and Xu, Wanyu and Lan, Yihuai and Hu, Zhiqiang and Lan, Yunshi and Lee, Roy Ka-Wei and Lim, Ee-Peng},
  booktitle={Proceedings of the 61st annual meeting of the association for computational linguistics (volume 1: long papers)},
  pages={2609--2634},
  year={2023}
}

@article{yao2023tree,
  title={Tree of thoughts: Deliberate problem solving with large language models},
  author={Yao, Shunyu and Yu, Dian and Zhao, Jeffrey and Shafran, Izhak and Griffiths, Tom and Cao, Yuan and Narasimhan, Karthik},
  journal={Advances in neural information processing systems},
  volume={36},
  pages={11809--11822},
  year={2023}
}

@article{madaan2023self,
  title={Self-refine: Iterative refinement with self-feedback},
  author={Madaan, Aman and Tandon, Niket and Gupta, Prakhar and Hallinan, Skyler and Gao, Luyu and Wiegreffe, Sarah and Alon, Uri and Dziri, Nouha and Prabhumoye, Shrimai and Yang, Yiming and others},
  journal={Advances in neural information processing systems},
  volume={36},
  pages={46534--46594},
  year={2023}
}

@article{shinn2023reflexion,
  title={Reflexion: Language agents with verbal reinforcement learning},
  author={Shinn, Noah and Cassano, Federico and Gopinath, Ashwin and Narasimhan, Karthik and Yao, Shunyu},
  journal={Advances in neural information processing systems},
  volume={36},
  pages={8634--8652},
  year={2023}
}

@article{zheng2023judging,
  title={Judging llm-as-a-judge with mt-bench and chatbot arena},
  author={Zheng, Lianmin and Chiang, Wei-Lin and Sheng, Ying and Zhuang, Siyuan and Wu, Zhanghao and Zhuang, Yonghao and Lin, Zi and Li, Zhuohan and Li, Dacheng and Xing, Eric and others},
  journal={Advances in neural information processing systems},
  volume={36},
  pages={46595--46623},
  year={2023}
}

@incollection{lamport2019time,
  title={Time, clocks, and the ordering of events in a distributed system},
  author={Lamport, Leslie},
  booktitle={Concurrency: the Works of Leslie Lamport},
  pages={179--196},
  year={2019}
}

@incollection{smith1988contract,
  title={The contract net protocol: High-level communication and control in a distributed problem solver},
  author={Smith, Reid G},
  booktitle={Readings in distributed artificial intelligence},
  pages={357--366},
  year={1988},
  publisher={Elsevier}
}

@article{yu2026interactive,
  title={Interactive Critique-Revision Training for Reliable Structured LLM Generation},
  author={Yu, Fei Xu and Zhang, Zuyuan and Imani, Mahdi and Bastian, Nathaniel D and Lan, Tian},
  journal={arXiv preprint arXiv:2605.08327},
  year={2026}
}

@incollection{li2012agent,
  title={Agent communication language},
  author={Li, Shujun and Kokar, Mieczyslaw M},
  booktitle={Flexible Adaptation in Cognitive Radios},
  pages={37--44},
  year={2012},
  publisher={Springer}
}

@techreport{nii1986blackboard,
  title={Blackboard Systems.},
  author={Nii, H Penny},
  year={1986}
}

@inproceedings{erol1994htn,
  title={HTN planning: Complexity and expressivity},
  author={Erol, Kutluhan and Hendler, James and Nau, Dana S},
  booktitle={AAAI},
  volume={94},
  pages={1123--1128},
  year={1994}
}

@article{torreno2017cooperative,
  title={Cooperative multi-agent planning: A survey},
  author={Torreno, Alejandro and Onaindia, Eva and Komenda, Anton{\'\i}n and {\v{S}}tolba, Michal},
  journal={ACM Computing Surveys (CSUR)},
  volume={50},
  number={6},
  pages={1--32},
  year={2017},
  publisher={ACM New York, NY, USA}
}

@article{van2003workflow,
  title={Workflow patterns},
  author={van Der Aalst, Wil MP and Ter Hofstede, Arthur HM and Kiepuszewski, Bartek and Barros, Alistair P},
  journal={Distributed and parallel databases},
  volume={14},
  number={1},
  pages={5--51},
  year={2003},
  publisher={Springer}
}

@article{van1998application,
  title={The application of Petri nets to workflow management},
  author={Van der Aalst, Wil MP},
  journal={Journal of circuits, systems, and computers},
  volume={8},
  number={01},
  pages={21--66},
  year={1998},
  publisher={World Scientific}
}

@article{murata1989petri,
  title={Petri nets: Properties, analysis and applications},
  author={Murata, Tadao},
  journal={Proceedings of the IEEE},
  volume={77},
  number={4},
  pages={541--580},
  year={1989},
  publisher={IEEE}
}

@article{leucker2009brief,
  title={A brief account of runtime verification},
  author={Leucker, Martin and Schallhart, Christian},
  journal={The journal of logic and algebraic programming},
  volume={78},
  number={5},
  pages={293--303},
  year={2009},
  publisher={Elsevier}
}

@article{havelund2001monitoring,
  title={Monitoring java programs with java pathexplorer},
  author={Havelund, Klaus and Ro{\c{s}}u, Grigore},
  journal={Electronic Notes in Theoretical Computer Science},
  volume={55},
  number={2},
  pages={200--217},
  year={2001},
  publisher={Elsevier}
}

@article{belhajjame2013prov,
  title={Prov-dm: The prov data model},
  author={Belhajjame, Khalid and B’Far, Reza and Cheney, James and Coppens, Sam and Cresswell, Stephen and Gil, Yolanda and Groth, Paul and Klyne, Graham and Lebo, Timothy and McCusker, Jim and others},
  journal={W3C Recommendation},
  volume={14},
  pages={15--16},
  year={2013},
  publisher={World Wide Web Consortium (W3C)}
}

@article{groth2013prov,
  title={PROV-overview. An overview of the PROV family of documents},
  author={Groth, Paul and Moreau, Luc},
  year={2013},
  publisher={World Wide Web Consortium}
}

@book{baader1998term,
  title={Term rewriting and all that},
  author={Baader, Franz and Nipkow, Tobias},
  year={1998},
  publisher={Cambridge university press}
}

@article{newman1942theories,
  title={On theories with a combinatorial definition of" equivalence"},
  author={Newman, Maxwell Herman Alexander},
  journal={Annals of mathematics},
  volume={43},
  number={2},
  pages={223--243},
  year={1942},
  publisher={JSTOR}
}

@inproceedings{zhang2025learning,
  title={Learning to collaborate with unknown agents in the absence of reward},
  author={Zhang, Zuyuan and Zhou, Hanhan and Imani, Mahdi and Lee, Taeyoung and Lan, Tian},
  booktitle={Proceedings of the AAAI Conference on Artificial Intelligence},
  volume={39},
  number={13},
  pages={14502--14511},
  year={2025}
}

@inproceedings{brauer1986petri,
  title={Petri Nets: Applications and Relationships to Other Models of Concurrency: Advances in Petri Nets 1986, Part II Proceedings of an Advanced Course Bad Honnef, 8.--19. September 1986},
  author={Brauer, Wilfried and Reisig, Wolfgang and Rozenberg, Grzegorz},
  booktitle={Advanced Course on Petri Nets},
  year={1986},
  organization={Springer}
}

@article{yang2026dig,
  title={DIG to Heal: Scaling General-purpose Agent Collaboration via Explainable Dynamic Decision Paths},
  author={Yang, Hanqing and Lee, Hyungwoo and Yao, Yuhang and Liu, Zhiwei and Liu, Kay and Chen, Jingdi and Joe-Wong, Carlee},
  journal={arXiv preprint arXiv:2603.00309},
  year={2026}
}
